\documentclass[thmsa,letterpaper,12pt]{article}
\newcommand{\blind}{0}

\usepackage[utf8]{inputenc}
\usepackage{geometry,setspace}
\def\spacingset#1{\renewcommand{\baselinestretch}%
	{#1}\small\normalsize} \spacingset{1}
\spacingset{1.4}

\usepackage{pdfpages}
\usepackage[title]{appendix}
\usepackage{soul}

\usepackage{booktabs,caption,widetable,threeparttable}
\usepackage{tikz,fullpage}
\usetikzlibrary{arrows,	petri, topaths}
\usepackage{rotating,graphicx,lscape,xr-hyper,pgf,float,color}
\usepackage[position=top]{subfig}

\usepackage[authoryear]{natbib}
\usepackage{bibentry}
\usepackage[colorlinks = true, allcolors = blue]{hyperref}
\usepackage{listings}
\usepackage[reftex]{theoremref}
\usepackage{amsmath, amsthm, amssymb, bm, mathtools, amsfonts, subdepth}
\mathtoolsset{showonlyrefs}

\theoremstyle{remark}
\newtheorem{rem}{Remark}

\theoremstyle{definition}

\newtheorem{definition}{Definition}

\theoremstyle{plain}

\newtheorem{assumption}{Assumption}

\newtheorem{innercustomass}{Assumption}
\newenvironment{customass}[1]
{\renewcommand\theinnercustomass{#1}\innercustomass}
{\endinnercustomass}

\newtheorem{theorem}{Theorem}[section]

\usepackage{fixme}

\DeclarePairedDelimiter{\abs}{\lvert}{\rvert}
\DeclarePairedDelimiter{\norm}{\lVert}{\rVert}

\newcommand\inverse{^{-1}}
\newcommand{\R}{\mathbb{R}}

\newcommand{\E}{\mathbb{E}}
\newcommand{\V}{\mathbb{V}}
\renewcommand{\Pr}{\mathbb{P}}

\makeatletter
\newcommand*{\addFileDependency}[1]{
\typeout{(#1)}
\@addtofilelist{#1}
\IfFileExists{#1}{}{\typeout{No file #1.}}
}\makeatother

\begin{document}


\date{January 2023}

\if0\blind
{
	\title{Testing Many Restrictions Under Heteroskedasticity%
	\thanks{We thank the Editor, Associate Editors, and a few diligent referees for suggestions that greatly improved the presentation. We are grateful to Bruce Hansen, Michael Jansson, Alexei Onatski, Jack Porter, and Tiemen Woutersen for useful discussions. We also appreciate helpful comments from seminar participants at UCL, Georgetown University, Erasmus University Rotterdam, and University of Arizona, as well as conference audiences at the 2022 Cowles Conference on Econometrics, 2021 Econometric Society Winter Meeting, ASSET 2021, MSA 2021, NordStat 2020, and LinStat 2020.}	
	}
	\author{\textsc{Stanislav Anatolyev}\thanks{CERGE-EI, Politick\'{y}ch v\v{e}z\v{n}\r{u} 7,
			11121 Prague 1, Czech Republic. E-mail: stanislav.anatolyev@cerge-ei.cz.
			This research was supported by the grant 20-28055S from the Czech Science
			Foundation.}, \
\textsc{Mikkel S\o lvsten}\thanks{Department of Economics and Business Economics, Aarhus University, Fuglsangs All\'{e} 4, DK-8210 Aarhus V, Denmark. E-mail: miso@econ.au.dk.}
}
} \fi

\if1\blind
{
	\title{Testing Many Restrictions Under Heteroskedasticity
		\thanks{ We thank Bruce Hansen, Michael Jansson, and Jack Porter along with seminar participants at UCL, Georgetown, Erasmus University Rotterdam the 2021 Econometric Society Winter Meetings, for helpful comments.}
	}
	\author{}
} \fi

\maketitle

\begin{abstract}
\noindent
We propose a hypothesis test that allows for many tested restrictions in a heteroskedastic linear regression model.
The test compares the conventional F statistic to a critical value that corrects for many restrictions and conditional heteroskedasticity.
This correction uses leave-one-out estimation to correctly center the critical value and leave-three-out estimation to appropriately scale it.
The large sample properties of the test are established in an asymptotic framework where the number of tested restrictions may be fixed or may grow with the sample size, and can even be proportional to the number of observations.
We show that the test is asymptotically valid and has non-trivial asymptotic power against the same local alternatives as the exact F test when the latter is valid.
Simulations corroborate these theoretical findings and suggest excellent size control in moderately small samples, even under strong heteroskedasticity.

\bigskip

\medskip

\end{abstract}

\thispagestyle{empty}\newpage

\section{Introduction}

One of the central tenets in modern economic research is to consider models that allow for flexible specifications of heterogeneity and to establish whether {meaningful heterogeneity is present or absent in a particular empirical setting}. For example, \cite{abowd1999high} study whether there is firm-specific heterogeneity in a linear model for individual log-wages, \cite{card2016bargaining,card2016firms} ask if this heterogeneity varies by the individual's {gender or} education, and \cite{lachowska2019firm} investigate whether the firm-specific heterogeneity is constant over time. Other work relies on similarly flexible models to investigate the presence of heterogeneity in health economics \citep{finkelstein2016sources} and to study neighborhood effects \citep{chetty2018impacts}. In all these examples, the absence of a particular dimension of heterogeneity corresponds to a hypothesis that imposes hundreds or thousands of restrictions on the model of interest. The present paper provides a tool to conduct a test of such hypotheses. {In contemporary work, \cite{kline2022systemic} apply our proposed test in a study of discrimination among U.S. employers.}

We develop a test for hypotheses that impose multiple restrictions and establish its asymptotic validity in a heteroskedastic linear regression model where the number of tested restrictions may be fixed or increasing with the sample size. In particular, we allow for the number of restrictions and the sample size to be proportional. The exact F test{, which compares the F statistic to a quantile of the F distribution,} fails to control size in this environment. Instead, our proposed test  rejects the null hypothesis if the F statistic exceeds a critical value that corrects for many restrictions and conditional heteroskedasticity. This critical value is a recentered and rescaled quantile of what is naturally called the {\it F-bar distribution} as it describes the distribution of a chi-bar-squared random variable divided by an independent chi-squared random variable over its degrees of freedom.\footnote{{Chi-bar-squared is a standard name used to describe a mixture of chi-squared distributions. See, e.g., \cite{chibar2} who studies asymptotic properties of chi-bar-squared distributions.}} This family of distributions can approximate both the finite sample properties of the F statistic under homoskedastic normal errors and---after recentering and rescaling---the asymptotic distribution of the F statistic in the presence of conditional heteroskedasticity and few or many restrictions.

The large sample validity of our proposed test holds uniformly in the number of regressors and tested restrictions. In combination with the F-bar distribution, the key to this uniformity is our proposed location and variance estimators that are used to recenter and rescale the critical value. {The location estimator utilizes unbiased leave-one-out estimators for individual error variances, while the variance estimator utilizes unbiased leave-\textit{three}-out estimators for \textit{products} of these variances. While the product of leave-one-out estimators is biased for the product of variances because of mutual dependence, dropping three observations in successive fashion breaks the dependence between estimators of individual error variances, and hence provides unbiasedness of their products.}
The use of leave-three-out estimation to implement this idea is novel in the literature. Because the essential elements of the test are built on \textit{l}eave-\textit{o}ut machinery, we will at times and for brevity refer to the proposed test using the acronym LO.

The LO test has exact asymptotic size when the regression design has full rank after leaving any combination of three observations out of the sample. This condition is satisfied in models with many continuous regressors and only a few discrete ones. However, the condition can fail when many discretely valued regressors are included, as occurs for models with {fixed individual or group effects.} { With group effects, {in particular,} leave-three-out {may not exist} when group sizes are two or three.} To handle such cases, the proposed test uses estimators for the products of individual error variances that are intentionally biased upward when the unbiased leave-three-out estimators do not exist. This construction ensures {large sample validity} but can potentially {lead to a slightly conservative test} when {a large fraction} of the leave-three-out estimators do not exist.

Using both theoretical arguments and simulations, \cite{huber1973robust} and \cite{berndt1977conflict} have highlighted the importance of allowing the number of regressors and potentially the number of tested restrictions to increase with sample size when studying asymptotic properties of inference procedures. {The latter paper specifically documents cases where asymptotically equivalent classical tests yield opposite outcomes when the number of tested restrictions is somewhat large.} Despite these early cautionary tales, most inference procedures that allow for proportionality between the number of regressors, sample size, and potentially the number of restrictions, are of a more recent vintage. Here, we survey the ones most relevant to the current paper and refer to \cite{anatolyev2019survey} for a more extensive review of the literature.\looseness=-1\footnote{In analysis of variance contexts, which are special cases of linear regression, \cite{akritas2004anova} and \cite{zhou2017hidim} propose heteroskedasticity robust tests for equality of means that are, however, specific to their models. An expanding literature considers (outlier) robust estimation of linear high-dimensional regressions \citep[e.g., ][]{elkaroui2013hidim} but does not provide valid tests of many restrictions.}

In homoskedastic regression models, \cite{anatolyev2012inference} and \cite{calhoun2011many} propose various corrections to classical tests that restore asymptotic validity in the presence of many restrictions. In heteroskedastic regressions with one tested restriction and many regressors, \cite{cattaneo2017inference} show that the use of conventional Eicker-White standard errors and their ``almost-unbiased'' variations \citep[see][]{mackinnon2012hetero} does not yield asymptotic validity. This failure may be viewed as a manifestation of the incidental parameters problem. To overcome this problem, \cite{cattaneo2017inference} and subsequently \cite{anatolyev2018almostunbiased} propose new versions of the Eicker-White standard errors, which restore size control in large samples. However, these proposals rely on the inversion of $n$-by-$n$ matrices ($n$ denotes sample size) that may fail to be invertible in examples of practical interest \citep{horn1975estimating,verdier2016estimation}. \cite{rao1970heterovar}'s unbiased estimator for individual error variances is closely related to \cite{cattaneo2017inference}'s proposal and suffers from the same existence issue.\footnote{{{A recent use of \cite{rao1970heterovar}'s MINQUE estimator is \cite{juhl2014hetero}, where it is incorporated into a test of heterogeneity in short panels with fixed effects. Here, the MINQUE estimator is applied to each cross-sectional unit, and its use therefore imposes a growing number of invertibility requirements.}}}

In homoskedastic regression models, Anatolyev (2012) and Calhoun (2011) propose various corrections to classical tests that restore asymptotic validity in the presence of many restrictions. 
In heteroskedastic regressions with one tested restriction and many regressors, Cattaneo et al. (2018b) show that the use of conventional Eicker-White standard errors and their “almost-unbiased” variations does not yield asymptotic validity. This failure may be viewed as a manifestation of the incidental parameters problem. To overcome this problem, Cattaneo et al. (2018b) and subsequently Anatolyev (2018) propose new versions of the Eicker-White standard errors, which restore size control in large samples. However, these proposals rely on the inversion of n-by-n matrices (n denotes sample size) that may fail to be invertible in examples of practical interest. Rao (1970)’s unbiased estimator for individual error variances is closely related to Cattaneo et al. (2018b)’s proposal and suffers from the same existence issue.

\cite*{kline2018leave} propose instead a version of the Eicker-White standard errors that relies only on leave-one-out estimators of individual error variances and show that its use leads to asymptotic size control when testing a single restriction.\footnote{\cite{jochmans2018variance} additionally uses simulations to investigate the behavior of this variance estimator.} While this conclusion extends to hypotheses that involve a fixed and small number of restrictions through the use of a heteroskedasticity-robust Wald test, it fails to hold in cases of many restrictions. When testing many coefficients equal to zero, \cite{kline2018leave} note that those leave-one-out individual variance estimators can be used to center the conventional F statistic\footnote{The use of leave-one-out estimation has a long tradition in the literature on instrumental variables \citep[see, e.g.,][]{phillips1977bias}, and our test shares an algebraic representation with the adjusted J test analyzed in \cite{chao2014testing} \citep[see][for a discussion]{kline2018leave}. An attractive feature of relying on leave-one-out is that challenging estimation of higher order error moments can be avoided, which is in contrast to the tests of \cite{calhoun2011many} and \cite{anatolyev2013manyexogenous}.} and propose a rescaling of the statistic that relies on successive sample splitting {\citep[][section 5]{kline2018leave} as a tool of breaking dependence among different estimates when error variances enter as pairwise products. However, first, sample splitting places restrictions on the data that will fail when the number of regressors is larger than half of the sample size.\footnote{This phenomenon is akin to the situation in \cite{cattaneo2017inference}, in which the `Hadamard square' of the orthogonal projection matrix may be non-invertible. \cite{cattaneo2017inference} rule out this possibility by imposing a sufficient condition such that the number of covariates is no larger than half of the sample size.} Second, sample splitting means that the error variances are estimated only from a part of the sample, which is clearly inefficient and undesirable for conditional objects that require the use of as many observations as possible.
Third, sample splitting may be undesirable because different ways of splitting the sample can lead to opposite conclusions.

We instead take another route and utilize information for estimation of error variances in the whole sample. In order to remove the dependence among individual estimates, we appeal to leave-three-out estimation instead of sample splitting. Importantly, leave-three-out estimation places much fewer restrictions on the number of regressors than sample splitting, exploits available sample information more efficiently, and does not require a researcher to choose a way to split the sample. Additionally, the robustified version of the LO test using the F-bar distribution enables asymptotic size control uniformly in the number of restrictions.}

We provide a theoretical study of the power properties under local and global alternatives.
Under local alternatives, the asymptotic power curve of the proposed LO test is parallel to that of the exact F test when the latter is valid, e.g., under homoskedastic normal errors. While the curves are parallel, the LO test tends to have power somewhat below the exact F test. This loss in power stems from the estimation of individual error variances and can be viewed as a cost of using a test that is robust to general heteroskedasticity. This cost is largely monotone in the number of tested restrictions and disappears when the number of restrictions is small relative to sample size.

We also conduct a simulation study that documents excellent performance of the LO test in small and moderately sized samples. We document that the LO test delivers nearly exact size control in samples as small as $100$ observations in both homoskedastic and heteroskedastic environments. On the other hand, conventional tools such as the Wald test and the exact F test can exhibit severe size distortions and reject a true null with near certainty for some configurations. These findings are documented using two simulation settings: one with continuous regressors only, and one with a mix of both continuous and discrete regressors. In the latter setting, roughly $7\%$ of observations cause a full rank failure when leaving up to three observations out, but the proposed test shows almost no conservatism even in this adverse environment. When both the LO and exact F tests are valid, the simulations document a power loss that varies between being negligible and up to roughly $15$ percentage points, depending on the type of deviation from the null and sample size. For many applications, this range of power losses is a small cost to incur for being robust to heteroskedasticity.

The paper is organized as follows. Section \ref{sec:model} introduces the setup and the proposed critical value in samples where all the leave-three-out estimators exist, while Section \ref{sec:asymptotia} analyzes the asymptotic size and power of the LO test for such samples. Section \ref{sec:cons} describes the critical value for use in samples where the design loses full rank after leaving certain triples of observations out. Section \ref{sec:sim} discusses the results of simulation experiments, and Section \ref{sec:conc} concludes. Proofs of theoretical results and some clarifying but technical details are collected in the online supplemental Appendix. An R package \citep{AScode} that implements the proposed test is available online.

\section{Leave-out test}\label{sec:model}

Consider a linear regression model
\begin{align}\label{eq:Model}
	y_{i} = \boldsymbol x_{i}'{\boldsymbol \beta} + \varepsilon_{i}, \quad \E[\varepsilon_i | \boldsymbol x_i] =0,
\end{align}
where an intercept is included in the regression function $\boldsymbol x_i'{\boldsymbol \beta}$ and the $n$ observed random vectors $\{(y_i,\boldsymbol x_i')'\}_{i=1}^n$ are independent across $i$. The dimension of the regressors $\boldsymbol x_i \in \R^m$ may be large relative to sample size {with $m < n$}, and there is conditional heteroskedasticity in the unobserved errors:
{
\begin{align}\label{eq:heterr}
	\E[\varepsilon_{i}^2 | \boldsymbol x_i]= \sigma^{2}(\boldsymbol x_i) \equiv \sigma_i^{2}.
\end{align}
The conditional variances are assumed to exist with no restrictions placed on the functional form, as in \cite{kline2018leave}.}

The hypothesis of interest involves $r \leq m$ linear restrictions
\begin{align}\label{eq:NullHypo}
	H_0 : \ \boldsymbol R{\boldsymbol \beta} =\boldsymbol q,
\end{align}
where the matrix $\boldsymbol R \in \R^{r\times m}$ has full row rank $r$, and $\boldsymbol q \in \R^r$. Both $\boldsymbol R$ and $\boldsymbol q$ are specified by the researcher. Specifically, they are assumed to be known and are allowed to depend on the observed regressors. The space of alternatives is $H_A : \ \boldsymbol R{\boldsymbol \beta} \neq \boldsymbol q$.

The attention of the paper is on settings where the design matrix $\boldsymbol S_{xx} = \sum_{i=1}^n \boldsymbol x_i \boldsymbol x_i'$ has full rank so that $\hat{\boldsymbol \beta}=\boldsymbol S_{xx}\inverse \sum_{i=1}^n \boldsymbol x_i y_i$, the ordinary least squares (OLS) estimator of $\boldsymbol \beta$, is defined. For compact reference, we define the degrees-of-freedom adjusted residual variance
\begin{align}\label{eq:residvar}
	\hat \sigma^2_{\varepsilon}=\frac{1}{n-m}\sum_{i=1}^n \big(y_i-\boldsymbol x_{i}'\hat{\boldsymbol \beta}\big)^2.
\end{align}

\begin{rem}
	{
	We maintain in this paper that observations are independent across $i$, as it facilitates simplicity when discussing some of our high-level conditions. {We conjecture that the results of the paper} continue to hold under the weaker assumption that the error terms are mean zero and independent across $i$ when conditioning on all the regressors $\{\boldsymbol{x}_i\}_{i=1}^n$. 
	}
\end{rem}

\subsection{Test statistic}

Our proposed test rejects $H_0$ for large values of Fisher's F statistic,
\begin{align}\label{eq:F}
F=\frac{\big( \boldsymbol R\hat {\boldsymbol\beta}- \boldsymbol q \big)'\!\left( \boldsymbol R \boldsymbol S_{xx}\inverse \boldsymbol R' \right)\inverse \!\big(\boldsymbol R\hat {\boldsymbol \beta} - \boldsymbol q \big)}{r\hat \sigma^2_{\varepsilon}},
\end{align}
which is a monotone transformation of the likelihood ratio statistic when the  regression errors are homoskedastic normal. Since we do not impose normality, $F$ may be viewed as a quasi likelihood ratio statistic. {The behavior of the test is governed by the numerator of $F$, which we denote by ${\cal F}$:
\begin{align}\label{eq:calF}
{\cal F}=\big( \boldsymbol R\hat {\boldsymbol\beta}- \boldsymbol q \big)'\!\left( \boldsymbol R \boldsymbol S_{xx}\inverse \boldsymbol R' \right)\inverse \!\big(\boldsymbol R\hat {\boldsymbol \beta} - \boldsymbol q \big).
\end{align}

By taking this statistic as a point of departure, we are able to construct a critical value that ensures size control in the presence of heteroskedasticity and an arbitrary number of restrictions.\footnote{An alternative approach might have taken a heteroskedasticity-robust Wald statistic $\boldsymbol R \boldsymbol S_{xx}\inverse (\sum_{i=1}^n \boldsymbol x_i \boldsymbol x_i' \hat\varepsilon_i^2) \boldsymbol S_{xx}\inverse \boldsymbol R'$, where $\{\hat \varepsilon_{i}\}_{i=1}^n$ are OLS residuals, in a similar attempt to ensure validity when the number of restrictions is proportional to the sample size. However, in such environments, any heteroskedasticity-robust Wald statistic relies on the inverse of a high-dimensional covariance matrix estimator, a feature that presents substantial challenges when attempting to control size. 
Specifically, the randomness in the residuals induced into this $r \times r$-matrix persists in large samples and is therefore a threat to valid inference. Our conjecture is that some regularization of the covariance matrix may be helpful in mitigating the noise arising from this estimated covariance matrix.	In addition, it is not clear if the weighting behind the heteroskedasticity-robust Wald statistic and hence a potential test will preserve optimality under asymptotics where the number of restrictions is proportional to the sample size. We leave investigation of these difficult but interesting questions to future research.}
The proposed critical value yields} asymptotic validity under two asymptotic frameworks, one where the number of restrictions is fixed, and one where the number of restrictions may grow as fast as proportionally to the sample size. To achieve such uniformity with respect to the number of restrictions, we rely on an auxiliary distribution, the F-bar distribution, that helps unite these two frameworks.

\subsection{F-bar distribution}

Our test rejects $H_0$ if Fisher's F exceeds a linearly transformed quantile of a distribution, which we call the \textit{F-bar distribution}.
We define this family of distributions and discuss its role before we turn to a description of the linear transformation mentioned above.
\begin{definition}[F-bar distribution]\label{def:fbar}
	Let $\boldsymbol w=(w_1,\dots,w_r)$ be a collection of non-negative weights summing to one, and $df$ be a positive real number.	The F-bar distribution with weights $\boldsymbol w$ and degrees of freedom $df$, denoted by $\bar F_{{\boldsymbol w},df}$, is a distribution of
	\begin{align}\label{eq:Fbar}
		\frac{\sum_{\ell=1}^r w_\ell Z_\ell}{Z_{0}/df  },
	\end{align}
	where $Z_0,Z_1,\dots,Z_r$ are mutually independent random variables with $Z_0 \sim \chi^2_{df}$ and $Z_\ell \sim \chi^2_{1}$ for $ 1\le \ell \le r$. Here, $\chi^2_\kappa$ denotes a chi-squared distribution with $\kappa>0$ degrees of freedom.
\end{definition}

The name attached to this family originates from its close relationship to both the chi-bar-squared distribution and to Snedecor's F distribution, which we denote as $\bar \chi^2_{\boldsymbol w}$ and $F_{r,df}$, respectively.
{In the Appendix, we show the following three essential properties of this family to be used later.
First, Snedecor's F is a special case when the entries of $\boldsymbol w$ are all equal. Second, the limiting case of $\bar F_{{\boldsymbol w},df}$ when $df \rightarrow \infty$ is $\bar \chi^2_{\boldsymbol w}$. Third, the standard normal distribution, whose CDF is denoted as $\Phi$,} is also a limiting case since, as $df \rightarrow \infty$ and $\max_{1\le\ell \le r}  w_\ell \rightarrow 0$,
\begin{align}\label{eq:normlim}
\frac{q_\tau(\bar F_{{\boldsymbol w},df}) - 1}{\sqrt{2\sum_{\ell=1}^r w_\ell^2 +2/df}}  \rightarrow q_\tau(\Phi)
\end{align}
for $\tau \in (0,1),$ where $q_\tau(G)$ denotes the $\tau$-th quantile of the distribution $G$.
The centering and rescaling in \eqref{eq:normlim} are done according to the limiting mean and variance of the underlying random variable from Definition~\ref{def:fbar} following the $\bar F_{{\boldsymbol w},df}$ distribution, while asymptotic normality results from mixing over infinitely many independent chi-squared variables.

Our reliance on the F-bar distribution is tied to its three properties described  in the previous paragraph and three closely related observations about the F statistic. These observations are: (i) the F statistic is distributed as $F_{r,n-m}$ if the errors are homoskedastic normal, (ii) the F statistic converges in distribution (after rescaling) to a chi-bar-squared if the number of restrictions $r$ is fixed, and (iii) the F statistic converges in distribution (after centering and rescaling) to a standard normal as $r$ grows. Therefore, the class of F-bar distributions serves as a roof designed both to match the finite sample distribution of the F statistic in an important special case and to approximate each of the possible limiting distributions after a suitable linear transformation.

\subsection{Critical value}

The proposed critical value {for the F statistic $F$} at a nominal size $\alpha \in (0,1)$ is a linear transformation of $q_{1-\alpha}\big( \bar F_{\hat {\boldsymbol w}, n-m} \big)$ given by
\begin{align}\label{eq:cv}
\hat c _\alpha &= \frac{1}{r\hat\sigma_\varepsilon^2} \left( \hat E_{\cal F}+ \hat V_{\cal F}^{1/2} \frac{ q_{1-\alpha}(\bar F_{\hat{\boldsymbol w},n-m}) - 1}{\sqrt{2\sum_{\ell=1}^r \hat w_\ell^2 +2/(n-m)}} \right)\!.
\end{align}
{The hat over $c _\alpha$ emphasizes that this value is data-dependent. The quantities $\hat{E}_{\cal F}$ and $\hat{V}_{\cal F}$ are related to $\cal F$ in \eqref{eq:calF}, the numerator of the F statistic. From this point forward, all means and variances are conditional on the regressors $\{\boldsymbol x_i\}_{i=1}^n$, and those with a subscript $0$ are calculated under $H_0$. The quantity} $\hat{E}_{\cal F}$ is an unbiased estimator of the conditional mean $\E_0[{\cal F}]$, while $\hat{V}_{\cal F}$  is either an unbiased or positively biased estimator of the conditional variance $\mathbb{V}_0[{\cal F}-\hat{E}_{\cal F}]$ as explained further below. The estimated weights $\hat{\boldsymbol w} = (\hat w_1,\dots,\hat w_r)$ are constructed to be consistent for weights $\boldsymbol w_{\cal F}$ in those cases where ${\cal F}/\E_0[{\cal F}]$ converges in distribution to $\bar \chi^2_{\boldsymbol w_{\cal F}}$.

The critical value $\hat c_\alpha$ ensures asymptotic size control irrespective of whether $r$ is viewed as fixed or growing with the sample size $n$. To explain why $\hat c_\alpha$ provides such uniformity, we consider first the case where $r$ grows. In this case, it is illuminating to rewrite the rejection rule as an equivalent event
\begin{align}
	\hat V_{\cal F}^{-1/2}({\cal F} - \hat{E}_{\cal F}) > \frac{ q_{1-\alpha}(\bar F_{\hat{\boldsymbol w},n-m}) - 1}{\sqrt{2\sum_{\ell=1}^r \hat w_\ell^2 +2/(n-m)}}.
\end{align}
Since $\hat V_{\cal F}^{-1/2}({\cal F} - \hat{E}_{\cal F})$ is asymptotically normal under the null, the validity in large samples follows from the relationship between the F-bar and standard normal distributions given in~\eqref{eq:normlim}.\looseness=-1

When instead $r$ is viewed as asymptotically fixed, it is more informative to express the rejection region through the inequality
\begin{align}\label{eq:fixedr}
	\frac{\cal F  }{\hat{E}_{\cal F}} >  q_{1-\alpha}(\bar F_{\hat {\boldsymbol w},n-m}) + \!\left(q_{1-\alpha}(\bar F_{\hat {\boldsymbol w},n-m})-1\right)\!\left(\frac{\hat V_{\cal F}^{1/2}/\hat{E}_{\cal F}}{\sqrt{2\sum_{\ell=1}^r \hat w_\ell^2 +2/(n-m)}} -1 \right)\!.
\end{align}
Note that rejecting when ${\cal F  }/{\hat{E}_{\cal F}}$ exceeds the quantile $q_{1-\alpha}(\bar F_{\hat {\boldsymbol w},n-m})$ suffices for validity; for the case of a single restriction such an approach corresponds to the standard practice of comparing squares of a heteroskedasticity robust t statistic and the $(1-\alpha)$-th quantile of Student's t distribution with $n-m$ degrees of freedom.\footnote{When testing a single restriction, $\hat {\boldsymbol w}$ must equal unity so that $\bar F_{\hat {\boldsymbol w},n-m}= F_{1,n-m} = t_{n-m}^2$, and in this case ${\cal F  }/{\hat{E}_{\cal F}}$ is the square of the t statistic studied in \citet[][Theorem 1]{kline2018leave}.} The last term on the right hand side of \eqref{eq:fixedr} can then be viewed as a finite sample correction that adjusts the critical value up or down depending on the relative size of the variance estimator for the ratio ${\cal F  }/{\hat{E}_{\cal F}}$, which is $\hat V_{\cal F}/\hat{E}_{\cal F}^2$, and the variance of the approximating distribution $\bar F_{\hat {\boldsymbol w},n-m}$, which is roughly $2\sum_{\ell=1}^r \hat w_\ell^2 +2/(n-m)$. As the ratio of these variances converges to unity when the number of restrictions is fixed, this term does not affect first order asymptotic validity.

Finally, note that if one is willing to rest on the assumption that the restrictions are numerous and the few restriction framework is superfluous, one might use the following simplified critical value not robust to few restrictions:\footnote{Such settings occur, for example, if the null of interest involves thousands of restrictions, in which case the two critical values $\hat c_\alpha$ and $\check c_\alpha$ are essentially equivalent but $\check c_\alpha$ is computationally simpler to construct as it circumvents computation of $\hat{\boldsymbol w}$.}
\begin{align}\label{eq:checkcv}
\check c _\alpha &=  \frac{1}{r\hat \sigma^2_{\varepsilon}} \left(\hat{E}_{\cal F}+{\hat{V}_{\cal F}^{1/2}} \frac{q_{1-\alpha}(F_{r,n-m})-1}{\sqrt{2/r+2/(n-m)}} \right)\!.
\end{align}


To complete the description of the proposed critical value, definitions of the quantities $\hat{E}_{\cal F}$, $\hat{V}_{\cal F}$ and $\hat{\boldsymbol w}$ are needed. Section~\ref{sec:centeredF} describes how we rely on leave-one-out OLS estimators to construct $\hat{E}_{\cal F}$ and $\hat{\boldsymbol w}$. For $\hat{V}_{\cal F}$, Section~\ref{sec:varest} provides the corresponding definition when it is possible to rely on leave-three-out OLS estimators, while Section~\ref{sec:cons} introduces the form of $\hat{V}_{\cal F}$ for settings where some of the leave-three-out estimators cease to exist. In the former case, it is possible to ensure that $\hat{V}_{\cal F}$ is unbiased, while the latter introduces a (small) positive bias. We initially consider the former case, a framework where the design matrix has full rank when any three observations are left out of the sample, and relax this condition in Section~\ref{sec:cons}.

\begin{assumption}\label{ass:Leave3Out}
	$\sum_{\ell \neq i,j,k} \boldsymbol x_\ell \boldsymbol x_\ell'$ is invertible for every $i,j,k \in \{1,\dots,n\}$.
\end{assumption}

When $\boldsymbol x_i$ is identically and continuously distributed with unconditional second moment $\E[\boldsymbol x_i \boldsymbol x_i']$ of full rank, Assumption~\ref{ass:Leave3Out} holds with probability one whenever $n - m \ge 3$. The asymptotic framework considers a setting where $n-m$ diverges so that Assumption~\ref{ass:Leave3Out} must hold in sufficiently large samples with continuous regressors. This conclusion also applies when $\boldsymbol x_i$ includes a few discrete regressors and, in particular, an intercept. In settings with many discrete regressors, Assumption~\ref{ass:Leave3Out} may fail to hold, even in large samples. For that reason, Section \ref{sec:cons} introduces the version of $\hat{V}_{\cal F}$ for empirical settings where the full rank condition is satisfied when any one observation is left out, but not necessarily when leaving two or three observations out.


\subsection{Leave-out algebra}\label{sec:LO}

Before describing $\hat{E}_{\cal F}$, $\hat{V}_{\cal F}$, and $\hat{\boldsymbol w}$ in detail, we will reformulate Assumption~\ref{ass:Leave3Out} using leave-out algebra. That is, we will derive an equivalent way of expressing this assumption while introducing notation that is essential for the construction of the critical value and for stating the asymptotic regularity conditions.

When $\boldsymbol S_{xx}$ has full rank, a direct implication of the Sherman-Morrison-Woodbury identity \cite[][SMW]{sherman1950adjustment,woodbury1949stability} is that the leave-one-out design matrix $\sum_{j \neq i}\boldsymbol x_j \boldsymbol x_j'$ is invertible if and only if the statistical leverage of the $i$-th observation $P_{ii} = \boldsymbol x_i' \boldsymbol S_{xx}\inverse \boldsymbol x_i$ is less than one. Letting $M_{ij} = \mathbf{1}\{i=j\} - \boldsymbol x_i'\boldsymbol S_{xx}\inverse \boldsymbol x_j$ be elements of the residual projection matrix $\boldsymbol M$ associated with the regressor matrix, this condition on the leverage is equivalently stated as $M_{ii}$ being greater than zero. When $M_{ii}>0$ holds, we can additionally use SMW to represent the inverse of  $\sum_{j \neq i} \boldsymbol x_j \boldsymbol x_j'$ as
\begin{align}\label{eq:SMW}
	\left(\sum\nolimits_{j \neq i} \boldsymbol x_j \boldsymbol x_j'\right)\inverse = \boldsymbol S_{xx}\inverse + \frac{\boldsymbol S_{xx}\inverse \boldsymbol x_i \boldsymbol x_i' \boldsymbol S_{xx}\inverse}{M_{ii}},
\end{align}
which highlights the role of a non-zero $M_{ii}$.

The representation in \eqref{eq:SMW} can also be used to understand when the leave-two-out design matrix $\sum_{k \neq i,j}\boldsymbol x_k \boldsymbol x_k'$ has full rank, since \eqref{eq:SMW} can be used to compute leverages in a sample that excludes $i$. After leaving observation $i$ out, the leverage of a different observation $j$ is $\boldsymbol x_j'\big(\sum_{k \neq i} \boldsymbol x_k \boldsymbol x_k'\big)\inverse \boldsymbol x_j$. To see when this leverage is less than one, note that \eqref{eq:SMW} yields
\begin{align}
	1- \boldsymbol x_j'\!\left(\sum\nolimits_{k \neq i} \boldsymbol x_k \boldsymbol x_k'\right)\inverse \boldsymbol x_j = M_{jj} - \frac{M_{ij}^2}{M_{ii}},
\end{align}
so that a necessary and sufficient condition for a full rank of $\sum_{k \neq i,j} \boldsymbol x_k \boldsymbol x_k'$ is that $D_{ij} >0$, where
\begin{align}\label{eq:D2}
	D_{ij} = \begin{vmatrix} M_{ii} & M_{ij} \\ M_{ij} & M_{jj} \end{vmatrix} = M_{ii}M_{jj} - M_{ij}^2,
\end{align}
and $\abs{\cdot}$ denotes the determinant.

Extending the previous argument to the case of leaving three observations out, we find that the invertibility of $\sum_{\ell \neq i,j,k}\boldsymbol x_\ell \boldsymbol x_\ell'$ for $i$, $j$, and $k$, all of which are different, is equivalent to $D_{ijk}>0$, where
\begin{align}\label{eq:D3}
	D_{ijk} &=  \begin{vmatrix} M_{ii} & M_{ij} & M_{ik} \\ M_{ij} & M_{jj} & M_{jk} \\ M_{ik} & M_{jk} & M_{kk} \end{vmatrix} = M_{ii} D_{jk} -\left( M_{jj} M_{ik}^2 + M_{kk} M_{ij}^2 - 2 M_{jk}M_{ij}M_{ik} \right).
\end{align}

This discussion reveals that Assumption~\ref{ass:Leave3Out} can equivalently be stated as requiring full rank of  $\boldsymbol S_{xx}$ and
\begin{align}\label{eq:L3O}
	D_{ijk} > 0 \text{ for every } i,j,k \in \{1,\dots,n\} \text{ with } i \neq j \neq k \neq i.
\end{align}
In addition to facilitating an algebraic description of Assumption~\ref{ass:Leave3Out}, the quantities $M_{ii}, \ D_{ij},$ and $D_{ijk}$ also play a role in the computation of the proposed critical value. Specifically, they can be used to avoid explicitly computing the OLS estimates after leaving one, two, or three observations out. Additionally, since construction of $\hat{E}_{\cal F}$, $\hat{V}_{\cal F}$, and $\hat{\boldsymbol w}$ relies on dividing by $M_{ii}, \ D_{ij},$ and $D_{ijk}$, the study of the asymptotic size of the proposed testing procedure imposes a slight strengthening of \eqref{eq:L3O}, which bounds the smallest $D_{ijk}$ away from zero.

\subsection{Location estimator}\label{sec:centeredF}

{
Recall that we recenter the numerator of the F statistic ($\cal F$ in \eqref{eq:calF}) by using $\hat{E}_{\cal F}$, which is an unbiased estimator of the conditional mean $\E_0[{\cal F}]$ under $H_0$. This mean equals}
\begin{align}\label{eq:Bias}
	\E_0[{\cal F} ] = \sum_{i=1}^{n}B _{ii}\sigma _{i}^{2},
\end{align}
where the values $B_{ij} = \boldsymbol x_i'\boldsymbol S_{xx}\inverse \boldsymbol R'\!\left( \boldsymbol R\boldsymbol S_{xx}\inverse \boldsymbol R' \right)\inverse \! \boldsymbol R \boldsymbol S_{xx}\inverse \boldsymbol x_j$ are observed and satisfy $\sum_{i=1}^{n}B _{ii} =r$. Furthermore, the exact null distribution of ${\cal F}/\E_0[{\cal F} ]$, under the additional condition of normally distributed regression errors, is $\bar \chi^2_{\boldsymbol w_{\cal F}}$, with $\boldsymbol w_{\cal F}$ containing the eigenvalues of the matrix\footnote{
{
Under error normality, the exact null distribution of $\boldsymbol R\hat {\boldsymbol \beta} - \boldsymbol q $ is ${\cal N}\big( 0, \V\big[ \boldsymbol R \hat{\boldsymbol \beta} \big] \big)$, and it follows from Lemma 3.2 of \cite{vuong1989} that ${\cal F}/\E_0[{\cal F} ]$ is distributed as a weighted sum of chi-squares with weights that are eigenvalues of $( \boldsymbol R\boldsymbol S_{xx}\inverse \boldsymbol R' )\inverse\V\big[ \boldsymbol R \hat{\boldsymbol \beta} \big] /\E_0[{\cal F} ]=\varOmega \big(\sigma_1^2,\dots,\sigma_n^2 \big).$
}
Note that the eigenvalues of $\varOmega \big(\sigma_1^2,\dots,\sigma_n^2 \big)$ are all real and non-negative as they can be expressed as the eigenvalues of the symmetric and positive semidefinite matrix $( \boldsymbol R\boldsymbol S_{xx}\inverse \boldsymbol R' )^{-1/2} \V\big[ \boldsymbol R \hat{\boldsymbol \beta} \big] ( \boldsymbol R\boldsymbol S_{xx}\inverse \boldsymbol R' )^{-1/2} /\E_0[{\cal F} ]$.
Furthermore, the entries of $\boldsymbol w_{\cal F}$ sum to one as $\E_0[{\cal F} ]$ is the trace of $( \boldsymbol R\boldsymbol S_{xx}\inverse \boldsymbol R' )\inverse \V\big[ \boldsymbol R \hat{\boldsymbol \beta} \big]$.
}
\begin{align}
	\varOmega \big(\sigma_1^2,\dots,\sigma_n^2 \big) 
=  \frac{1}{\sum_{i=1}^{n}B _{ii}\sigma _{i}^{2}}	\!\left( \boldsymbol R\boldsymbol S_{xx}\inverse \boldsymbol R' \right)\inverse \boldsymbol R \boldsymbol S_{xx}\inverse \left(\sum_{i=1}^n \boldsymbol x_i \boldsymbol x_i' \sigma_i^2 \right)\boldsymbol S_{xx}\inverse \boldsymbol R' .
\end{align}
Both $\E_0[{\cal F} ]$ and $\boldsymbol w_{\cal F}$ are thus functions of $\{\sigma_i^2\}_{i=1}^n$, and the relevance of the vector $\boldsymbol w_{\cal F}$ for asymptotic size control transcends the normality assumption on the errors that we used in order to introduce it.

As shown in \cite{kline2018leave}, the individual specific error variances can be estimated without bias for any value of $\boldsymbol\beta$ using leave-one-out estimators. Let the leave-$i$-out OLS estimator of $\boldsymbol\beta$ be $\hat {\boldsymbol\beta}_{-i} = \big(\sum_{j \neq i}\boldsymbol x_j\boldsymbol x_{j}'\big)\inverse \sum_{j \neq i} \boldsymbol x_j y_j,$ and construct
\begin{align}\label{eq:LOvar}
	\hat \sigma_i^2 = y_i\big(y_i -\boldsymbol x_i'\hat {\boldsymbol\beta}_{-i}\big).
\end{align}
With these leave-one-out estimators, we can estimate the null mean of ${\cal F}$ using
\begin{align}
	\hat{E}_{\cal F} = \sum_{i=1}^{n}B _{ii} \hat \sigma _{i}^{2},
\end{align}
which ensures that the first moment of ${\cal F} - \hat{E}_{\cal F} $ is zero under the null. Since $\hat \sigma_i^2$ is unbiased for any value of $\boldsymbol\beta$, this centered statistic still has its expectation minimized under $H_0$, so that large values of the statistic can be taken as evidence against the null. Following the same approach, we can estimate $\boldsymbol w_{\cal F}$ using the sample analog $\check {\boldsymbol w} = (\check w_1,\dots,{ \check w_r})'$, where $\check w_\ell$ is the $\ell$-th eigenvalue of $\varOmega \!\left(\hat \sigma_1^2,\dots,\hat \sigma_n^2 \right)$. However, $\check {\boldsymbol w}$ may not have non-negative entries summing to one, so we ensure that these conditions hold by letting $\hat {\boldsymbol w} = (\hat w_1,\dots,{ \hat w_r})'$, where
\begin{align}
\hat w_\ell = \frac{\check w_\ell \vee 0}{\sum_{\ell=1}^r  (\check w_\ell \vee 0) }.
\end{align}

While our construction of $\hat{E}_{\cal F}$ implies that the first moment of ${\cal F} - \hat{E}_{\cal F} $ is known when $H_0$ holds, its second moment still depends heavily on unknown parameters. Under $H_0$,
\begin{align}\label{eq:NullVar}
	{\V}_{0}\!\left[{\cal F} - \hat{E}_{\cal F} \right] = \sum_{i=1}^n \sum_{j \neq i} U_{ij} \sigma_i^2 \sigma_j^2 + \sum_{i=1}^n \!\left( \sum\nolimits_{j\neq i} V_{ij}\boldsymbol x_j'{\boldsymbol\beta} \right)^2 \sigma_i^2,
\end{align}
where $U_{ij} = 2\!\left( B_{ij} - {M_{ij}} \!\left({B_{ii}}/{M_{ii}}+{B_{jj}}/{M_{jj}}\right)\!/2 \right)^2$ and $V_{ij}=M_{ij} \!\left( {B_{ii}}/{M_{ii}}-{B_{jj}}/{M_{jj}}\right)$ are known quantities. This representation of the null-variance stems from writing ${\cal F} - \hat{E}_{\cal F}$ as a second order $U$-statistic with squared kernel weights of $U_{ij}/2$ plus a linear term with weights $\sum\nolimits_{j\neq i} V_{ij}\boldsymbol x_j'\boldsymbol\beta$ (see the Appendix for details).

\subsection{Variance estimator}\label{sec:varest}

This subsection describes the construction of an unbiased estimator of the conditional variance ${\V}_{0}\big[{\cal F} - \hat{E}_{\cal F} \big]$. As is evident from the representation in \eqref{eq:NullVar}, this variance depends on products of second moments such as the product $\sigma_i^2 \sigma_j^2$. While $\hat \sigma_i^2$ and $\hat \sigma_j^2$ are unbiased for $\sigma_i^2$ and $\sigma_j^2$, their product is not unbiased, as the estimation error is correlated across the two estimators. Some of this dependence can be removed by leaving both $i$ and $j$ out, but a bias remains as the remaining sample is used in estimating both $\sigma_i^2$ and $\sigma_j^2$. We therefore propose a leave-\textit{three}-out estimator of the variance product $\sigma_i^2 \sigma_j^2$. The product $\boldsymbol x_j'{\boldsymbol\beta}\,\boldsymbol x_k'{\boldsymbol\beta}\, \sigma_i^2$ appearing in the second component of ${\V}_{0}\big[{\cal F} - \hat{E}_{\cal F} \big]$ can similarly be estimated without bias using leave-three-out estimators.

Towards this end, let $\hat {\boldsymbol\beta}_{-ijk} = \big(\sum_{\ell \neq i,j,k}\boldsymbol x_\ell\boldsymbol x_{\ell}'\big)\inverse \sum_{\ell \neq i,j,k}\boldsymbol x_\ell y_{\ell}$ denote the OLS estimator of $\boldsymbol\beta$ applied to the sample that leaves observations $i$, $j$, and $k$ out. Then, define a leave-three-out estimator of $\sigma_i^2$ as
\begin{align}
	\hat \sigma_{i,-jk}^2 = y_i\big(y_i - \boldsymbol x_i'\hat {\boldsymbol\beta}_{-ijk}\big).
\end{align}
When $j$ and $k$ are identical, only two observations are left out, and we also write $\hat {\boldsymbol \beta}_{-ij}$ and $\hat \sigma_{i,-j}^2$. To construct an estimator of $\sigma_i^2 \sigma_j^2$, we first write the leave-two-out variance estimator $\hat \sigma_{i,-j}^2$ as a weighted sum (see Section \ref{sec:comp} for details)
\begin{align}\label{eq:L2OVar}
	\hat \sigma_{i,-j}^2 = y_i \sum_{k \neq j} \check M_{ik,-ij} y_k \qquad \text{where} \qquad \check M_{ik,-ij} = \frac{M_{jj} M_{ik} - M_{ij} M_{jk}}{D_{ij}}.
\end{align}
Then we multiply each summand above by a leave-three-out variance estimator $ \hat \sigma_{j,-ik}^2$, which leads to an unbiased estimator of $\sigma_i^2 \sigma_j^2$:
\begin{align}\label{eq:VarProd}
	\widehat{\sigma_i^2 \sigma_j^2} = y_i \sum_{k \neq j} \check M_{ik,-ij} y_k \cdot \hat \sigma_{j,-ik}^2.
\end{align}
While this construction appears to treat $i$ and $j$ in an asymmetric fashion, we show to the contrary that \eqref{eq:VarProd} is invariant to a permutation of the indices; $\widehat{\sigma_i^2 \sigma_j^2}=\widehat{\sigma_j^2 \sigma_i^2}$.

To understand why this proposal is unbiased for $\sigma_i^2 \sigma_j^2$, it is useful to highlight that $\hat \sigma_{j,-ik}^2$ is conditionally independent of $(y_i,y_k)$ and unbiased for $\sigma_{j}^2$, which, when coupled with \eqref{eq:L2OVar}, leads to unbiasedness immediately:
\begin{align}
	\E\big[\, \widehat{\sigma_i^2 \sigma_j^2} \,\big] &= \sum_{k \neq j} \E\!\left[ y_i \check M_{ik,-ij} y_k\right]\! \cdot \E\!\left[\hat \sigma_{j,-ik}^2 \right]
	= \E\!\left[ \hat \sigma_{i,-j}^2 \right]\!\sigma_j^2  =  \sigma_i^2 \sigma_j^2.
\end{align}

An unbiased estimator of the variance expression in \eqref{eq:NullVar} that utilizes the variance product estimator in \eqref{eq:VarProd} is
\begin{align}\label{eq:VcalF}
\hat{V}_{\cal F}\,
&= \sum_{i=1}^n \sum_{j \neq i} \!\left( U_{ij} - V_{ij}^2\right)\! \cdot \widehat{\sigma_i^2 \sigma_j^2}
+ \sum_{i=1}^n \sum_{j\neq i} \sum_{k\neq i} V_{ij} y_j \cdot V_{ik} y_k \cdot \hat \sigma_{i,-jk}^2.
\end{align}
Note that the product of the $(j=k)$-th terms in the second component generate, for each $i$, a term not present in \eqref{eq:NullVar} and whose non-zero expectation contains $V_{ij}^2\sigma_i^2 \sigma_j^2$; hence the use of $U_{ij} - V_{ij}^2$ instead of $U_{ij}$ in the first component.

\begin{rem}\label{rem:neg}
	In the process of establishing the asymptotic validity of the proposed test, we show that the variance estimator $\hat{V}_{\cal F}$ is close to the null variance $\V_0 \big[{\cal F} - \hat{E}_{\cal F} \big]$. In particular, this property implies that the variance estimator is positive with probability approaching one in large samples. However, negative values may still emerge in small samples. In such cases, we propose to replace the variance estimator with an upward biased alternative that uses squared outcomes as estimators of \textit{all} the error variances. This replacement is guaranteed positive, as is detailed in the Appendix, and therefore ensures that the critical value is always defined. Relatedly, Section \ref{sec:cons} considers settings where the design matrix may turn rank deficient after leaving certain triples of observations out of the sample. There, we similarly propose to use squared outcomes as estimators of \textit{some} error variances, namely those whose observations cause rank deficiency when left out of the sample.
\end{rem}

	\begin{rem}
		Note that in finite samples, the proposed critical value $\hat c_\alpha$ is not invariant to the value of $\boldsymbol \beta$. In practice, this means that finite sample size and power may be influenced by the size of the regression coefficients.  In Section \ref{sec:sim} we analyze via simulations the impact of the average signal size on the finite sample size and power through the regression $R^2$.
	\end{rem}

{
	\begin{rem}\label{rem:restricted}
The null restrictions $\boldsymbol{R\beta }=\mathbf{q}$ can be imposed during the estimation of auxiliary quantities (i.e., $\sigma _{i}^{2},\ \sigma _{i}^{2}\sigma _{j}^{2},\dots$), and the restricted estimates may be used in place of our proposals that do not impose those restrictions (i.e., $\widehat{\sigma }_{i}^{2},\  \widehat{\sigma _{i}^{2}\sigma _{j}^{2}}, \dots)$.\footnote{{We thank an anonymous referee for raising this point.}} In the Appendix, we show how this idea can be implemented, and leave further investigation to future research. Incorporating restricted estimates may make the the procedure more complex and slow down computations, but it may also improve the efficiency of the auxiliary estimates.
	\end{rem}
}



\subsection{Computational remarks}\label{sec:comp}

While the previous subsections introduced the location estimator $\hat{E}_{\cal F}$, variance estimator $\hat{V}_{\cal F}$, and empirical weights $\hat{\boldsymbol w}$ using leave-out estimators of $\boldsymbol\beta$, we note here that direct computation of $\hat {\boldsymbol\beta}_{-i}$, $\hat {\boldsymbol\beta}_{-ij}$, and $\hat {\boldsymbol\beta}_{-ijk}$ can be avoided by using the Sherman-Morrison-Woodbury (SMW) identity. Specifically, \eqref{eq:SMW} implies that
\begin{align}
	y_i - \boldsymbol x_i'\hat {\boldsymbol\beta}_{-i} = \frac{y_i -\boldsymbol x_i'\hat {\boldsymbol\beta}}{M_{ii}},
\end{align}
so that computation of $\hat {\boldsymbol\beta}_{-i}$ can be avoided when constructing the leave-one-out variance estimator $\hat \sigma _{i}^{2} = {y_i(y_i -\boldsymbol x_i'\hat {\boldsymbol\beta}_{-i})}$. Similarly, it is possible to show that for $i$ and $j$ not equal,
\begin{align}
y_i -\boldsymbol x_i'\hat {\boldsymbol\beta}_{-ij} = \frac{M_{jj}(y_i -\boldsymbol x_i'\hat {\boldsymbol\beta}) - M_{ij}(y_j -\boldsymbol x_j'\hat {\boldsymbol\beta})}{D_{ij}},
\end{align}
which leads to \eqref{eq:L2OVar}, and, for $i$, $j$, and  $k$, all of which are different,
\begin{align}
y_i -\boldsymbol x_i'\hat {\boldsymbol\beta}_{-ijk} = \frac{(y_i -\boldsymbol x_i'\hat {\boldsymbol\beta}) - M_{ij}(y_j-\boldsymbol x_j'\hat {\boldsymbol\beta}_{-jk}) - M_{ik}(y_k-\boldsymbol x_k'\hat {\boldsymbol\beta}_{-jk})}{D_{ijk}/D_{jk}}.
\end{align}

These relationships allow for recursive computation of the leave-out residuals and therefore for simple construction of the variance estimators $\hat \sigma _{i}^{2}$, $\hat \sigma _{i,-j}^{2}$, and $\hat \sigma _{i,-jk}^{2}$ needed to compute the components of the critical value $c_{\alpha}$. In particular, the location estimator $\hat{E}_{\cal F}$ and empirical weights $\hat{\boldsymbol w}$, which require only the leave-one-out residuals, can be computed without explicit loops, by relying instead on elementary matrix operations applied to the matrices containing $M_{ij}$ and $B_{ij}$ as well as the data matrices. Similarly, all doubly indexed objects entering the variance estimator $\hat{V}_{\cal F}$ can be computed by elementary matrix operations. Those objects are $D_{ij}$, $V_{ij}$, $U_{ij}$, and the leave-two-out residuals. The remaining objects entering $\hat{V}_{\cal F} $ can be computed by a single loop across $i$ with matrices containing $D_{ijk}$ and leave-three-out residuals renewed at each iteration.

Additionally, the above representations of leave-out residuals demonstrate how $M_{ii}\inverse$, $D_{ij}\inverse$ and $D_{ijk}\inverse$ enter the critical value, and thus highlight the need for bounding $D_{ijk}$ away from zero when analyzing the large sample properties of the proposed test. 

\begin{rem}
	The quantile $q_{1-\alpha}(\bar F_{\hat{\boldsymbol{w}},n-m})$ can easily be constructed by simulating the distribution of the random variable in \eqref{eq:Fbar} conditional on the realized value of $\hat{\boldsymbol{w}}$.
\end{rem}

\section{Asymptotic size and power}\label{sec:asymptotia}

This section studies the asymptotic properties of the proposed test. Specifically, we provide a set of regularity conditions under which the test has a correct asymptotic size and non-trivial power against local alternatives. All limits are taken as the sample size $n$ approaches infinity. In studying asymptotic size, we allow for the number of restrictions $r$ and/or number of regressors $m$ to be fixed or diverging with $n$, and show that the asymptotic size is controlled uniformly over the two situations by the test, which is therefore robust to the type of asymptotics. When studying asymptotic power, we focus on the case of many restrictions, i.e., $r$ diverging with $n$. The ordering $r \le m < n-3$ is maintained throughout.

{
\subsection{Assumptions}

To establish asymptotic validity of the proposed test, we impose some regularity conditions. We begin by outlining the assumptions for the sampling scheme.

\begin{assumption}\label{ass:iid}
 $\{(y_i,\boldsymbol x_i')'\}_{i=1}^n$ are i.i.d., $\E[\varepsilon_i | \boldsymbol x_i] =0$, $\max_{1\le i\le n} \big(\E[\varepsilon_i^4 |\boldsymbol x_i] + \sigma_i^{-2}\big) = O_p(1)$.
\end{assumption}

Assumption~\ref{ass:iid}} places restrictions on the error conditional moments: an upper bound on the conditional fourth moments and lower bound on the skedastic function. Such restrictions are typically required when heteroskedasticity is allowed \citep[see, e.g.,][]{chao2012asymptotic,cattaneo2017inference,kline2018leave}.

{Next, we impose regularity conditions on the regressors to ensure convergence of the centered statistic ${\cal F} - \hat{E}_{\cal F}$ to the normal distribution when the number of restrictions grow large. These conditions restrict the weights on the regression errors in the bilinear form of ${\cal F} - \hat{E}_{\cal F}$ \footnote{{See, e.g., expansion \eqref{eq:Fminus} below and its extended version \eqref{eq:Ustat} in the Appendix.}}. These weights contain various functions of regressors (in particular, potentially unbounded regression function values $\boldsymbol x_i'{\boldsymbol\beta}$), and the purpose of Assumption~\ref{ass:reg} is to restrict their asymptotic behavior.



\begin{assumption}\label{ass:reg}
There exists a sequence $\epsilon_n \rightarrow 0$ such that $(i)$ $\epsilon_n^{1/3} \max_{1\le i\le n} (\boldsymbol x_i'{\boldsymbol\beta})^2 = O_p( 1 )$ and (ii) at least one of the following two conditions is satisfied:
	\begin{enumerate}
		\item[(a)] $\max_{1\le i\le n} B_{ii} = o_p( \epsilon_n )$,
		
		\item[(b)] $\max_{1\le i\le n}\,(\sum_{j \neq i} V_{ij}\boldsymbol x_j'{\boldsymbol\beta} )^2/r = o_p( 1 )$ and $\epsilon_n r \rightarrow \infty$.
	\end{enumerate}	
\end{assumption}
}

{Part $(i)$ of Assumption~\ref{ass:reg},} which places bounds on regression function values, is used to control the variance of the leave-out estimators $\hat \sigma_i^2$ and $\hat \sigma_{i,-jk}^2$. This condition places a rate bound on extreme outliers among the individual signals, and its role is primarily to control certain higher order terms. The condition $(i)$ is used to establish both size control and local power properties, so we stress that it pertains to the actual data generating process, not just the hypothesized value of $\boldsymbol\beta$.
Also note that {Assumption~\ref{ass:reg}$(i)$} differs from Assumption 1(iii) of \cite{kline2018leave} in that we allow $\boldsymbol x_i'\boldsymbol\beta$ to have an unbounded support so the maximum over $i$ may be slowly diverging with $n$. This relaxation is important, as it allows regressors with unbounded support and associated non-zero regression coefficients.

{Part $(ii)$ of Assumption~\ref{ass:reg} is an analogue of the Lindeberg condition in the central limit theorem for weighted sums of independent random variables, in that it also controls the collective asymptotic behavior of weights, but the weights in a bilinear form of independent regression errors.} When the number of tested restrictions is fixed, {this condition} implies that the estimator of the tested contrasts $\boldsymbol R \hat{\boldsymbol\beta} - \boldsymbol q$ is asymptotically normal. When the number of restrictions is growing, {this condition} is weaker and involves only a high-level transformation of the regressors $\sum_{j \neq i} V_{ij}\boldsymbol x_j'\boldsymbol\beta$, which enters ${\cal F} - \hat E_{\cal F}$ as a weight on the $i$-th error term $\varepsilon_i$. To ensure that the asymptotic distribution of ${\cal F} - \hat E_{\cal F}$ does not depend on the unknown distribution of any one error term, we therefore require that no squared $\sum_{j \neq i} V_{ij}\boldsymbol x_j'\boldsymbol\beta$ dominates the variance ${\V}_{0}\big[{\cal F} - \hat{E}_{\cal F} \big]$, which in turn is proportional to $r$.
{Assumption~\ref{ass:reg}$(ii)$} can be verified in particular applications of interest. For example, the Appendix shows that {part $(ii)(b)$ of Assumption~\ref{ass:reg}} holds in models characterized by group specific regressors.

The next assumption imposes the previously discussed regularity condition that the determinant $D_{ijk}$ is bounded away from zero for any $i$, $j$, and $k$, all of which are different. This condition will be relaxed in Section \ref{sec:cons}, where such a version of $\hat V_{\cal F}$ is introduced that exists even when leaving two or three observations out leads to rank deficiency of the design.
\begin{assumption}\label{ass:asympLeave3Out}
	$\max_{i \neq j \neq k \neq i} D_{ijk}\inverse = O_p(1)$.
\end{assumption}

\subsection{Asymptotic size}

Under the regularity conditions in Assumptions \ref{ass:reg} and \ref{ass:asympLeave3Out}, the following theorem establishes the asymptotic validity of the proposed testing procedure.

\begin{theorem}\label{thm:size}
	If Assumptions \ref{ass:Leave3Out}, \ref{ass:iid}, \ref{ass:reg}, and \ref{ass:asympLeave3Out} hold, then, under $H_0$,
	\begin{align}
		\lim_{n \rightarrow \infty} \Pr\left( F> \hat c_\alpha \right) = \alpha.
	\end{align}
\end{theorem}

{A discussion of the structure of the decision event $F > \hat c_\alpha$ may aid in understanding why size control occurs in the critical case of many restrictions, $r\to\infty$, 
and where the challenges come from. The critical value is then asymptotically close to $\check c _\alpha$ defined in \eqref{eq:checkcv}, and the decision event is equivalent to $\hat V_{\cal F}^{-1/2}({\cal F} - \hat{E}_{\cal F}) > \big(q_{1-\alpha}(F_{r,n-m})-1\big)/\sqrt{2/r+2/(n-m)}$. The right side is asymptotically standard normal, so the size control rests on asymptotic standard normality of $\hat V_{\cal F}^{-1/2}({\cal F} - \hat{E}_{\cal F})$, which, in turn, is shown using a central limit theorem for ${\cal F} - \hat{E}_{\cal F}$ and its asymptotically correct standardization by $\hat V_{\cal F}^{1/2}$.

The demeaned statistic ${\cal F} - \hat{E}_{\cal F}$ has a representation of} a bilinear form in independent, not necessarily identically distributed, random variables (see equation (3) in the Appendix):
\begin{align}\label{eq:Fminus}
\mathcal{F}-\hat{E}_{\mathcal{F}}=
\sum_{i=1}^{n}\sum_{j\neq i}C_{ij}u_{i}v_{j}+\sum_{j=1}^{n}c_{\boldsymbol \beta j}w_{j},
\end{align}
where, in our case, $u_{j}=v_{j}=w_{j}=\varepsilon _{j}$ are regression errors, and the
coefficients $c_{\boldsymbol \beta j}$ in the second, linear component, explicitly
depend on the unknown parameters $\boldsymbol \beta $. {Note that the quadratic term in \eqref{eq:Fminus} has a jackknife structure and lacks terms of the type $C_{ii}\varepsilon _{i}^2$, whose presence would introduce higher-order moments into the variance.\footnote{{See also \cite{calhoun2011many} for an example of a structure where higher-order moments do arise and need to be tediously estimated.}} This further highlights the importance of demeaning the original statistic ${\mathcal{F}}$. }

The econometric theory literature is populated by central limit theorems (CLTs) handling the
asymptotics of \eqref{eq:Fminus} under various assumptions, starting from an
early CLT in \cite{kelejian2001spatial} that was originated in a spatial
regression environment. More recent examples are the CLT of \cite{chao2012asymptotic}, formulated in the many-weak-instrument context, the CLT of \cite{kline2018leave} designed for many-regressor models, and the CLT of \cite{anatmik2021factor} targeting big factor models. \cite{cattaneo2016alternative} presented a unifying framework leading to the
description of asymptotical behavior of V-statistics that further generalize
bilinear forms like \eqref{eq:Fminus}. {Finally, \cite{kuersteiner2020} provide a CLT for expressions like \eqref{eq:Fminus}, allowing for data-dependent $C_{ij}$ and $c_{\boldsymbol \beta j}$, heteroskedasticity, and some forms of dependence across $j$.} In all these setups, the asymptotic normality of \eqref{eq:Fminus} eventually results in an asymptotically normal
test statistic when \eqref{eq:Fminus} is suitably standardized.

The next challenge is converting the asymptotic normality of \eqref{eq:Fminus} into an asymptotically valid test.\footnote{Often, the bilinear form \eqref{eq:Fminus} has a tighter structure, which simplifies the emergence of asymptotic normality and further pivotization, removing the challenges handled in the present paper. For example, simplification may come from a slow growth of incidental parameters' dimensionality \citep{hong1995nparmtest}, absence of the quadratic component in \eqref{eq:Fminus} \citep{breitung2016hetpanel}, or assumption of conditional homoskedasticity \citep{anatolyev2012inference}.
Likewise, in the many weak instrument literature, a variety of tests are also based on asymptotic normality of bilinear forms, with simplifying deviations from the general form \eqref{eq:Fminus}. In J type tests for validity of many instruments and thus many restrictions \citep{anatgosp2011test,leeokui2012hausmantest,chao2014testing}, the dependence on a parameter of asymptotically fixed dimensionality can be handled using its plug-in estimate. In Anderson-Rubin type tests for few parameter restrictions \citep{anatgosp2011test,crudu2021ar,miksun2021manyweak}, the number of parameter restrictions is asymptotically fixed, and one can use restricted null values of the parameters. In our situation, in contrast, both the restriction numerosity and parameter dimensionality are asymptotically increasing, at least when $r$ is asymptotically growing.} To obtain an asymptotically pivotal statistic, \eqref{eq:Fminus} needs to be standardized by a consistent estimate of its (rescaled) variance \eqref{eq:NullVar}. Constructing such an estimate is difficult because of explicit dependence of the coefficients $c_{\boldsymbol \beta j}$ and implicit dependence of the regression errors $\varepsilon _{i}$ on $\boldsymbol \beta ,$ a high-dimensional parameter when regressors are many. {The estimator \eqref{eq:VcalF} based on leave-three-out error variance estimates accomplishes this goal.}

As our treatment covers two asymptotically different frameworks under one roof, we use two CLTs in this paper.  One, formulated as Lemma B.1 in \cite{kline2018leave}, which builds on Lemmas A2.1 and A2.2 in \cite{solv2020robmany}, pertains to the case of asymptotically growing $r$. The other is the regular Lyapounov central limit theorem, which pertains to the case of asymptotically fixed $r$. 

\subsection{Asymptotic power}

To describe the power of the proposed test, we introduce a drifting sequence of local alternatives indexed by a deviation $\boldsymbol\delta$ from the null times $(\boldsymbol R\boldsymbol S_{xx}\inverse \boldsymbol R')^{1/2}$, which specifies the precision the tested linear restrictions can be estimated with in the given sample. Thus, we consider alternatives of the form
\begin{align}\label{DriftDGP}
	H_\delta \ : \ \boldsymbol R {\boldsymbol\beta} =\boldsymbol q + (\boldsymbol R\boldsymbol S_{xx}\inverse\boldsymbol R')^{1/2} \cdot \boldsymbol\delta,
\end{align}
for ${\boldsymbol\delta} \in \R^r$ satisfying the limiting condition
\begin{align}
	\lim_{n,r \rightarrow \infty} \frac{\norm{\boldsymbol\delta}}{r^{1/4}} = \Delta_{\delta} \in [0,\infty].
\end{align}
Below we show that the power of the test is monotone in $\Delta_\delta$, with power equal to size when $\Delta_\delta=0$ and power equal to one when $\Delta_\delta=\infty$.

The role of $(\boldsymbol R\boldsymbol S_{xx}\inverse\boldsymbol R')^{1/2}$ in indexing the local alternatives is analogous to that of $n^{-1/2}$ often used in parametric problems. However, in settings with many regressors some linear restrictions may be estimated at rates that are substantially lower than the standard parametric one. Therefore, we index the deviations from the null by the actual rate of $(\boldsymbol R\boldsymbol S_{xx}\inverse\boldsymbol R')^{1/2}$ instead of $n^{-1/2}$.

The alternative is additionally indexed by $\boldsymbol\delta$, which in standard parametric problems is typically fixed. However, fixed $\boldsymbol \delta$ is less natural here, as the dimension of $\boldsymbol\delta$ increases with sample size. Instead, we fix the limit of its Euclidean norm when scaled by $r^{1/4}$. This approach allows us to discuss different types of alternatives and how the numerosity of the tested restrictions affects the test's ability to detect deviations from the null. Specifically, note that when the deviation $\boldsymbol\delta$ is sparse, i.e., only a bounded number of its entries are non-zero, then the test has a non-trivial power against alternatives, {whose individual elements on average diverges} at a rate that is $r^{1/4}$ lower than when only a fixed number of restrictions is tested. This observation highlights the cost for the power of including many \textit{irrelevant} restrictions in the hypothesis. On the other hand, if $\boldsymbol\delta$ is dense, e.g., with all entries bounded away from zero, then the test can detect local deviations, {in which an individual element on average shrinks} at a rate that is $r^{1/4}$ greater than the usual. This means that if the tested restrictions can be estimated at the parametric rate and they are all \textit{relevant}, then the test can detect deviations from the null of order ${n^{-1/2} r^{-1/4}}$.

The following theorem states the asymptotic power under sequences of local alternatives of the form given in \eqref{DriftDGP} and discussed above.

\begin{theorem}\label{thm:power}
	If Assumptions \ref{ass:Leave3Out}, \ref{ass:iid}, \ref{ass:reg}, and \ref{ass:asympLeave3Out} hold, then, under $H_\delta$,
	\begin{align}
	\lim_{n,r \rightarrow \infty}\Pr\left( F> \hat c_\alpha \right)
	- \Phi \left( \Phi ^{-1}\left( \alpha \right)
	+ \Delta_{\delta}^2 \left( {\V_{0}\!\big[ \mathcal{F}-\hat{E}_{\cal F}\big]}/r\right)^{-1/2}
	\right) = 0 ,
	\end{align}
	where $\Phi $ denotes the cumulative distribution function of the standard normal and $\Phi(\infty)=1$.
\end{theorem}

\begin{rem}\label{rem:power}
	It is instructive to compare the power curve documented in Theorem~\ref{thm:power} with the asymptotic power curve of the exact F test when both tests are valid. When the individual error terms are homoskedastic normal with variance $\sigma^2$, the asymptotic power of the exact F test is the limit of \citep{anatolyev2012inference}
	\begin{align}
		\Phi \left( \Phi ^{-1}\left( \alpha \right)
		+ \Delta_{\delta}^2 \left(2\sigma^4 + 2\sigma^4r/(n-m)\right)^{-1/2}
		\right).
	\end{align}
	Thus, the relative asymptotic power of the proposed LO test and the exact F test is determined by the limiting ratio of $r\inverse\V_{0}\big[ \mathcal{F}-\hat{E}_{\cal F}\big]$ to $2\sigma^4\!\left(1 + r/{(n-m)}\right)$. The Appendix shows that this ratio approaches one in large samples if the number of tested restriction is small relative to the sample size or if the limiting variability of  $B_{ii}/M_{ii}$ is small  \citep[][calls this a balanced design]{kline2018leave}. When neither of these conditions holds, the proposed test will, in general, have a slightly lower power than the exact F test, which we also document in the simulations in Section \ref{sec:sim}.	
\end{rem}

\begin{rem}\label{rem:multi}
	The order of the numerosity of alternatives that can be detected with the proposed test is optimal in the minimax sense when the alternatives are moderately sparse to dense, i.e., when $O(\sqrt{r})$ or more of the tested restrictions are violated \citep{arias2011global}. However, if the alternative is strongly sparse so that at most $o(\sqrt{r})$ tested restrictions are violated, a higher power can be achieved by tests that redirect their power towards those alternatives. Such tests typically focus their attention on a few largest t statistics (i.e., smallest p values) and are often described as multiple comparison procedures \citep{donoho2004higher,romano2010multiple}. While such tests can control size when the error terms are homoskedastic normal, it is not clear whether they can do so in the current semiparametric framework with an unspecified error distribution. The issue is that the size control for multiple comparisons relies on knowing the (normal or t) distributions of individual t statistics, but in the current framework with many regressors those distributions are not necessarily known (even asymptotically).
\end{rem}

\section{If leave-three-out fails}\label{sec:cons}

This section extends the definition of the critical value $c_{\alpha}$ to settings where the design matrix may turn rank deficient after leaving certain pairs or triples of observations out of the sample. When Assumption~\ref{ass:Leave3Out} fails in this way, $\hat E_{\cal F}$ is still an unbiased estimator of $\E_0[{\cal F}]$, but the unbiased variance estimator introduced in Section~\ref{sec:varest} does not exist. For this reason, we propose an adjustment to the variance estimator that introduces a positive bias for pairs of observations where we are unable to construct an unbiased estimator of the variance product $\sigma_i^2 \sigma_j^2$ and for triples of observations where we are unable to construct an unbiased estimator of $\boldsymbol x_j'{\boldsymbol\beta} \,\boldsymbol x_k'{\boldsymbol\beta}\, \sigma_i^2$. This introduction of a positive bias to the variance estimator ensures asymptotic size control, even when Assumption~\ref{ass:Leave3Out} fails.


Since this section considers a setup where Assumption~\ref{ass:Leave3Out} may fail, we introduce a weaker version of the assumption, which only imposes the full rank of the design matrix after dropping any one observation.

\begin{customass}{1'}\label{ass:Leave1Out}
	$\sum_{j \neq i}\boldsymbol x_j\boldsymbol x_j'$ is invertible for every $i \in \{1,\dots,n\}$.
\end{customass}

One can always satisfy this assumption by appropriately pruning the sample, the model, and the hypothesis of interest. For example, if $\boldsymbol S_{xx}$ does not have full rank, then one can remove unidentified parameters from both the model and hypothesis of interest, and proceed by testing the subset of restrictions in $H_0$ that are identified by the sample. Similarly, if $\sum_{j \neq i}\boldsymbol x_j\boldsymbol x_j'$ does not have full rank for some observation  $i$, then there is a parameter in the model which is identified only by this observation. Therefore, one can proceed as in the case of rank deficiency of $\boldsymbol S_{xx}$, by dropping observation $i$ from the sample and by removing the parameter that determines the mean of this observation from the model and null hypothesis. When doing this for any observation $i$ such that $\sum_{j \neq i}\boldsymbol x_j\boldsymbol x_j'$ is non-invertible, one obtains a sample that satisfies Assumption~\ref{ass:Leave1Out} and can be used to test the restrictions in $H_0$ that are identified by this \textit{leave-one-out sample}.

\subsection{Variance estimator}\label{sec:adjusted}

When Assumption~\ref{ass:Leave3Out} fails, some of the unbiased estimators $\hat \sigma^2_{i,-jk}$ and $\widehat{\sigma_i^2 \sigma_j^2}$ cease to exist. For such cases, the variance estimator $\hat{V}_{\cal F}$ utilizes replacements that are either also unbiased or positively biased, depending on the \textit{cause} of the failure. Assumption~\ref{ass:Leave3Out} fails if $D_{ijk} =0$ for some triple of observations, and we say that this failure of full rank is \textit{caused by $i$} if $D_{jk}>0$ or $D_{ij}D_{ik}=0$, i.e., if the design retains full rank when only observations $j$ and $k$ are left out or if leaving out observations $(i,j)$ or $(i,k)$ leads to rank deficiency. Our replacement for $\hat \sigma^2_{i,-jk}$ is biased when $i$ causes $D_{ijk} =0$, while the replacement for $\widehat{\sigma_i^2 \sigma_j^2}$ is biased when both $i$ and $j$ cause $D_{ijk} =0$ for some $k$.

To introduce the replacement for $\hat \sigma^2_{i,-jk}$, we consider the case when it does not exist, or equivalently, when $D_{ijk}=0$. If $i$ causes this leave-three-out failure, then our replacement is the upward biased estimator $y_i^2$. When this failure of leave-three-out is not caused by $i$, the leave-two-out estimators $\hat \sigma_{i,-j}^2$ and $\hat \sigma_{i,-k}^2$ are equal and independent of both $y_j$ and $y_k$ (as shown in the Appendix). These properties imply that $y_j y_k \hat \sigma_{i,-j}^2$ is an unbiased estimator of $\boldsymbol x_j'{\boldsymbol\beta}\,\boldsymbol x_k'{\boldsymbol\beta}\, \sigma_i^2$, and we therefore use $\hat \sigma_{i,-j}^2$ as a replacement for $\hat \sigma^2_{i,-jk}$. To summarize, we let
\begin{align}
	\bar \sigma_{i,-jk}^2 =
	\begin{cases}
	\hat \sigma^2_{i,-jk}, & \text{if }  D_{ijk}>0, \\
	\hat \sigma^2_{i,-j}, & \text{if }  D_{jk}=0 \text{ and } D_{ij} D_{ik}>0, \\
	y_i^2, & \text{otherwise} .
	\end{cases}
\end{align}
When $j$ is equal to $k$, we consider pairs of observations, and the definition only involves the last two lines since $D_{ijj}=0$. In this case, we also write $\bar \sigma_{i,-j}^2$ for $\bar \sigma_{i,-jj}^2$.

For the replacement of $\widehat{\sigma_i^2 \sigma_j^2} = y_i \sum_{k \neq j} \check M_{ik,-ij} y_k \cdot \hat \sigma_{j,-ik}^2$, we similarly consider the case where this estimator does not exist, i.e., where $D_{ijk}=0$ for a $k$ not equal to $i$ or $j$. When any such rank deficiency is caused by both $i$ and $j$, we rely on the upward biased replacement $y_i^2 \bar \sigma_{j,-i}^2$. When none of the leave-three-out failures are caused by both $i$ and $j$, the replacement uses $\bar \sigma_{i,-jk}^2$ in place of $\hat \sigma_{i,-jk}^2$. To summarize, we define
\begin{align}
\overline{\sigma_i^2 \sigma_j^2} =
\begin{cases}
y_i \sum_{k \neq j} \check M_{ik,-ij} y_k \cdot \bar \sigma_{j,-ik}^2, & \text{if } D_{ij}>0 \text{ and } (D_{ijk}>0 \text{ or } D_{ik} D_{jk}=0 \text{ for all } k), \\
y_i^2 \bar \sigma_{j,-i}^2, & \text{otherwise}.
\end{cases}
\end{align}

This estimator is unbiased for $\sigma_i^2 \sigma_j^2$ when none of the leave-three-out failures are caused by both $i$ and $j$, i.e., when the first line of the definition applies.  Unbiasedness holds because the presence of a bias in $\bar \sigma_{j,-ik}^2$ implies that $j$ is causing the leave-three-out failure. Therefore, $i$ cannot be the cause, which yields that $\hat \sigma^2_{i,-j}$ is independent of $y_k$, or equivalently, that $\check M_{ik,-ij} =0$.


Now, we describe how these replacement estimators enter the variance estimator $\hat{V}_{\cal F}$. When $\overline{\sigma_i^2 \sigma_j^2}$ or $\bar \sigma_{i,-jk}^2$ are biased and would enter the variance estimator with a negative weight, we remove these terms, as they would otherwise introduce a negative bias. For $\overline{\sigma_i^2 \sigma_j^2}$, the weight is $U_{ij} - V_{ij}^2$, so a biased variance product estimator is removed when $U_{ij} - V_{ij}^2 < 0$. For $\bar \sigma_{i,-jk}^2$, the weight is $V_{ij}  y_j \cdot V_{ik}  y_k$, but $\bar \sigma_{i,-jk}^2$ does not depend on $j$ and $k$ when it is biased, so we sum these weights across all such $j$ and $k$, and we remove the term if this sum is negative.


The following variance estimator extends the definition of $\hat{V}_{\cal F}$ to settings where leave-three-out may fail:
\begin{align}
	\hat{V}_{\cal F}\,
	&= \sum_{i=1}^n \sum_{j \neq i} \big( U_{ij} - V_{ij}^2\big) \cdot G_{ij} \cdot \overline{\sigma_i^2 \sigma_j^2}
	+ \sum_{i=1}^n \sum_{j\neq i} \sum_{k\neq i}  V_{ij}  y_j \cdot V_{ik}  y_k \cdot G_{i,-jk} \cdot \bar \sigma_{i,-jk}^2,
\end{align}
where the indicators $G_{ij}$ and $G_{i,-jk}$ remove biased estimators with negative weights:
\begin{align}
G_{ij} &= \begin{cases}
	0,	& \text{if } \overline{\sigma_i^2 \sigma_j^2} = y_i^2 \bar \sigma_{j,-i}^2 \text{ and } U_{ij} - V_{ij}^2 <0, \\
	1, & \text{otherwise},
\end{cases} \\
G_{i,-jk} &=\begin{cases}
	0, & \text{if } \bar \sigma_{i,-jk}^2 = y_i^2 \text{ and } \sum_{j \neq i} \sum_{k \neq i} V_{ij}  y_j \cdot V_{ik}  y_k \cdot \mathbf{1}\!\left\{ \bar \sigma_{i,-jk}^2 = y_i^2  \right\}< 0,  \\
	1, & \text{otherwise}.
\end{cases}
\end{align}

\subsection{Asymptotic size}\label{sec:adjsize}

In order to establish that the proposed test controls asymptotic size when there are some failures of leave-three-out, we replace the regularity condition in Assumption~\ref{ass:asympLeave3Out} with an analogous version that allows for some of the determinants $D_{ij}$ and $D_{ijk}$ to be zero. Otherwise, the role of Assumption~\ref{ass:asympLeave1Out} below is the same as Assumption~\ref{ass:asympLeave3Out} in that it rules out denominators that are arbitrarily close to zero.

\begin{customass}{3'}\label{ass:asympLeave1Out}
	(i) $\max\limits_{1 \le i \le n} M_{ii}\inverse = O_p(1)$, and (ii) $\max\limits_{i,j : D_{ij} \neq 0}  D_{ij}\inverse + \max\limits_{i,j,k : D_{ijk} \neq 0}  D_{ijk}\inverse = O_p(1)$.
\end{customass}

When computing $\hat{V}_{\cal F}$, one must account for machine zero imperfections while comparing $D_{ij}$ and $D_{ijk}$ with zero in the definitions of $\bar \sigma_{i,-jk}^2$ and $\overline{\sigma_i^2 \sigma_j^2}$. Such imperfections are typically of order $10^{-15}$; however, we propose to compare $D_{ij}$ to $10^{-4}$ and $D_{ijk}$ to $10^{-6}$. Doing so will replace any potential case of a small denominator with an upward biased alternative and ensures that Assumption~\ref{ass:asympLeave1Out}$(ii)$ is automatically satisfied.

The following theorem establishes the asymptotic validity of the proposed leave-out test in settings where Assumption~\ref{ass:Leave3Out} fails. The theorem pertains to a nominal size below $0.31$, as the upward biased variance estimator may not ensure validity in cases where a nominal size above $0.31$ is desired. This happens because the quantile $q_{1-\alpha}(\bar F_{\hat{\boldsymbol w},n-m})$ may fall below $1$ when $\alpha$ is greater than $0.31$.

\begin{theorem}\label{thm:sizeLeave1Out}
	If $\alpha \in (0,0.31]$ and Assumptions \ref{ass:Leave1Out}, \ref{ass:iid}, \ref{ass:reg}, and \ref{ass:asympLeave1Out} hold, then, under $H_0$,
	\begin{align}
	\limsup_{n \rightarrow \infty}\, \Pr\left( F >  \hat c_\alpha \right) \le \alpha.
	\end{align}
\end{theorem}

An important difference between this result and that of Theorem~\ref{thm:size} is that the asymptotic size may be smaller than desired, which can happen when leave-three-out fails for a large fraction of possible triples. When such conservatism materializes,
there will be a corresponding loss in power relative to the result in Theorem~\ref{thm:power}. Otherwise, the power properties are analogous to those reported in Theorem~\ref{thm:power} and we therefore omit a formal result. 

\begin{rem}\label{rem:invariance}
	Before turning to a study of the finite sample performance of the proposed test, we describe an adjustment to the test which is based on finite sample considerations. This adjustment is to rely on demeaned outcome variables in the definitions of $\hat E_{\cal F}$, $	\hat V_{\cal F}$, and $\hat{\boldsymbol w}$.
  The benefit of relying on demeaned outcomes is that it makes the critical value invariant to the location of the outcomes. On the other hand, this adjustment removes the exact unbiasedness used to motivate the estimators of ${\E}_0[ {\cal F} ]$ and ${\V}_{0}\big[\mathcal{F} - \hat{E}_{\cal F}\big]$. However, one can show that the biases introduced by demeaning 
	vanish at a rate that ensures asymptotic validity. Therefore, we deem the gained location invariance sufficiently desirable that we are willing to introduce a small finite sample bias to achieve it. We refer to the Appendix for exact mathematical details but note that this adjustment is used {in the simulations that follow; we also probe the version that uses non-demeaned outcomes (see the end of Section~\ref{sec:sim}).}
\end{rem}

\section{Simulation evidence}\label{sec:sim}

{This section documents finite sample performance of the leave-out test and compares it with that of benchmark tests that could be used by a researcher in the present context:}

\begin{enumerate}
{

\item The proposed leave-out test, which will be marked as LO in the resulting tables.

\item The exact F test, marked as EF, which uses critical values from the F distribution to reject when $F > q_{1-\alpha}(F_{r,n-m})$. This test has actual size equal to nominal size in finite samples under conditionally homoskedastic normal errors for any number of regressors and restrictions. It is also asymptotically valid with conditional homoskedasticity and non-normality under certain regressor homogeneity conditions \citep{anatolyev2012inference}, but not under general regressor designs \citep{calhoun2011many}.

%

\item Three Wald tests that reject when a heteroskedasticity-robust Wald statistic exceeds the $(1-\alpha)$-th quantile of a $\chi^2_r$ distribution, i.e., when $W > q_{1-\alpha}(\chi^2_r)$ for
\begin{align}
 		W = \big( \boldsymbol R\hat {\boldsymbol\beta}-\boldsymbol q \big)'\!\left( \boldsymbol R\boldsymbol S_{xx}\inverse \left(\sum\nolimits_{i=1}^n \boldsymbol x_i \boldsymbol x_i' \tilde \sigma_i^2\right) \boldsymbol S_{xx}\inverse \boldsymbol R' \right)\inverse \!\big(\boldsymbol R\hat {\boldsymbol\beta} - \boldsymbol q \big).
\end{align}
The three Wald tests differ only by how one constructs variance estimates $\{\tilde \sigma_i^2\}_{i=1}^n$, and are only palliatives:
\begin{enumerate}
 \item $W_\text{1}$ most closely corresponds to the original Wald test, but with the degrees-of-freedom adjustment \citep{mackinnon2012hetero}: $\tilde \sigma_i^2 = (y_i - \boldsymbol x_i' \hat{\boldsymbol \beta})^2 n/(n-m)$,
\item $W_\text{K}$ uses variance estimates of \cite{cattaneo2017inference},
\item $W_\text{L}$ uses leave-one-out estimates $\tilde \sigma_i^2 = \hat \sigma_i^2$ as in \eqref{eq:LOvar}.
\end{enumerate}
 Asymptotically, the Wald tests $W_\text{L}$ and $W_\text{K}$ are valid with many regressors under arbitrary heteroskedasticity but not necessarily with many restrictions, while $W_\text{1}$ is valid only with few regressors and few restrictions under arbitrary heteroskedasticity.\footnote{{That the baseline version $W_\text{1}$ is invalid with many restrictions was noticed empirically in \cite{berndt1977conflict} and shown in \cite{anatolyev2012inference} under homoskedasticity; one can hardly expect that such measures as simply altering estimation of individual variances is able to solve the matters in a more complex heteroskedastic situation.}} By their comparison with the LO test one can see Wald test's potential to control size, and how much distortions are due to its wrong structure when restrictions are many.

\item The test based on the split-sample idea of \cite{kline2018leave} is not going to be available for our simulations, because it requires regressor numerosity to be at most a half of the sample size, which is not satisfied in the simulation design.}
\end{enumerate}

\subsection{Simulation design}

The simulation setup borrows elements of \cite{mackinnon2012hetero} and adapts it to the case of many regressors as in \cite{richard2019manyboot} but with richer heterogeneity in the design. The outcome equation is
\begin{equation*}
y_{i}=\beta _{1}+\sum_{k=2}^{m}\beta _{k}x_{ik}+\varepsilon _{i},\quad i=1,\dots,n,
\end{equation*}
where data is drawn \emph{i.i.d.} across $i$. Following \cite{mackinnon2012hetero}, the sample sizes take the values $80,\ 160,\ 320,\ 640,$ and $1280$. The number of unknown coefficients is $m=0.8n$ throughout to demonstrate the validity of the proposed test even with very many regressors.  The null restricts the values of the last $r$ coefficients using $\boldsymbol R=\left[ \boldsymbol 0_{r\times ( m-r) },\boldsymbol I_{r}\right]$. We consider both a design that contains only continuous regressors and a \textit{mixed} one that also includes some discrete regressors.

In the continuous design, the regressors  $x_{i2}, \dots, x_{im}$ are products of independent standard log-normal random variables and a common multiplicative mean-unity factor drawn independently from a shifted standard uniform distribution, i.e., $0.5+u_{i}$ where $u_{i}$ is standard uniform. This common factor induces dependence among the regressors and rich heterogeneity in the statistical leverages of individual observations. For this design, we consider $r=3$ and $r=0.6n$.

When also including discrete regressors, we let $x_{i2}, \dots, x_{i,m-r}$ be as above and let the last $r$ regressors be group dummies. This mixed design corresponds to random assignment into $r+1$ groups with the last group effect removed due to the presence of an intercept in the model. The assigned group number is the integer ceiling of $(r+1)(u_i + u_i^2)/2$, where $u_i$ is the multiplicative factor used to generate dependence among the continuous regressors. By reusing $u_i$ we maintain dependence between all regressors, and by using a nonlinear transformation of $u_i$ we induce systematic variability among the $r+1$ expected group sizes. We let $r = 0.15 n$, which leads the expected group sizes to vary between $4$ and $13$ with an average group size of about $6.5$. The null corresponds to a hypothesis of equality of means across all groups.

Each regression error is a product of a standard normal random variable and an individual specific standard deviation $\sigma_i$. The standard deviation is generated by
\begin{align}
\sigma _{i}=z_{\zeta }\left( 1+s_{i}\right) ^{\zeta },\quad i=1,\dots,n,
\end{align}
where $s_{i}>0$ depends on the design and the multiplier $z_{\zeta }$ is such that the mean of $\sigma_i^2$ is unity. The parameter $\zeta \in \left[ 0,2\right] $ indexes the strength of heteroskedasticity, with $\zeta=0$ corresponding to homoskedasticity. We consider only the two extreme cases of $\zeta \in \left \{ 0,2\right \} $. In the continuous design, we let $s_i = \sum_{k=2}^{m}x_{ik}$, and in the mixed design, $s_i = \sum_{k=2}^{m-r}x_{ik} + z_u u_i$. The factor $z_u = 2r \exp(1/2)$ ensures that $s_i$ has the same mean in both designs.

Under the null, the coefficients on the continuous regressors are all equal to $\varrho$, where $\varrho$ is such that the coefficient of determination, $\mathrm{R}^{2}$, equals $0.16$. The coefficients on the included group dummies are zero, which correspond to the null of equality across all groups. The intercept is chosen such that the mean of the outcomes is unity. For the continuous design this yields an intercept of $1-(m-1)\varrho \exp(1/2)$, while the intercept is $1-(m-r-1)\varrho \exp(1/2)$ in the mixed design. With these parameter values, the null is $(\beta_{m-r},\dots,\beta_r)' =\boldsymbol q$, where $\boldsymbol q=(\varrho,\dots,\varrho)'\in \R^r$ in the continuous design, and $\boldsymbol q = (0,\dots,0)' \in \R^r$ in the mixed design.

To document power properties, we consider both a sparse and dense deviations from the null, and focus on the settings where $r$ is proportional to $n$. In parallel to the theoretical power analysis in Section \ref{sec:asymptotia}, we consider deviations for the last $r$ coefficients that are parameterized using
\begin{align}
	(\beta_{m-r},\dots,\beta_m)' =\boldsymbol q + \left(\boldsymbol R \E[\boldsymbol S_{xx}]\inverse \boldsymbol R'\right)^{1/2} \boldsymbol\delta , \label{eq:dev}
\end{align}
where we use the lower triangular square-root matrix. This choice of square-root implies that the alternative is sparse when only the last few entries of $\boldsymbol\delta$ are non-zero. As shown in Section~\ref{sec:asymptotia}, asymptotic power is governed by the norm of $\boldsymbol\delta$ over $r^{1/4}$, but whether an alternative is fixed or local, additionally depends on the rate at which the tested coefficients are estimated. This rate is governed by $\E\boldsymbol [\boldsymbol S_{xx}]$, which is reported in the Appendix.

In the continuous design, the tested coefficients are estimated at the standard parametric rate of $n^{-1/2}$. To specify a fixed sparse alternative we therefore use ${\boldsymbol\delta} = 0.5n^{1/2} (0,\dots,0,1)' \in \R^{r}$, for which $\beta_m$ differs from the null value by approximately $0.2$ (here and hereafter, the scaling is chosen so that the power is bounded away from the size and away from unity for the sample sizes we consider). Since the norm of $\boldsymbol\delta$ grows faster than $r^{1/4}$, the power will be an increasing function of the sample size. For the dense alternative, we consider instead ${\boldsymbol\delta} = 0.5 n^{1/2} r^{-1/2} {\boldsymbol\iota}_r$ where ${\boldsymbol\iota}_r = (1,\dots,1)' \in \R^{r}$, for which all deviations between the tested coefficients and $\varrho$ shrink at the standard parametric rate of $n^{-1/2}$. Here, power is again increasing in the sample size due to numerous deviations from the null.

In the mixed design, the group effects are not estimated consistently as the group sizes are bounded. A possible fixed sparse alternative is then ${\boldsymbol\delta} = (0,\dots,0,6)' \in \R^{r}$, for which $\beta_m$ differs from the null value of zero by roughly $3$. In contrast to the continuous design, the power will decrease with sample size as the precision, with which $\beta_m$ can be estimated, does not increase with $n$. For the dense alternative, we use ${\boldsymbol\delta} = 1.5 {\boldsymbol\iota}_r$, which corresponds to a fixed alternative for \textit{every} tested coefficient. Here, the power will be increasing in $n$ due to the numerosity of deviations.

\subsection{Simulation results}

We present rejection rates based on 10000 Monte-Carlo replications and consider tests with nominal sizes of $1\%,$ $5\%$ and $10\%$. Furthermore, we report the frequency with which the proposed variance estimate $\hat{V}_{\cal F}$ is negative and therefore replaced by the upward biased and positive alternative introduced in Remark~\ref{rem:neg}. For the design that includes discrete regressors, we also report the average fraction of observations that cause a failure of leave-three-out full rank, and for which we therefore rely on an upward biased estimator of the corresponding error variance. For all sample sizes, this fraction is around $7\%$ in the mixed design, which corresponds to the percentage of observations that belong to groups of size 2 or 3. The fraction is zero in the design that only involves continuous regressors.

\begin{table}[hbtp]
	\centering
	{\footnotesize 
		\caption{Empirical size (in percent)}
		\label{tab:Size}
		\begin{threeparttable}
			\begin{widetable}{.98\columnwidth}{llcrrrrrcrrrrrcrrrrrcr}
				\toprule
				\multicolumn{2}{l}{Nominal size}
				&& \multicolumn{5}{c}{$1\%$}	
				&& \multicolumn{5}{c}{$5\%$}
				&& \multicolumn{5}{c}{$10\%$}  \\
				\cmidrule(lr){4-8} \cmidrule(lr){10-14} \cmidrule(lr){16-20}
				\multicolumn{2}{l}{Test} &
				&LO&EF&$\text{W}_\text{1}$&$\text{W}_\text{K}$&$\text{W}_\text{L}$
				&&LO&EF&$\text{W}_\text{1}$&$\text{W}_\text{K}$&$\text{W}_\text{L}$
				&&LO&EF&$\text{W}_\text{1}$&$\text{W}_\text{K}$&$\text{W}_\text{L}$&
				& \textsc{neg} \\
				\midrule
				&&& \multicolumn{17}{c}{{Homoskedasticity}} \\
				\cmidrule(lr){1-22}
				\multicolumn{2}{l}{Continuous design} \\
				\cmidrule(lr){1-2}
				$n=80$    &$r=3$&& $2$ & $1$ & $4$ & $14$& $11$&&    $7$ & $5$ &$10$ & $19$ & $18$    && $12$   & $10$ & $15$ & $23$ & $23$ && $8.7$ \\
				$n=160$  &$r=3$ && $1$ & $1$ & $2$& $14$ & $7$ &&    $6$ & $5$ & $6$ & $19$ & $14$    && $12$ & $10$ & $10$ & $23$ & $20$ && $2.0$ \\
				$n=320$  &$r=3$ && $1$ & $1$ & $1$& $15$ & $4$ &&   $6$ & $5$ & $4$ & $22$ & $10$    && $11$ & $10$ & $8$ & $27$ & $15$ && $0.6$ \\
				$n=640$  &$r=3$ && $1$ & $1$ &$1$ & $15$ & $2$ &&    $6$ & $5$ & $3$ & $23$ & $8$      && $11$ & $10$ &$7$& $29$ & $13$ && $0.1$ \\
				$n=1280$&$r=3$  && $1$ & $1$ &$1$  & $9$ & $2$  &&   $6$ & $5$ & $3$ & $17$ & $7$         && $11$ & $10$ &$7$& $24$ & $12$ && $0.0$ \\
				\cmidrule(lr){1-2}
				$n=80$    &$r=48$  && $1$ & $1$ & $99$ & $54$& $16$&&    $3$ & $5$ &$100$ & $54$ & $18$    && $7$   & $10$ & $100$ & $54$ & $19$ && $19.8$ \\
				$n=160$  &$r=96$  && $2$ & $1$ & $100$ & $54$& $20$&&    $5$ & $5$ &$100$ & $54$ & $21$   && $10$   & $10$ & $100$ & $54$ & $22$ && $6.4$ \\
				$n=320$  &$r=192$&& $2$ & $1$ & $100$ & $53$& $20$&&    $6$ & $5$ &$100$ & $54$ & $21$    && $11$   & $10$ & $100$ & $54$ & $21$ && $1.5$ \\
				$n=640$  &$r=384$&& $1$ & $1$ & $100$ & $54$& $22$&&    $6$ & $5$ &$100$ & $54$ & $23$    && $11$   & $10$ & $100$ & $54$ & $23$ && $0.3$ \\
				$n=1280$&$r=768$&& $1$ & $1$ & $100$ & $53$& $22$&&    $5$ & $5$ &$100$ & $53$ & $23$    && $11$   & $10$ & $100$ & $53$ & $23$ && $0.0$ \\
				\cmidrule(lr){1-2}
				\multicolumn{2}{l}{Mixed design} \\
				\cmidrule(lr){1-2}
				$n=80$    & $r=12$ && $2$ & $1$ & $21$ & $25$& $19$&&    $7$ & $5$ &$33$ & $29$ & $23$    && $12$   & $10$ & $41$ & $31$ & $26$ && $4.8$ \\
				$n=160$  & $r=24$ && $1$ & $1$ & $25$ & $27$& $23$&&    $6$ & $5$ &$38$ & $30$ & $27$    && $12$   & $10$ & $47$ & $31$ & $29$ && $0.4$ \\
				$n=320$  & $r=48$ && $1$ & $1$ & $33$ & $29$& $29$&&    $5$ & $5$ &$49$ & $31$ & $32$    && $11$   & $10$ & $58$ & $32$ & $34$ && $0.1$ \\
				$n=640$  & $r=96$ && $1$ & $1$ & $47$ & $31$& $32$&&    $5$ & $5$ &$64$ & $33$ & $35$    && $11$   & $10$ & $72$ & $33$ & $36$ && $0.0$ \\
				$n=1280$&$r=192$ && $1$ & $1$ & $69$  & $31$ & $35$  &&   $5$ & $5$ & $82$ & $33$ & $37$         && $10$ & $10$ &$87$& $33$ & $38$ && $  0.0$ \\
				\midrule
				&&& \multicolumn{17}{c}{{Heteroskedasticity}} \\
				\cmidrule(lr){1-22}
				\multicolumn{2}{l}{Continuous design} \\
				\cmidrule(lr){1-2}
				$n=80$    &$r=3$&& $2$ & $3$ & $7$ & $17$& $10$&&    $6$ & $10$ &$15$ & $22$ & $16$    && $11$   & $17$ & $22$ & $26$ & $21$ && $10.3$ \\
				$n=160$  &$r=3$ && $2$ & $2$ & $4$& $14$ & $7$ &&    $6$ & $8$ & $10$ & $20$ & $14$    && $12$ & $15$ & $16$ & $24$ & $19$ && $2.3$ \\
				$n=320$  &$r=3$ && $1$ & $2$ & $2$& $16$ & $4$ &&    $6$ & $8$ & $8$ & $23$ & $10$    && $12$ & $15$ & $14$ & $28$ & $16$ && $0.6$ \\
				$n=640$  &$r=3$ && $1$ & $2$ &$2$ & $15$ & $2$ &&    $6$ & $8$ & $7$ & $23$ & $8$      && $11$ & $14$ &$12$& $29$ & $13$ && $0.3$ \\
				$n=1280$&$r=3$  && $1$ & $2$ &$1$  & $9$ & $2$  &&    $5$ & $8$ & $6$ & $17$ & $6$         && $11$ & $15$ &$12$& $24$ & $12$ && $0.0$ \\
				\cmidrule(lr){1-2}
				$n=80$    &$r=48$  && $1$ & $22$ & $100$ & $52$& $13$&&    $5$ & $47$ &$100$ & $52$ & $15$    && $9$   & $61$ & $100$ & $52$ & $16$ && $12.8$ \\
				$n=160$  &$r=96$  && $1$ & $32$ & $100$ & $50$& $15$&&    $5$ & $61$ &$100$ & $50$ & $16$   && $11$   & $75$ & $100$ & $51$ & $17$ && $3.7$ \\
				$n=320$  &$r=192$&& $1$ & $56$ & $100$ & $49$& $17$&&    $6$ & $81$ &$100$ & $49$ & $18$    && $11$   & $90$ & $100$ & $49$ & $19$ && $0.9$ \\
				$n=640$  &$r=384$&& $1$ & $86$ & $100$ & $49$& $18$&&    $5$ & $96$ &$100$ & $49$ & $19$    && $11$   & $99$ & $100$ & $49$ & $20$ && $0.1$ \\
				$n=1280$&$r=768$&& $1$ & $99$ & $100$  & $48$ & $19$  &&   $5$ & $100$ & $100$ & $48$ & $19$         && $11$ & $100$ &$100$& $48$ & $20$ && $0.0$ \\
				\cmidrule(lr){1-2}
				\multicolumn{2}{l}{Mixed design} \\
				\cmidrule(lr){1-2}
				$n=80$    & $r=12$ && $1$ & $5$ & $38$ & $27$& $14$&&    $6$ & $17$ &$51$ & $30$ & $18$    && $12$   & $26$ & $58$ & $32$ & $21$ && $4.8$ \\
				$n=160$  & $r=24$ && $1$ & $9$ & $49$ & $28$& $19$&&    $6$ & $24$ &$63$ & $30$ & $23$    && $13$   & $35$ & $71$ & $31$ & $25$ && $0.5$ \\
				$n=320$  & $r=48$ && $1$ & $14$ & $67$ & $29$& $22$&&    $5$ & $33$ &$80$ & $31$ & $25$    && $11$   & $46$ & $85$ & $32$ & $27$ && $0.1$ \\
				$n=640$  & $r=96$ && $1$ & $28$ & $89$ & $31$& $26$&&    $5$ & $52$ &$95$ & $32$ & $29$    && $11$   & $65$ & $97$ & $33$ & $30$ && $0.0$ \\
				$n=1280$&$r=192$ && $1$ & $52$ & $99$   & $33$ & $30$  &&   $5$ & $75$ & $100$ & $34$ & $31$         && $10$ & $84$ &$100$& $34$ & $32$ && $0.0$ \\
				\bottomrule
			\end{widetable}
			\begin{tablenotes}
				\scriptsize 
				\item \textsc{NOTE}:  LO: leave-out test, EF: exact F test, $\text{W}_\text{1}$: heteroskedastic Wald test with degrees-of-freedom correction, $\text{W}_\text{K}$: heteroskedastic Wald test with \cite{cattaneo2017inference} correction, $\text{W}_\text{L}$: heteroskedastic Wald test with \cite{kline2018leave} correction; \textsc{neg}: fraction of negative variance estimates for LO (in percent). Results from 10000 Monte-Carlo replications.
			\end{tablenotes}
		\end{threeparttable}
	}
\end{table}

Table~\ref{tab:Size} contains the actual rejection rates under the null for both the continuous and mixed designs. In settings with many regressors and restrictions, the considered versions of the  ``heteroskedasticity-robust" Wald test fail to control size irrespective of the design, presence of heteroskedasticity, and nominal size. The failure of the conventional Wald test, $W_1$, is spectacular, with type I error rates close to one for the continuous design, but the two versions that are robust to many regressors, $W_K$ and $W_L$, also exhibit size well above the nominal level. With few restrictions, the Wald tests show a more moderate inability to match actual size with nominal size, and the table suggests that the leave-one-out version, $W_L$, can control size in samples that are somewhat larger than considered here. Under homoskedasticity, the table reports that the exact F test indeed has exact size. However, in the heteroskedastic environments with many restrictions the exact F test is oversized with a type I error rate that approaches unity as the sample size increases.

By contrast, the proposed leave-out test exhibits nearly flawless size control as it is oversized by at most one percent across nearly all designs, nominal sizes, and whether heteroskedasticity is present or not. In the smallest sample for the continuous design, the test is somewhat conservative, presumably due to the relatively high rate of negative variance estimates ($20\%$ with homoskedasticity and $13\%$ with heteroskedasticty) that are replaced by a strongly upward biased alternative. This rate diminishes quickly with sample size, and the fraction of negative variance estimates is already essentially zero in samples with 640 observations and 512 regressors. In the mixed design, negative variance estimates are even less prevalent, potentially due to the fact that the test uses some upward biased variance estimators for $7\%$ of observations. Perhaps somewhat surprisingly, having $7\%$ of observations causing failure of leave-three-out is not sufficient to bring about any discernible conservativeness in the leave-out test for this design.

%

Table~\ref{tab:Power} contains simulated rejection rates for the continuous and mixed designs under alternatives where the parameters deviate from their null values in one of two ways -- either one tested coefficient deviates (sparse) or all tested coefficients deviate (dense). The table reports these power figures for tests with a nominal size of $5\%$ and $10\%$ that also control the size well, i.e., the LO and exact F tests under homoskedasticity and the LO test under heteroskedasticity.

\begin{table}[tb!]
	\centering
	{\footnotesize 
		\caption{Empirical power (in percent) corresponding to 5\% and 10\% size}
		\label{tab:Power}
		\begin{threeparttable}
			\begin{widetable}{.98\columnwidth}{llcrrrrrrrrrrrr}
				\toprule
				&&& \multicolumn{8}{c}{{Homoskedasticity}}
				& \multicolumn{4}{c}{{Heteroskedasticity}} \\
				\cmidrule(lr){4-11} \cmidrule(lr){12-15}
				Deviation &  &
				&  \multicolumn{4}{c}{Sparse} & \multicolumn{4}{c}{Dense}
				& \multicolumn{2}{c}{Sparse} & \multicolumn{2}{c}{Dense} \\
				\cmidrule(lr){4-7} \cmidrule(lr){8-11} \cmidrule(lr){12-13} \cmidrule(lr){14-15}
				\multicolumn{2}{l}{Nominal size} &
				& \multicolumn{2}{c}{$5\%$} & \multicolumn{2}{c}{$10\%$}
				& \multicolumn{2}{c}{$5\%$} & \multicolumn{2}{c}{$10\%$}
				& \multicolumn{1}{c}{$5\%$} & \multicolumn{1}{c}{$10\%$}
				& \multicolumn{1}{c}{$5\%$} & \multicolumn{1}{c}{$10\%$}  \\
				\cmidrule(lr){4-5} \cmidrule(lr){6-7} \cmidrule(lr){8-9} \cmidrule(lr){10-11} \cmidrule(lr){12-13} \cmidrule(lr){14-15}
				Test &  &  &  LO & EF  & LO & EF & LO & EF  & LO & EF &  LO &  LO &  LO &  LO  \\
				\midrule
				\multicolumn{2}{l}{Continuous design} \\
				\cmidrule(lr){1-2}
				{$n=80$}  & $r=48$   &&    $6$   & $15$ &$12$ &$25$&    $5$   & $15$ &$10$&$25$&       $10$   &$18$&        $7$ & $14$\\
				{$n=160$}& $r=96$   &&    $16$ & $23$&$26$ & $34$ &    $12$ & $21$ &$22$&$34$&       $20$ & $32$ &       $17$ & $29$\\
				{$n=320$}& $r=192$ &&    $29$ & $35$& $43$ &$48$ &    $26$ & $36$ &$39$& $51$ &      $31$ & $45$ &      $29$ & $44$\\
				{$n=640$}& $r=384$ &&    $49$ & $55$& $63$ &$69$ &    $44$ & $57$ &$58$& $71$ &      $52$ & $66$ &      $49$& $64$\\
				{$n=1280$}&$r=768$&&    $74$ & $80$ &$84$ & $88$ &    $68$ & $84$ &$81$& $92$ &     $76$ & $86$ &      $74$ & $85$ \\
				\cmidrule(lr){1-2}
				\multicolumn{2}{l}{Mixed design} \\
				\cmidrule(lr){1-2}
				{$n=80$}  & $r=12$   && $18$ & $23$&$30$&$36$&          $17$ & $19$ &$29$& $30$ &          $24$ & $38$&          $23$ & $37$\\
				{$n=160$}  & $r=24$ && $18$ & $18$&$29$& $28$ &          $27$ & $28$ &$41$& $41$ &          $24$ & $35$ &          $34$ &$49$\\
				{$n=320$}  & $r=48$ && $13$ & $13$&$22$& $22$ &          $40$ & $42$ &$54$& $56$ &          $16$ & $27$ &          $48$ & $64$ \\
				{$n=640$}  & $r=96$ && $10$ & $10$&$18$& $18$ &          $60$ & $65$ &$73$& $77$ &          $11$ & $20$ &          $70$ & $82$ \\
				{$n=1280$}&$r=192$&& $8$ & $8$ & $16$ & $15$ &         $87$ &$91$ & $94$ &$95$&          $9$ & $17$ &          $92$ & $97$ \\
				\bottomrule
			\end{widetable}
			\begin{tablenotes}
				\scriptsize 
				\item \textsc{NOTE}:  LO: leave-out test, EF: exact F test. Results from 10000 Monte-Carlo replications.
			\end{tablenotes}
		\end{threeparttable}
	}
\end{table}

For the continuous design, the power of the tests increases from slightly above nominal size to somewhat below unity as the number of observations increases from $80$ to $1280$. This pattern largely holds irrespective of the type of deviation and presence of heteroskedasticity, although the LO test is a bit more responsive to sparse deviations than to dense ones. Along this stretch of the power curve, the LO test exhibits a power loss that varies between $4$ and $16$ percentage points when compared to the exact F test, and in relative terms, this gap in power shrinks as the sample size grows. Given that the number of tested restrictions in this setting is above half of the sample size, we conjecture that these figures are towards the high end of the power loss that a typical practitioner would incur in order to be robust with respect to heteroskedasticity.

In the mixed design, the fixed dense alternative exhibits similar power figures as in the continuous design, while the fixed sparse deviation generates a power function that decreases with sample size. The reason for the latter is, as discussed in the previous subsection, that the deviating group effect is not estimated more precisely as additional groups are added to the data. Upon comparison of the LO and exact F tests, we see that the differences in the power figures are only $0$--$7$ percentage points. In light of Remark~\ref{rem:power}, which explains that there is no power difference between the LO and exact F tests when ${r}/{n}$ is small, it is natural to attribute this almost non-existent power loss to the fact that there are four times fewer tested restrictions in this mixed design than in the continuous one.

We have also run additional simulation experiments with our baseline continuous regressor design, where we track the impact of the relative numerosities of regressors and restrictions $r/m$ and $m/n$ and of the coefficient of determination $\mathrm{R}^{2}$ on the positivity failure rate of $\hat{V}_{\cal F}$, the empirical size of the test, and its empirical power. The two numerosity ratios show the severity of deviations from the standard regression testing setup, while the coefficient of determination summarizes the magnitude of regression coefficients relative to the size of error variances. The results are relegated to the Appendix (Tables A1 and A2), and here we give a brief summary. The general observation is that the percentage of negative $\hat{V}_{\cal F}$ positively varies with all three parameters, with $\mathrm{R}^{2}$ having most pronounced impact and the ratio $r/m$ having smallest impact. The percentage, however, is still kept within 0.1-0.2\% for all combinations when $n=640$, and the negativity issue is practically non-existent when $n=1280$.
Next, while the three parameters do affect some of the wrongly sized tests from the existing literature, they do not influence the actual empirical size of our proposal, except for minor variation in very small samples. The empirical power of the proposed test, however, is non-trivially affected by all the three parameters, whose higher values imply somewhat smaller power. The coefficient of determination, in particular, has such an effect because a higher signal relative to noise increases the variability of the individual error variance estimators relative to their targets, and so the power tends to be negatively affected by large(r) coefficients. The numerosity ratios also have a negative effect on power because the signal gets dispersed across a larger number of regressors or restrictions as the ratios increase, which naturally reduces power. These tendencies are shared by the EF test when it is appropriately sized.

{Finally, we have examined the differences that result from the use of non-demeaned outcomes when estimating individual variances and their products (see Remark \ref{rem:invariance}). The general impression from those simulations is, first, the use of non-demeaned outcomes makes size control less stable; in particular, for smaller sample sizes, the LO test is undersized. Second, it seriously decreases power at all sample sizes. In practice, we therefore recommend exploiting the version with demeaned outcomes.}

\section{Concluding remarks}\label{sec:conc}

This paper develops an inference method for use in a linear regression with conditional heteroskedasticity where the objective is to test a hypothesis that imposes many linear restrictions on the regression coefficients. The proposed test rejects the null hypothesis if the conventional F statistic exceeds a linearly transformed quantile from the F-bar distribution. The central challenges for construction of the test is estimation of individual error variances and their products, which requires new ideas when the number of regressors is large. We overcome these challenges by using the idea of leaving up to three observations out when estimating individual error variances and their products. In some samples the variance estimate used for rescaling of the critical value may either be negative or cease to exist due to the presence of many discrete regressors. For both of these issues, we propose an automatic adjustment that relies on intentionally upward biased estimators which in turn leaves the resulting test somewhat conservative. Simulation experiments show that the test controls size in small samples, even in strongly heteroskedastic environments, and only exhibits very limited adjustment-induced conservativeness. The simulations additionally illustrate good power properties that signal a manageable cost in power from relying on a test that is robust to heteroskedasticity and many restrictions.

Bootstrapping and closely related resampling methods are often advocated as automatic approaches for construction of critical values. However, in the context of linear regression with proportionality between the number of regressors and sample size, multiple papers \citep{bickel1983bootstrapping,elkaroui2018hidimboot,cattaneo2017inference} demonstrate invalidity of standard bootstrap schemes even when inferences are made on a single regression coefficient. Under additional assumptions of homoskedasticity and restrictions on the design, \cite{elkaroui2018hidimboot} and \cite{richard2019manyboot} show that problem-specific corrections to bootstrap methods can restore validity. We leave it to future research to determine whether bootstrap or other resampling methods can be corrected to ensure validity in our context of a heteroskedastic regression with many regressors and tested restrictions.\looseness=-1

\bibliographystyle{chicago}
\bibliography{lit}

\begin{appendices}
	
\section{Leave-out test}

{
\subsection{F-bar distribution}

First, when the entries of $\boldsymbol w$ are all equal and thus equal to $1/r$, we have
$\sum_{\ell=1}^r w_\ell Z_\ell=r^{-1} \sum_{\ell=1}^r  Z_\ell.$
Because $Z_1,\dots,Z_r$ are independent $\chi^2_{1}$, we have $\sum_{\ell=1}^r  Z_\ell\overset{d}{=} \chi^2_{r}$ that is independent of $Z_0$, and thus
\begin{align}
\bar F_{{\boldsymbol w},df} \overset{d}{=} \frac{\sum_{\ell=1}^r w_\ell Z_\ell}{Z_{0}/df  }\overset{d}{=} \frac{df}{r} \frac{\chi^2_{r}}{\chi^2_{df}}\overset{d}{=}F_{r,df},
\end{align}
so Snedecor's F distribution is a special case.

Second, as $df \rightarrow \infty$, by the law of large numbers, $Z_{0}/df\overset{d}{=} df^{-1}\sum_{\ell=1}^{df} z_\ell^2\overset{p}{\rightarrow} \E[z_\ell^2]=1$, where $z_1,\dots,z_{df}$ are independent standard normals. Hence, $\bar F_{{\boldsymbol w},df} \overset{d}{\rightarrow} \bar \chi^2_{\boldsymbol w}$, and so the limiting case of $\bar F_{{\boldsymbol w},df}$ when $df \rightarrow \infty$ is $\bar \chi^2_{\boldsymbol w}$.

Third, note that $\sum_{\ell=1}^r w_\ell Z_\ell$ converges in probability to its expectation of $\sum_{\ell=1}^r w_\ell =1$ by Kolmogorov's law of large numbers for sums of independent heterogeneous random variables. Then,
\begin{align}
\bar F_{{\boldsymbol w},df} &\overset{d}{=} \frac{1+\big(\sum_{\ell=1}^r w_\ell Z_\ell-1\big)}{1+(Z_{0}/df-1)  }\\
& = \left(1+ \big(\sum_{\ell=1}^r w_\ell Z_\ell-1\big)\right)
\left(1- \big(\sum_{\ell=1}^{df} \frac{z_\ell^2}{df}-1\big)(1 + o_p(1))\right) \\
& = 1+\left(\sum_{\ell=1}^r w_\ell Z_\ell -\sum_{\ell=1}^{df}\frac{z_\ell^2}{df}\right)(1 + o_p(1)).
\end{align}
Now, the expression in the first pair of brackets
is distributed as $\bar \chi^2_{\boldsymbol w,-(df^{-1},\dots,df^{-1})}$, with mean $\sum_{\ell=1}^r w_\ell-\sum_{\ell=1}^{df}df^{-1}=1-1=0$ and variance $\sum_{\ell=1}^r 2 w^2_\ell+\sum_{\ell=1}^{df}2 df^{-2}=2\sum_{\ell=1}^r w_\ell^2 +2/df.$ Using Lyapounov's central limit theorem for sums of independent heterogeneous random variables, we obtain that
\begin{align}
\frac{\bar F_{{\boldsymbol w},df}-1}{\sqrt{2\sum_{\ell=1}^r w_\ell^2 +2/df}}\overset{d}{\rightarrow} N(0,1),
\end{align}
and \eqref{eq:normlim} follows.

}

\subsection{Leave-out algebra}\label{app:leavealg}

	For an arbitrary triple $(i,j,k)$ with $i \neq j \neq k \neq i$, the following shows that $\sum_{\ell \neq i,j,k} {\boldsymbol x}_\ell {\boldsymbol x}_\ell'$ is invertible if and only if $D_{ijk} > 0$. By SMW it suffices to show that $D_{ijk} > 0$ is equivalent to $1 - {\boldsymbol x}_i'\!\left(\sum_{\ell \neq j,k} {\boldsymbol x}_\ell {\boldsymbol x}_\ell'\right)\inverse {\boldsymbol x}_i > 0$ when $D_{jk}>0$. Now, we have
	\begin{align}
		\!\left(\sum_{\ell \neq j,k} {\boldsymbol x}_\ell {\boldsymbol x}_\ell'\right)\inverse
		&= \boldsymbol S_{xx}\inverse + \boldsymbol S_{xx}\inverse \begin{pmatrix}	{\boldsymbol x}_j' \\ {\boldsymbol x}_k' \end{pmatrix}' \begin{bmatrix} M_{jj} & M_{jk} \\ M_{jk} & M_{kk} \end{bmatrix}\inverse  \begin{pmatrix}	{\boldsymbol x}_j' \\ {\boldsymbol x}_k' \end{pmatrix} \boldsymbol S_{xx}\inverse,
		\label{eq:SMW2}
		\shortintertext{thus}
		1 - {\boldsymbol x}_i'\!\left(\sum_{\ell \neq j,k} {\boldsymbol x}_\ell {\boldsymbol x}_\ell'\right)\inverse {\boldsymbol x}_i
		&= M_{ii} - \begin{pmatrix}	 M_{ij} \\ M_{ik} \end{pmatrix}' \begin{bmatrix} M_{jj} & M_{jk} \\ M_{jk} & M_{kk} \end{bmatrix}\inverse  \begin{pmatrix}	M_{ij} \\ M_{ik}  \end{pmatrix}  = \frac{D_{ijk}}{D_{jk}}.
		\label{eq:SMW3}
	\end{align}
	Therefore, $D_{ijk} >0$ if and only if $1 - {\boldsymbol x}_i'\!\left(\sum_{\ell \neq j,k} {\boldsymbol x}_\ell {\boldsymbol x}_\ell'\right)\inverse {\boldsymbol x}_i > 0$.

	\subsection{Location estimator}\label{app:centF}

	The following shows that $\E_0[ {\cal F} ] = \sum_{i=1}^n B_{ii} \sigma_i^2$ and that $\E\big[  \hat{E}_{\cal F} \big] = {\E}_{0}[{\cal F} ]$ which yields that ${\cal F}-\hat{E}_{\cal F}$ is centered at zero under the null. When $H_0$ holds so that $\boldsymbol R{\boldsymbol\beta}=\boldsymbol q$, we have $\boldsymbol R\hat {\boldsymbol\beta} =\boldsymbol q + \boldsymbol R\boldsymbol S_{xx}\inverse \sum_{i=1}^n {\boldsymbol x}_i \varepsilon_i$. Inserting this relationship into the definition of $\cal F$ yields
	\begin{align}
		{\cal F} =\!\left(\boldsymbol R\boldsymbol S_{xx}\inverse \sum_{i=1}^n {\boldsymbol x}_i \varepsilon_i \right)' \!\left(\boldsymbol R \boldsymbol S_{xx}\inverse \boldsymbol R'\right)\inverse \!\left(\boldsymbol R\boldsymbol S_{xx}\inverse \sum_{i=1}^n {\boldsymbol x}_i \varepsilon_i \right)\! = \sum_{i=1}^n \sum_{j=1}^n B_{ij} \varepsilon_i \varepsilon_j ,
	\end{align}
	where  $B_{ij} = {\boldsymbol x}_i'\boldsymbol S_{xx}\inverse \boldsymbol R'\!\left( \boldsymbol R\boldsymbol S_{xx}\inverse \boldsymbol R' \right)\inverse \! \boldsymbol R \boldsymbol S_{xx}\inverse {\boldsymbol x}_j$. Independent sampling and exogenous regressors yield $\E[\varepsilon_i \varepsilon_j ]=0$ whenever $i \neq j$, so
	\begin{align}
		\E_0[ {\cal F} ] = \E_0\!\left[ \sum_{i=1}^n \sum_{j=1}^n B_{ij} \varepsilon_i \varepsilon_j  \right]\! = \sum_{i=1}^n B_{ii} \sigma_i^2.
	\end{align}
	The matrix $\boldsymbol B = (B_{ij})$ is a projection matrix, so it is symmetric and satisfies $r = {\rm tr}(\boldsymbol I_r) = {\rm tr}(\boldsymbol B) = \sum_{i=1}^n B_{ii}$ as claimed in the main text. It follows from \cite[][Lemma 1]{kline2018leave} that $\E[\hat \sigma_i^2 ]= \sigma_i^2$, so $\E[ \sum_{i=1}^n B_{ii} \hat \sigma_i^2 ] = \E_0[ {\cal F} ]$ since $B_{11},\dots,B_{nn}$ are known.

	Next, we show that the conditional variance of ${\cal F}-\hat{E}_{\cal F}$ satisfies the relation given in \eqref{eq:NullVar}. Since $\hat \sigma_i^2 = {y_i(y_i-{\boldsymbol x}_i'\hat {\boldsymbol\beta})}/{M_{ii}} = \sum_{j=1}^n \frac{M_{ij}}{M_{ii}} y_i \varepsilon_j$, we have that, under $H_0$,
	\begin{align}\label{eq:Ustat}
		{\cal F}-\hat{E}_{\cal F}\, &= \sum_{i=1}^n \sum_{j=1}^n B_{ij} \varepsilon_i \varepsilon_j - \sum_{i=1}^n \sum_{j=1}^n \tfrac{B_{ii}}{M_{ii}} M_{ij} y_i \varepsilon_j \\
		&=\sum_{i=1}^n \sum_{j=1}^n\!\left( B_{ij} -\tfrac{B_{ii}}{M_{ii}} M_{ij} \right)\! \varepsilon_i \varepsilon_j - \sum_{i=1}^n \sum_{j=1}^n \tfrac{B_{ii}}{M_{ii}} M_{ij} {\boldsymbol x}_i'{\boldsymbol\beta} \varepsilon_j \\
		&=\sum_{i=1}^n \sum_{j\neq i} C_{ij} \varepsilon_i \varepsilon_j - \sum_{j=1}^n\!\left( \sum\nolimits_{i=1}^n\!\left(\tfrac{B_{ii}}{M_{ii}}-\tfrac{B_{jj}}{M_{jj}}\right)\! M_{ij} {\boldsymbol x}_i'{\boldsymbol\beta}  \right)\! \varepsilon_j ,
	\end{align}
	where $C_{ij} = B_{ij} - \frac{M_{ij}}{2}\!\left(\tfrac{B_{jj}}{M_{jj}}+ \tfrac{B_{ii}}{M_{ii}}\right)\!$ is a set of symmetric weights, i.e., $C_{ij}=C_{ji}$. Note that we have subtracted off $n$ zeroes in the form of $\varepsilon_j \frac{B_{jj}}{M_{jj}} \sum_{i=1}^n M_{ij} {\boldsymbol x}_i'{\boldsymbol\beta}$, which exploits the identity $\sum_{i=1}^n M_{ij} {\boldsymbol x}_i=0$. Independent sampling yields $\E[\varepsilon_i \varepsilon_j \varepsilon_k ]=0$ whenever $i \neq j$ for any $k$, so the two components in this representation of ${\cal F}-\hat{E}_{\cal F}$ are uncorrelated. A straightforward variance calculation for each component leads to the variance expression in \eqref{eq:NullVar}:
	\begin{align}
		{\V}_{0}\!\left[{\cal F} - \hat{E}_{\cal F} \right] &= 2\sum_{i=1}^n \sum_{j \neq i} C_{ij}^2 \sigma_i^2 \sigma_j^2 + \sum_{i=1}^n \left( \sum\nolimits_{j\neq i} \left(\tfrac{B_{jj}}{M_{jj}}- \tfrac{B_{ii}}{M_{ii}}\right) M_{ij} {\boldsymbol x}_j'{\boldsymbol\beta} \right)^2 \sigma_i^2 \\
		&= \sum_{i=1}^n \sum_{j \neq i} U_{ij} \sigma_i^2 \sigma_j^2 + \sum_{i=1}^n \!\left( \sum\nolimits_{j\neq i} V_{ij} {\boldsymbol x}_j'{\boldsymbol\beta} \right)^2 \sigma_i^2,
	\end{align}
	where $U_{ij} = 2 C_{ij}^2$ and $V_{ij}=M_{ij} \!\left( \tfrac{B_{ii}}{M_{ii}}-\tfrac{B_{jj}}{M_{jj}}\right)$.

	\subsection{Variance estimator}\label{app:varest}

	First, we show that $\widehat{\sigma_i^2 \sigma_j^2}=\widehat{\sigma_j^2 \sigma_i^2}$. To establish this equality we introduce some notation used to describe $\hat \sigma_{i,-j}^2$ and $\hat \sigma_{i,-jk}^2$. Define
	\begin{align}\label{def:checkM}
    \check M_{ij,-i} &= \frac{M_{ij}}{M_{ii}},\\
		\check M_{ik,-ij} &= \frac{M_{ik} - M_{ij} \check M_{jk,-j}}{D_{ij}/M_{jj}},
		\shortintertext{and}
		\check M_{i\ell,-ijk} &= \frac{M_{i\ell} - M_{ij}\check M_{j\ell,-jk} - M_{ik} \check M_{k\ell,-jk}}{D_{ijk}/D_{jk}},
	\end{align}
	where the indices following the commas are all different and their ordering is irrelevant (note that $\check M_{ik,-ij}$ was also introduced in the main text). In addition, we will also at times write $\check M_{i\ell,-ijj}$ for $\check M_{i\ell,-ij}$. With these definitions we now have
	\begin{align}
    \hat \sigma_{i}^2 &= y_i\big(y_i - {\boldsymbol x}_i'\hat {\boldsymbol\beta}_{-i}\big) = y_i \sum_{j=1}^n \check M_{ij,-i} y_k,\\
		\hat \sigma_{i,-j}^2 &= y_i\big(y_i - {\boldsymbol x}_i'\hat {\boldsymbol\beta}_{-ij}\big) = y_i \sum_{k=1}^n \check M_{ik,-ij} y_k,
		\shortintertext{and}
		\hat \sigma_{i,-jk}^2 &= y_i\big(y_i - {\boldsymbol x}_i'\hat {\boldsymbol\beta}_{-ijk}\big) = y_i \sum_{\ell \neq k} \check M_{i\ell,-ijk} y_\ell.
	\end{align}
	To see why these relationships hold note that $\check M_{ii,-ij}=1$, $\check M_{ij,-ij}=0$, and
	\begin{align}\label{eq:checkM}
		- {\boldsymbol x}_i'\!\left(\sum_{\ell \neq i,j} {\boldsymbol x}_\ell {\boldsymbol x}_\ell'\right)\inverse {\boldsymbol x}_k
		&= M_{ik} - \begin{pmatrix}	 M_{ii} -1 \\ M_{ij} \end{pmatrix}' \begin{bmatrix} M_{ii} & M_{ij} \\ M_{ij} & M_{jj} \end{bmatrix}\inverse  \begin{pmatrix}	M_{ik} \\ M_{jk} \end{pmatrix}
		=\check M_{ik,-ij},
	\end{align}
	where the first equality follows from \eqref{eq:SMW2}. Similarly, note that $\check M_{ii,-ijk} =1$, $\check M_{ij,-ijk} =\check M_{ik,-ijk} =0$, and use SMW and \eqref{eq:SMW3} to see that
	\begin{align}
		\!\left(\sum_{l \neq i,j,k} {\boldsymbol x}_l {\boldsymbol x}_l'\right)\inverse =\!\left(\sum_{l \neq j,k} {\boldsymbol x}_l {\boldsymbol x}_l'\right)\inverse  + \frac{\!\left(\sum_{l \neq j,k} {\boldsymbol x}_l {\boldsymbol x}_l'\right)\inverse  {\boldsymbol x}_i {\boldsymbol x}_i'\!\left(\sum_{l \neq j,k} {\boldsymbol x}_l {\boldsymbol x}_l'\right)\inverse }{D_{ijk}/D_{jk}}
	\end{align}
	which together with \eqref{eq:SMW2} yields
	\begin{align}
		- {\boldsymbol x}_i'\!\left(\sum_{l \neq i,j,k} {\boldsymbol x}_l {\boldsymbol x}_l'\right)\inverse {\boldsymbol x}_\ell
		&= -\frac{{\boldsymbol x}_i\!\left(\sum_{l \neq j,k} {\boldsymbol x}_l {\boldsymbol x}_l'\right)\inverse  {\boldsymbol x}_\ell }{D_{ijk}/D_{jk}} = \check M_{i\ell,-ijk}.
	\end{align}

	Relying on the newly defined $\check M_{ik,-ij}$ and $\check M_{i\ell,-ijk}$ we can write
	\begin{align}
		\widehat{\sigma_j^2 \sigma_i^2} = y_i y_j  \sum_{k =1}^n \sum_{\ell \neq k}  \check M_{jk,-ij} \check M_{i\ell,-ijk} y_k y_\ell
		\quad \text{and} \quad
		\widehat{\sigma_i^2 \sigma_j^2} = y_i y_j  \sum_{k =1}^n \sum_{\ell \neq k}  \check M_{ik,-ij} \check M_{j\ell,-ijk} y_k y_\ell ,
	\end{align}
	from which $\widehat{\sigma_j^2 \sigma_i^2} = \widehat{\sigma_i^2 \sigma_j^2}$ will follow if
	\begin{align}
		 \check M_{jk,-ij} \check M_{i\ell,-ijk} +  \check M_{j\ell,-ij} \check M_{ik,-ij\ell} =  \check M_{ik,-ij} \check M_{j\ell,-ijk} + \check M_{i\ell,-ij} \check M_{jk,-ij\ell} .
	\end{align}
	That this equality holds follows immediately from the observation that
	\begin{align}
		\check M_{i\ell,-ijk} = \check M_{i\ell,-ij} - \check M_{ik,-ij} \check M_{k\ell,-ijk},
	\end{align}
	which shows equality between
	\begin{align}
		\check M_{jk,-ij} \check M_{i\ell,-ijk} +  \check M_{j\ell,-ij} \check M_{ik,-ij\ell}
		&= \check M_{jk,-ij}\!\left(\check M_{i\ell,-ij} - \check M_{ik,-ij} \check M_{k\ell,-ijk}\right)\!  \\
		&\quad + \check M_{j\ell,-ij}\!\left(\check M_{ik,-ij} - \check M_{i\ell,-ij} \check M_{\ell k,-ij\ell}\right)\!
		\shortintertext{and}
		 \check M_{ik,-ij} \check M_{j\ell,-ijk} + \check M_{i\ell,-ij} \check M_{jk,-ij\ell}
		 &= \check M_{ik,-ij}\!\left(\check M_{j\ell,-ij} - \check M_{jk,-ij} \check M_{k\ell,ijk} \right)\! \\
		 &\quad +\check M_{i\ell,-ij}\!\left(\check M_{jk,-ij} - \check M_{j\ell,-ij} \check M_{\ell k,ij\ell} \right)\! .
	\end{align}
	
	Now we derive that $\hat{V}_{\cal F}$ is a conditionally unbiased estimator of the null variance given in \eqref{eq:NullVar}. That $\widehat{\sigma_i^2 \sigma_j^2}$ is conditionally unbiased for $\sigma_i^2 \sigma_j^2$ was given in the main text, so here we elaborate on the bias introduced by the second component. Note that $(y_j,y_k)$ is conditionally independent of $\hat \sigma_{i,-jk}^2$ and $\E[y_j y_k ] = {\boldsymbol x}_j'{\boldsymbol\beta} {\boldsymbol x}_k'{\boldsymbol\beta} + \sigma_j^2 \boldsymbol 1_{\{j=k\}} $ so that
	\begin{align}
		\E \!\left[ \sum_{i=1}^n \sum_{j\neq i} \sum_{k\neq i} V_{ij} y_j \cdot V_{ik} y_k \cdot \hat \sigma_{i,-jk}^2 \right]\!
		&=  \sum_{i=1}^n \sum_{j\neq i} \sum_{k\neq i} V_{ij}  V_{ik} \cdot \E\!\left[ y_k y_j  \right]\! \cdot \E\!\left[ \hat \sigma_{i,-jk}^2  \right]\! \\
		&= \sum_{i=1}^n\!\left( \sum\nolimits_{j\neq i}V_{ij} {\boldsymbol x}_j'{\boldsymbol\beta} \right)^2 \sigma_i^2
		+ \sum_{i=1}^n \sum_{j\neq i} V_{ij}^2 \sigma_j^2 \sigma_i^2.
	\end{align}
	The first component of this expectation is equal to the corresponding second part of the target variance ${\V}_{0}\!\big[{\cal F} - \hat{E}_{\cal F} \big]$, but the second component is a bias which we correct for by using $\sum_{i=1}^n \sum_{j \neq i} \!\left(U_{ij} - V_{ij}^2 \right)\! \widehat{\sigma_i^2 \sigma_j^2}$ instead of $\sum_{i=1}^n \sum_{j \neq i} U_{ij} \widehat{\sigma_i^2 \sigma_j^2}$ as an estimator of the first part in ${\V}_{0}\!\big[{\cal F} - \hat{E}_{\cal F} \big]$.

{
Now we derive the test statistic that relies on restricted variance estimates.
Denote by $\boldsymbol{S}_{-i}=\sum_{j\neq i}\boldsymbol{x}_{j}\boldsymbol{x}_{j}^{\prime }$ leave-one-out analogs of $\boldsymbol{S}_{xx}.$ The restricted LO estimates are $\widetilde{\boldsymbol{\beta }}_{-i}=\hat{\boldsymbol{\beta }}_{-i}-\boldsymbol{S}_{-i}^{-1}\boldsymbol{R}^{\prime }\big( \boldsymbol{R}\boldsymbol{S}_{-i}^{-1}\boldsymbol{R}^{\prime }\big) ^{-1}\big(\boldsymbol{R}\hat{\boldsymbol{\beta }}_{-i}-\mathbf{q}\big) ,$ the restricted LO residuals are $y_{i}-\boldsymbol{x}_{i}^{\prime }\widetilde{\boldsymbol{\beta }}_{-i}$, and the resulting restricted individual variance estimates are $\widetilde{\sigma }_{i}^{2}=y_{i}\big( y_{i}-\boldsymbol{x}_{i}^{\prime }\widetilde{\boldsymbol{\beta }}_{-i}\big) ,$ which are conditionally unbiased under $H_{0}$.

To derive the restricted estimator $\widetilde{E}_{\mathcal{F}}$ of $\mathbb{E}_{0}[\mathcal{F}]$ and null conditional variance $\mathbb{V}_{0}\big[ \mathcal{F-}\widetilde{E}_{\mathcal{F}}\big]$ of $\mathcal{F-}\widetilde{E}_{\mathcal{F}}$, note first that the restricted LO residuals are
\begin{equation*}
y_{i}-\boldsymbol{x}_{i}^{\prime }\widetilde{\boldsymbol{\beta }}_{-i}=y_{i}-\boldsymbol{x}_{i}^{\prime }\hat{\boldsymbol{\beta }}_{-i}-\sum_{j\neq
i}\Upsilon _{ij}\varepsilon _{j},
\end{equation*}
where $\Upsilon _{ij}=\boldsymbol{x}_{i}^{\prime }\boldsymbol{S}_{-i}^{-1}\boldsymbol{R}^{\prime }\big(\boldsymbol{R}\boldsymbol{S}_{-i}^{-1}\boldsymbol{R}^{\prime }\big) ^{-1}\boldsymbol{RS}_{-i}^{-1}\boldsymbol{x}_{j}$.
By the Woodbury matrix identity,
\begin{eqnarray*}
\Upsilon _{ij} &=&\frac{\boldsymbol{x}_{i}^{\prime }\boldsymbol{S}_{xx}^{-1}\boldsymbol{R}^{\prime }}{M_{ii}}\left( \boldsymbol{R}\boldsymbol{S}_{xx}^{-1}\boldsymbol{R}^{\prime }+\frac{\boldsymbol{RS}_{xx}^{-1}\boldsymbol{x}_{i}\boldsymbol{x}_{i}^{\prime }\boldsymbol{S}_{xx}^{-1}\boldsymbol{R}^{\prime }}{M_{ii}}\right) ^{-1}\boldsymbol{RS}_{xx}^{-1}\left( \boldsymbol{x}_{j}-\frac{M_{ij}}{M_{ii}}\boldsymbol{x}_{i}\right)  \\
&=&\boldsymbol{x}_{i}^{\prime }\frac{\boldsymbol{S}_{xx}^{-1}\boldsymbol{R}^{\prime }\left( \boldsymbol{R}\boldsymbol{S}_{xx}^{-1}\boldsymbol{R}^{\prime }\right) ^{-1}\boldsymbol{RS}_{xx}^{-1}}{M_{ii}+\boldsymbol{x}_{i}^{\prime }\boldsymbol{S}_{xx}^{-1}\boldsymbol{R}^{\prime }\left(\boldsymbol{R}\boldsymbol{S}_{xx}^{-1}\boldsymbol{R}^{\prime }\right) ^{-1}\boldsymbol{RS}_{xx}^{-1}\boldsymbol{x}_{i}}\left( \boldsymbol{x}_{j}-\frac{M_{ij}}{M_{ii}}\boldsymbol{x}_{i}\right).
\end{eqnarray*}
Note that $\Upsilon _{ii}=0$. As a result, the restricted variance estimates are%
\begin{equation*}
\widetilde{\sigma }_{i}^{2}=\widehat{\sigma }_{i}^{2}-\left( \boldsymbol{x}_{i}^{\prime }\boldsymbol{\beta }+\varepsilon _{i}\right) \sum_{j\neq i}\Upsilon _{ij}\varepsilon _{j},
\end{equation*}%
and the restricted estimator of $\mathbb{E}_{0}[\mathcal{F}]$ is
\begin{equation*}
\widetilde{E}_{\mathcal{F}}=\sum_{i=1}^{n}B_{ii}\widetilde{\sigma }_{i}^{2}=\widehat{E}_{\mathcal{F}}-\sum_{i=1}^{n}B_{ii}\left(\boldsymbol{x}_{i}^{\prime }\boldsymbol{\beta }+\varepsilon _{i}\right)\sum_{j\neq i}\Upsilon _{ij}\varepsilon _{j},
\end{equation*}%
so that
\begin{equation*}
\mathcal{F-}\widetilde{E}_{\mathcal{F}}=\sum_{i=1}^{n}\sum_{j\neq i}\left(C_{ij}-B_{ii}\Upsilon _{ij}\right) \varepsilon _{i}\varepsilon_{j}-\sum_{j=1}^{n}\left( \sum_{j\neq i}\left( V_{ij}+B_{ii}\Upsilon _{ij}\right)\boldsymbol{x}_{i}^{\prime }\boldsymbol{\beta }\right) \varepsilon _{j},
\end{equation*}%
and so
\begin{equation*}
\mathbb{V}_{0}\left[ \mathcal{F-}\widetilde{E}_{\mathcal{F}}\right]=\sum_{i=1}^{n}\sum_{j\neq i}\widetilde{U}_{ij}\sigma _{i}^{2}\sigma_{j}^{2}+\sum_{i=1}^{n}\left( \sum_{j\neq i}\widetilde{V}_{ij}\boldsymbol{x}_{j}^{\prime }\boldsymbol{\beta }\right) ^{2}\sigma _{i}^{2},
\end{equation*}
where $\widetilde{U}_{ij}=2\left( C_{ij}-B_{ii}\Upsilon _{ij}\right) ^{2}$ and $\widetilde{V}_{ij}=V_{ij}+B_{jj}\Upsilon _{ji}$.
The estimate $
\widetilde{V}_{\mathcal{F}}$ can be constructed similarly to $\widehat{V}_{\mathcal{F}}$ using coefficients $\widetilde{U}_{ij}$ and $\widetilde{V}_{ij}$ in place of $U_{ij}$ and $V_{ij}$ and restricted leave-three-out variance estimates $\widetilde{\sigma }_{i,-jk}^{2}=y_{i}\big(y_{i}-\boldsymbol{x}_{i}^{\prime }\widetilde{\boldsymbol{\beta }}_{-ijk}\big) $, where $\widetilde{\boldsymbol{\beta }}_{-ijk}=\hat{\boldsymbol{\beta }}_{-ijk}-\boldsymbol{S}_{-ijk}^{-1}\boldsymbol{R}^{\prime}\big( \boldsymbol{R}\boldsymbol{S}_{-ijk}^{-1}\boldsymbol{R}^{\prime
}\big) ^{-1}\big( \boldsymbol{R}\hat{\boldsymbol{\beta }}_{-ijk}-\mathbf{q}\big) $ are restricted leave-three-out parameter estimates, and $\boldsymbol{S}_{-ijk}=\sum_{\ell\neq i,j,k}\boldsymbol{x}_{\ell}\boldsymbol{x}_{\ell}^{\prime }$ are leave-three-out analogs of $\boldsymbol{S}_{xx}.$

Note that restricted estimation of error variances when the null is imposed may in fact facilitate existence of leave-three-out estimators. For example, suppose that the null's parameters are $\boldsymbol{R }=\big(\boldsymbol{0}_{m_2\times m_1},I_{m_2}\big)$ and $\mathbf{q}=\boldsymbol{0}_{m_2\times 1}$, where $m_1+m_2=m$ is the total regressor dimensionality, so that one tests for joint insignificance of the last $m_2$ parameters.
Suppose that the first $m_1$ regressors are continuously distributed, while the last $m_2$ regressors are discrete. Then, when the null is imposed, the discrete regressors do not enter the design matrix, and there is no problem with the existence of leave-out estimators.

}	
	
	\subsection{Computational remarks}\label{app:computation}

	The representation of leave-one-out residuals and individual leave-one-out variance estimators given in the main text follows immediately from \eqref{eq:SMW}. Here we derive the representation of the leave-two-out and leave-three-out residuals given in the main text and used in implementation of the testing procedure. In \eqref{eq:checkM}, we showed for $j \neq i$ that $y_i - {\boldsymbol x}_i'\hat {\boldsymbol\beta}_{-ij} = \sum_{k = j}^n \check M_{ik,-ij} y_k$ where $\check M_{ik,-ij} = \frac{M_{ik} - M_{ij} M_{jk}/M_{jj}}{D_{ij}/M_{jj}}$. Thus it follows that
	\begin{align}
		y_i - {\boldsymbol x}_i'\hat {\boldsymbol\beta}_{-ij} &= \sum_{k=1}^n \frac{M_{jj} M_{ik} - M_{ij} M_{jk}}{D_{ij}} y_k
		= \frac{M_{jj} \sum_{k=1}^n M_{ik} y_k - M_{ij}\sum_{k=1}^n M_{jk} y_k}{D_{ij}} \\
		&=  \frac{M_{jj} (y_i - {\boldsymbol x}_i'\hat {\boldsymbol\beta}) - M_{ij}(y_j - {\boldsymbol x}_j'\hat {\boldsymbol\beta})}{D_{ij}} ,
	\end{align}
	as claimed.
	
	To break the monotonicity of the constant reliance on SMW, we establish the representation of the leave-three-out residuals using blockwise inversion. For $i \neq j \neq k \neq i$, $y_i - {\boldsymbol x}_i'\hat {\boldsymbol\beta}_{-ijk}$ is the first entry of the vector
	\begin{align}
		\begin{bmatrix} M_{ii} & M_{ij} & M_{ik} \\ M_{ij} & M_{jj} & M_{jk} \\ M_{ik} & M_{jk} & M_{kk} \end{bmatrix}\inverse \begin{pmatrix} y_i - {\boldsymbol x}_i'\hat{\boldsymbol\beta} \\ y_j - {\boldsymbol x}_j'\hat{\boldsymbol\beta} \\ y_k - {\boldsymbol x}_k'\hat{\boldsymbol\beta} \end{pmatrix}
	\end{align}
	which by blockwise inversion equals
	{\small
	\begin{align}
		\underbrace{\!\left(M_{ii} - \begin{pmatrix} M_{ij} \\ M_{ik} \end{pmatrix}' \begin{bmatrix}  M_{jj} & M_{jk} \\ M_{jk} & M_{kk} \end{bmatrix}\inverse \begin{pmatrix} M_{ij} \\ M_{ik} \end{pmatrix}\right)\inverse}_{=D_{ijk}/D_{jk}} \! \Bigg[ y_i - {\boldsymbol x}_i'\hat{\boldsymbol\beta} - \begin{pmatrix} M_{ij} \\ M_{ik} \end{pmatrix}' \underbrace{\begin{bmatrix}  M_{jj} & M_{jk} \\ M_{jk} & M_{kk} \end{bmatrix}\inverse \begin{pmatrix}   y_j - {\boldsymbol x}_j'\hat{\boldsymbol\beta} \\ y_k - {\boldsymbol x}_k'\hat{\boldsymbol\beta} \end{pmatrix}}_{=\begin{pmatrix}   y_j - {\boldsymbol x}_j'\hat{\boldsymbol\beta}_{-jk} \\ y_k - {\boldsymbol x}_k'\hat{\boldsymbol\beta}_{-jk} \end{pmatrix}} \Bigg]
	\end{align}
	}%
	which in turn is the representation provided in the main text.

	\section{Asymptotic size and power}

	\subsection{Asymptotic size}\label{app:size}

	As a preliminary observation, note that $\max_i B_{ii} = O_p(\epsilon_n)$, { Assumption~\ref{ass:reg}$(i)$,} and Assumption~\ref{ass:asympLeave3Out} imply that $\max_i (\sum\nolimits_{j \neq i} V_{ij} {\boldsymbol x}_j'{\boldsymbol\beta})^2/r = o_p(1)$. This follows from the idempotency of $M$ through
	\begin{align}
	\max_i \frac{1}{r}\left(\sum\nolimits_{j \neq i} V_{ij} {\boldsymbol x}_j'{\boldsymbol\beta}\right)^2 &\le \frac{1}{r} \sum_{i=1}^n \!\left( \sum\nolimits_{j=1}^n M_{ij} \tfrac{B_{jj}}{M_{jj}} {\boldsymbol x}_j'{\boldsymbol\beta} \right)^2 \\
	&\le \max_i \frac{({\boldsymbol x}_i'{\boldsymbol\beta})^2}{M_{ii}^2} \frac{1}{r} \sum\nolimits_{j=1}^n B_{jj}^2
	\le \max_i \frac{({\boldsymbol x}_i'{\boldsymbol\beta})^2}{M_{ii}^2} \max_{j} B_{jj} = o_p(1).
	\end{align}
	Thus, we have $\max_i (\sum\nolimits_{j \neq i} V_{ij} {\boldsymbol x}_j'{\boldsymbol\beta})^2/r = o_p(1)$ under either of the two possible conditions in { Assumption~\ref{ass:reg}$(ii)$.} Similarly, we have that $\max_{i} B_{ii}/(\epsilon_n r) = o_p(1)$ under either of the two possible conditions in { Assumption~\ref{ass:reg}$(ii)$.} Finally, we will repeatedly rely on the simple bound that $\max_i \sum_{j \neq i} U_{ij} + V_{ij}^2= O( \max_{i} B_{ii})$.

	Finally, as a further motivation of the high-level condition $\max_i (\sum\nolimits_{j \neq i} V_{ij} {\boldsymbol x}_j'{\boldsymbol\beta})^2/r = o_p(1)$ we provide a simple example where it holds with $r$ proportional to $n$. This example is characterized by
	\begin{enumerate}
		\item $n/r = O(1)$ and $ \max_{i} \sum_{i=1}^n \mathbf{1}\{M_{ij} \neq 0\} = O_p\big( \epsilon_n n^{1/2}\big)$.
	\end{enumerate}
	This example focus on settings where the number of restrictions is large relative to sample size, and covers any model with group specific regressors only and maximal group sizes that grow slower than $n^{1/2}$. This is so since $M_{ij}=0$ for any two observations in different groups. Here, we have
	\begin{align}
	\max_i \frac{1}{r}\left(\sum\nolimits_{j \neq i} V_{ij} {\boldsymbol x}_j'{\boldsymbol\beta}\right)^2 &= \frac{1}{r} \max_i \!\left( \sum\nolimits_{j=1}^n M_{ij} \tfrac{B_{jj}}{M_{jj}} {\boldsymbol x}_j'{\boldsymbol\beta} \right)^2 \\ &\le \frac{1}{r} \max_i \frac{({\boldsymbol x}_i'{\boldsymbol\beta})^2}{M_{ii}^2} \left(\sum\nolimits_{j=1}^n \mathbf{1}\{M_{ij} \neq 0\}\right)^2 = o_p(1),
	\end{align}
	where the order statement use 1., { Assumption~\ref{ass:reg}$(i)$, and Assumption~\ref{ass:asympLeave3Out}.	}

	\begin{proof}[Proof of Theorem~\ref{thm:size}]
		The proof naturally separates into three parts. In the first two parts, we consider an infeaible version of the test that relies on ${\V}_{0}\!\big[{\cal F} - \hat{E}_{\cal F}\big]$ instead of $\hat{V}_{\cal F}$. The first part then establishes asymptotic size control when $r$ grows to infinity with $n$, while the second part establishes size control when $r$ is fixed in the asymptotic regime. The third part shows consistency of the proposed variance estimator, i.e., $\hat{V}_{\cal F}/{\V}_{0}\!\big[{\cal F} - \hat{E}_{\cal F}\big] \xrightarrow{p} 1$. Together, these results and the continuous mapping theorem lead to the conclusion of the theorem irrespective of how $r$ is viewed in relation to the sample size.
		
		\noindent \textbf{Asymptotic size control when $r$ is growing}
		Using \eqref{eq:Ustat} and defining the vector $\check {\boldsymbol x}_i = \sum_{j=1}^n M_{ij} \frac{B_{jj}}{M_{jj}} {\boldsymbol x}_j = - \sum_{j \neq i} V_{ij} {\boldsymbol x}_j$, we can write
		\begin{align}
			{\cal F} - \hat{E}_{\cal F} =  \sum_{i=1}^n \sum_{j\neq i} C_{ij} \varepsilon_i \varepsilon_j - \sum_{i=1}^n \check {\boldsymbol x}_i'{\boldsymbol\beta} \varepsilon_i.
		\end{align}
		Under { Assumption~\ref{ass:iid},} it follows {\citep[see][Lemma B.1 and its proof]{kline2018leave}} that ${\cal F} - \hat{E}_{\cal F}$ scaled down by ${\V}_{0}\!\big[{\cal F} - \hat{E}_{\cal F}\big]^{1/2}$ is asymptotically standard normal provided that
		\begin{align}
			(a) \ \frac{\text{trace}(\boldsymbol C^4)}{{\V}_{0}\!\big[{\cal F} -\hat{E}_{\cal F} \big]^2}=o_p(1) \quad \text{and} \quad (b) \ \frac{\max_i (\check {\boldsymbol x}_i'{\boldsymbol\beta})^2}{{\V}_{0}\!\big[{\cal F} - \hat{E}_{\cal F} \big]\!} = o_p(1),
		\end{align}
		where $\boldsymbol C$ is a matrix with $C_{ij}$ as its $(i,j)$-th entry. To show that $(a)$ and $(b)$ holds, we first note that $\boldsymbol C =\boldsymbol B - \frac{1}{2}\!\left(\boldsymbol D_{B \oslash M}\boldsymbol M +\boldsymbol M\boldsymbol D_{B \oslash M} \right)$, where $\boldsymbol B$ has $B_{ij}$ as its $(i,j)$-th entry and $\boldsymbol D_{B \oslash M}$ is a diagonal matrix with $B_{ii}/M_{ii}$ as its $(i,i)$-th entry. Note also that for even $p$, $\text{trace}(\boldsymbol C^p) = \text{trace}(\boldsymbol B) + \text{trace}\!\left(2^{-p}\!\left(\boldsymbol D_{B \oslash M} \boldsymbol M +\boldsymbol M\boldsymbol D_{B \oslash M} \right)^p \right)$, as $\boldsymbol B$ and $\boldsymbol M$ are idempotent and orthogonal. Since ${\V}_{0}\!\big[{\cal F} - \hat{E}_{\cal F} \big]\! \ge \min_i \sigma_i^4 \sum_{i=1}^n \sum_{j\neq i} 2C_{ij}^2$ these observations yield
		\begin{align}
			\sum_{i=1}^n \sum_{j\neq i} 2C_{ij}^2 = 2\text{trace}(\boldsymbol C^2) \ge 2\text{trace}(\boldsymbol B) = 2r.
		\end{align}
		For $(a)$, we can now observe that
		\begin{align}
			\text{trace}(\boldsymbol C^4) &= \text{trace}(\boldsymbol B) + \frac{1}{16}\text{trace}\!\left((\boldsymbol D_{B \oslash M}\boldsymbol M +\boldsymbol M\boldsymbol D_{B \oslash M})^4\right)\! \\
			&\le r + \sum_{i=1}^n \tfrac{B_{ii}^4}{M_{ii}^4} \le r\!\left(1 + \left(\max\nolimits_{i} B_{ii}\right)^3\left(\max\nolimits_i M_{ii}\inverse\right)^4\right)\!.
		\end{align}
		Since Assumption~\ref{ass:asympLeave3Out} implies that $\max_i M_{ii}\inverse = O_p(1)$, $(a)$ therefore holds with a rate of $1/r$.
		
		Condition $(b)$ follows immediately from { Assumption~\ref{ass:reg}$(ii)$} and that the variance ${\V}_{0}\!\big[{\cal F} - \hat{E}_{\cal F}\big]$ is at least of order $r$. 
		
%

		
		As the above establishes asymptotic standard normality of ${\cal F} - \hat{E}_{\cal F}$ scaled down by ${\V}_{0}\!\big[{\cal F} - \hat{E}_{\cal F}\big]^{1/2}$, it now suffices for asymptotic size control to show that
		\begin{align}
			\frac{ q_{1-\alpha}(F_{\hat{\boldsymbol w},n-m})-1 }{ \sqrt{2\sum_{\ell=1}^r \hat w_\ell^2 + 2/(n-m) } } \xrightarrow{p} q_{1-\alpha}(\Phi)
		\end{align}
		which by \eqref{eq:normlim} follows provided that $\max_{\ell} \hat w_\ell \xrightarrow{p} 0$ and $n-m \rightarrow \infty$. The latter condition is implied by $\max_i M_{ii}\inverse = O_p(1)$ as it leads to $\limsup_{n \rightarrow \infty} m/n < 1$.
		
		We prove that $\max_{\ell} \hat w_\ell \xrightarrow{p} 0$, by establishing that entries of $\boldsymbol{w}_{\cal F} = (w_1,\dots,w_r)'$ converges to zero when $r \rightarrow \infty$ and that
		\begin{align}\label{eq:wcons}
			\max_{\ell} \left(\tilde w_\ell - w_\ell\right)^2 = O\left( \frac{\max_{i} B_{ii}^{1/2}}{\epsilon_n^{1/2} r} \right)
		\end{align}
		where the entries of both $\boldsymbol{w}_{\cal F}$ and $\tilde {\boldsymbol{w}}$ are sorted by magnitude. Since $\max_{\ell} \hat w_\ell \le \max_{\ell} \tilde w_\ell$ and $\frac{\max_{i} B_{ii}}{\epsilon_n r} = o_p(1)$ these observations yield the desired conclusion.
		
		First, have that
		\begin{align}
			\max_{\ell} w_\ell \le \norm{\boldsymbol{w}_{\cal F}}
			= \frac{ \sqrt{ \sum_{i=1}^n \sum_{j=1}^n B_{ij}^2  \sigma_i^2  \sigma_j^2 } }{\sum_{i=1}^n B_{ii} \sigma_i^2 } \le \frac{ \max_{i} \sigma_i^2}{ \min_i \sigma_i^2 } \frac{\sqrt{r}}{r} = o_p(1).
		\end{align}
		Second, we have that
		\begin{align}
			\max_{\ell} \left(\tilde w_\ell - w_\ell\right)^2 \le \left(\frac{\E_0[{\cal F}]}{\hat E_{\cal F}}\right)^2 \frac{ \sum_{i=1}^n \left(\hat \sigma_i^2 -\sigma_i^2\right) \sum_{j=1}^n B_{ij}^2 \left(\hat \sigma_j^2 - \sigma_j^2\right)  }{ \left( \sum_{i=1}^n B_{ii} \sigma_i^2 \right)^2 }.
		\end{align}
		It follows from the first part of this proof that $\hat E_{\cal F}/\E_0[{\cal F}] \xrightarrow{p} 1$, and from an application of Cauchy-Schwarz that
		\begin{align}
			\E\left[\sum_{i=1}^n \left(\hat \sigma_i^2 -\sigma_i^2\right) \sum_{j=1}^n B_{ij}^2 \left(\hat \sigma_j^2 - \sigma_j^2\right) \right] &\le \sum_{i=1}^n \V\left[ \hat \sigma_i^2\right]^{1/2} \V\left[ \sum_{j=1}^n B_{ij}^2 \left(\hat \sigma_j^2 - \sigma_j^2\right) \right] \\
			= O\left(  \max_{i} \frac{ (x_i'\beta)^2 }{M_{ii}} \sum_{i=1}^n \sqrt{\sum_{j=1}^n B_{ij}^4} \right) &= O\left( \frac{\max_{i} B_{ii}^{1/2}}{\epsilon_n^{1/2} r} \right).
		\end{align}

		\noindent \textbf{Asymptotic size control when $r$ is fixed} For $r$ fixed, it must be that $\max_i B_{ii} = o_p(1)$ by Assumption 2. When $r$ is fixed, we can without loss of generality suppose that $\boldsymbol{w}_{\cal F}$ converges in probability (as we can otherwise argue along subsequences) and will use $\overrightarrow{\boldsymbol{w}}_{\cal F}$ to denote this limit. The entries of this limit are necessarily strictly positive. It follows from \eqref{eq:wcons}, that $\tilde{\boldsymbol w} \xrightarrow{p} \overrightarrow{\boldsymbol{w}}_{\cal F}$ and thus also that $\hat{\boldsymbol w} \xrightarrow{p} \overrightarrow{\boldsymbol{w}}_{\cal F}$. This conclusion naturally implies that we also have $\hat E_{\cal F}/\E_0[{\cal F}] \xrightarrow{p} 1$.
		
		Lyapounovs central limit theorem and $\max_i B_{ii} = o_p(1)$ implies that $\V[ \boldsymbol R \hat {\boldsymbol\beta}]^{-1/2} (\boldsymbol R\hat {\boldsymbol \beta} - \boldsymbol q) \xrightarrow{d} N(0,I_r)$ which when coupled with the conclusions above and the continuous mapping theorem implies that ${\cal F}/\hat E_{\cal F} \xrightarrow{d} \bar \chi^2_{\overrightarrow{\boldsymbol{w}}_{\cal F}}$. Finally, we have that $\V_0[ {\cal F} - \hat E_{\cal F}]/\E_0[{\cal F}] ^2 = V_0[ {\cal F}]/\E_0[{\cal F}]^2 + o_p(1) = 2\norm{\overrightarrow{\boldsymbol{w}}_{\cal F}}^2 + o_p(1)$ and due to the continuous nature of this variance we also have that $2\norm{\hat {\boldsymbol{w}}}^2 + 2/(n-m) =  2\norm{\overrightarrow{\boldsymbol{w}}_{\cal F}}^2 + o_p(1)$. Thus it follows that
		\begin{align}
			\lim_{n \rightarrow \infty} \Pr\left( \frac{\cal F  }{\hat{E}_{\cal F}} >  q_{1-\alpha}(\bar F_{\hat {\boldsymbol w},n-m}) + \!\left(q_{1-\alpha}(\bar F_{\hat {\boldsymbol w},n-m})-1\right)\!\left(\tfrac{\V_0[ {\cal F} - \hat E_{\cal F}]^{1/2}/\hat{E}_{\cal F}}{\sqrt{2\sum_{\ell=1}^r \hat w_\ell^2 +2/(n-m)}} -1 \right) \right) = 1-\alpha.
		\end{align}

		To finish the proof we only need to establish that $\hat{V}_{\cal F}/{\V}_{0}\!\big[{\cal F} - \hat{E}_{\cal F}\big] \xrightarrow{p} 1$.
		
		\noindent \textbf{Consistency of variance estimator}
		In the remainder of this proof $\sum_{i \neq j}^n$ is shorthand for the double sum $\sum_{i=1}^n \sum_{j\neq i}$, $\sum_{i \neq j \neq k}^n$ denotes $\sum_{i=1}^n \sum_{j\neq i} \sum_{k \neq i,j}$, and $\sum_{i \neq j \neq k \neq \ell}^n$ abbreviates $\sum_{i=1}^n \sum_{j\neq i} \sum_{k \neq i,j} \sum_{\ell \neq i,j,k}$. Similarly, we use $\sum_{i,j}^n$ to denote the double sum $\sum_{i=1}^n \sum_{j=1}^n$. Note that $\sum_{j \neq i}$ (without a raised $n$) will still denote a single sum that excludes $i$.
		
		From the algebraic manipulations of the leave-out estimators provided in Appendix \ref{app:varest}, it follows that the proposed variance estimator satisfies the decomposition
			\begin{align}\label{eq:Varsum}
			\hat{V}_{\cal F}\,&= \sum_{i \neq j}^n U_{ij} y_i^2 y_j \varepsilon_j
			+ \!\sum_{i\neq j \neq k}^n\! a_{ijk} y_i^2 y_j \varepsilon_k  + b_{ijk} y_i \varepsilon_i y_j y_k  +
			\!\sum_{i\neq j \neq k \neq \ell}^n\! b_{ijk\ell} y_i y_j y_k \varepsilon_\ell
			\quad
			\end{align}
		where the weights in \eqref{eq:Varsum} are
		\begin{align}
			a_{ijk} = U_{ij} \check M_{jk,-ij}, \quad b_{ijk} &= \!\left( U_{ij}  - V_{ij}^2 \right)\! \check M_{jk,-ij}+V_{ij}V_{ik}, \quad \text{and} \quad b_{ijk\ell} = b_{ijk}\check M_{i\ell,-ijk}.
		\end{align}
		
		Appendix \ref{app:varest} already showed that $\hat{V}_{\cal F}$ is conditionally unbiased, so consistency follows if the conditional variance of $\hat{V}_{\cal F}$ is small relative to the squared estimand ${\V}_0\big[{\cal F} - \hat{E}_{\cal F} \big]^2$. The derivations further below establish that this is the case by working with the four components of \eqref{eq:Varsum} one at a time.
		
		An essential algebraic trick that is used repeatedly below is that the property $\boldsymbol M^2=\boldsymbol M$ or $\sum_{k=1}^n M_{jk} M_{\ell k} = M_{j\ell} = M_{\ell j}$ translate into similar statements regarding the leave-out analogs $\check M_{ij,-i}$, $\check M_{jk,-ij}$ and $\check M_{i \ell,-ijk}$:
		\begin{align}
			\sum\nolimits_{j=1}^n \check M_{ij,-i} \check M_{\imath j,-\imath}
			&=\tfrac{M_{\imath i,-\imath}}{M_{ii}}
			= \tfrac{\check M_{i\imath,-i}}{M_{\imath \imath}}
			= \tfrac{M_{i\imath}}{M_{ii} M_{\imath \imath}}
			\shortintertext{for leave-one-out,}
			\sum\nolimits_{k=1}^n \check M_{jk,-ij} \check M_{\jmath k,-\imath \jmath}
			&= \tfrac{\check M_{\jmath j,-\imath \jmath}- \check M_{\jmath i,-\imath \jmath}\check M_{ij,-i}}{D_{ij}/M_{ii}}
			=\tfrac{\check M_{j \jmath,-i j} - \check M_{j \imath,-i j}\check M_{\imath \jmath,-\imath}}{D_{\imath \jmath}/M_{\imath \imath}} \label{eq:sumcheckM} \\
			&= \tfrac{M_{ii} \!\left( M_{\imath\imath} M_{\jmath j} -  M_{\imath\jmath} M_{\imath j}  \right)\! - \!\left( M_{\imath\imath} M_{\jmath i} - M_{\imath\jmath} M_{\imath i}    \right) \! M_{ij}  }{D_{ij} D_{\imath \jmath}}
		\end{align}
		for leave-two-out, and in the case of leave-three-out:
		\begin{align}
			\sum\nolimits_{\ell=1}^n \check M_{i\ell,-ijk} \check M_{\imath \ell,-\imath \jmath \kappa}
			%
			&= \tfrac{\check M_{i \imath,-ijk} - \check M_{i \jmath,-i j k} \check M_{\jmath \imath,-\jmath \kappa}  - \check M_{i \kappa,-ijk} \check M_{\kappa \imath,-\jmath \kappa} }{D_{\imath \jmath \kappa}/D_{\jmath \kappa}} \label{eq:sumcheckM3} \\
			&= \tfrac{\!\left(M_{i \imath} D_{j k}  - M_{ij} \!\left( M_{k k} M_{j \imath} - M_{k j} M_{k \imath} \right)\!   -  M_{ik} \!\left( M_{j j} M_{k \imath} - M_{k j} M_{j \imath} \right)\! \right)\! D_{\jmath \kappa}}{D_{i j k} D_{\imath \jmath \kappa}} \\
			&- \tfrac{\!\left(M_{i \jmath}D_{j k}  - M_{ij} \!\left( M_{k k} M_{j\jmath} - M_{k j} M_{k \jmath} \right)\! -  M_{ik} \!\left( M_{j j} M_{k\jmath} - M_{k j} M_{j \jmath} \right)\! \right)\!\left( M_{\kappa \kappa} M_{\jmath \imath} - M_{\kappa \jmath} M_{\kappa \imath} \right)\! }{D_{i j k} D_{\imath \jmath \kappa}} \\
			&- \tfrac{\!\left(M_{i \kappa}D_{j k} - M_{ij} \!\left( M_{k k} M_{j\kappa} - M_{k j} M_{k \kappa} \right)\!  -  M_{ik} \!\left( M_{j j} M_{k\kappa} - M_{k j} M_{j \kappa} \right)\! \right)\!\left( M_{\jmath \jmath} M_{\kappa \imath} - M_{\kappa \jmath} M_{\jmath \imath} \right)\! }{D_{i j k} D_{\imath \jmath \kappa}}.
		\end{align}
		Beyond these identities, the remaining arguments rely on well-known inequalities such as Cauchy-Schwarz, Minkowski, and Courant-Fischer.
		
		\noindent \textbf{First component of $\hat{V}_{\cal F}$.}
		 For the first component of \eqref{eq:Varsum}, we have
		\begin{align}
		\V\!\left[ \sum\nolimits_{i \neq j}^n U_{ij} y_i^2 y_j \varepsilon_j  \right]\!
		&= \sum\nolimits_{i \neq j}^n U_{ij}^2\!\left(\V\!\left[ y_i^2 y_j \varepsilon_j  \right]\! + \mathbb{C}\!\left[y_i^2 y_j \varepsilon_j,y_j^2 y_i \varepsilon_i\right]\!\right)\! \\
		&+ \sum\nolimits_{i\neq j \neq k}^n U_{ij} U_{jk}\!\left(\mathbb{C}\!\left[y_i^2 y_j \varepsilon_j,y_k^2 y_j \varepsilon_j\right]\!+\mathbb{C}\!\left[y_i^2 y_j \varepsilon_j,y_j^2 y_k \varepsilon_k\right]\!\right)\! \\
		&+ \sum\nolimits_{i\neq j \neq k}^n U_{ij} U_{ik}\!\left(\mathbb{C}\!\left[y_i^2 y_j \varepsilon_j,y_i^2 y_k \varepsilon_k\right]\!+\mathbb{C}\!\left[y_i^2 y_j \varepsilon_j,y_k^2 y_i \varepsilon_i\right]\!\right)\!  \\
		&\le \max_{i,j} \E\!\left[ y_i^4  \right]\! \E\!\left[ y_j^2 \varepsilon_j^2  \right]\! 4 \sum\nolimits_{i=1}^n\!\left(  \sum\nolimits_{j \neq i} U_{ij} \right)^2.
		\label{eq:comp1}
		\end{align}
		The upper bound on this variance is $4$ times a product between a conditional moment $\max_{i,j} \E\!\left[ y_i^4  \right]\!\E\!\left[ y_j^2 \varepsilon_j^2  \right]\!$ which is $O_p(\epsilon_n^{-1})$ by { Assumptions \ref{ass:iid} and \ref{ass:reg}$(i)$} and a sum of squared influences $\sum\nolimits_{i=1}^n\!\big(  \sum\nolimits_{j \neq i} U_{ij} \big)^2$. The latter term we can write as four times $\text{trace}\big(\boldsymbol C^2 \odot\boldsymbol C^2\big)\!$, where $\odot$ denotes Hadamard (elementwise) product. This representation immediately yields
		\begin{align}\label{eq:comp1bound}
		\tfrac{\sum_{i=1}^n\!\left(  \sum_{j \neq i} U_{ij} \right)^2}{{\V}_{0}\!\left[{\cal F} - \hat{E}_{\cal F} \right]^2} \le \tfrac{4 \cdot \text{trace}(\boldsymbol C^2) \cdot \max_{i }\sum_{j \neq i} U_{ij}}{{\V}_{0}\!\left[{\cal F} - \hat{E}_{\cal F}\right]^2}   = O\!\left( \tfrac{\max_{i }\sum_{j \neq i} U_{ij}}{\!\left(\sum_{i\neq j}^n C_{ij}^2\right)^2} \right)\! = O\!\left( \tfrac{\max_i B_{ii}}{\epsilon_n r}\right),
		\end{align}
		where the last two equalities follow from the asymptotic normality step of this proof.
		
		
		\noindent \textbf{Second component of $\hat{V}_{\cal F}$.}
		The second component of \eqref{eq:Varsum} we further decompose into two parts
		\begin{align}\label{eq:comp2}
		\!\sum\nolimits_{i\neq j \neq k}\! a_{ijk} y_i^2 \varepsilon_j \varepsilon_k + a_{ijk} {\boldsymbol x}_j'{\boldsymbol\beta} y_i^2 \varepsilon_k
		\end{align}
		Proceeding with variance calculations and bounds for the first part we have
		\begin{align}
		\V\!\left[\sum\nolimits_{i\neq j \neq k}^n a_{ijk} y_i^2 \varepsilon_j \varepsilon_k  \right]\! &= \sum\nolimits_{i\neq j \neq k}^n a_{ijk}\!\left( a_{ijk} + a_{ikj}\right)\! \E\!\left[y_i^4 \varepsilon_j^2 \varepsilon_k^2 \right]\! \\
		&+ \sum\nolimits_{i\neq j \neq k}^n\!\left( a_{ijk}\!\left(a_{jik} + a_{jki}\right)\! + a_{ikj}\!\left(a_{jik} + a_{jki}\right)\!\right)\! \E\!\left[y_i^2 \varepsilon_i y_j^2 \varepsilon_j \varepsilon_k^2 \right]\! \\
		&+ \sum\nolimits_{i\neq j \neq k \neq \ell}^n a_{ijk}\!\left( a_{\ell jk} + a_{\ell kj}\right)\! \E\!\left[y_i^2 y_\ell^2 \varepsilon_j^2 \varepsilon_k^2 \right]\! \\
		&\le \max_{i,j} \E\!\left[y_i^4\right]\! \sigma_j^4  8 \sum\nolimits_{i\neq j \neq k}^n  a_{ijk}^2
		+ \max_{j} \sigma_i^4 \sum\nolimits_{j \neq k}^n \!\left(\sum\nolimits_{i \neq j}^n a_{ijk} \E\!\left[y_i^2 \right]\! \right)^2
		\end{align}
		where we utilize $a_{iji}=0$. Now observe that a special case of \eqref{eq:sumcheckM}
		\begin{align}
		\sum\nolimits_{k=1}^n \check M_{jk,-ij} \check M_{jk,-\imath j}
		= \tfrac{1 - \check M_{ij,-i}\check M_{j i,-\imath j} }{D_{ij}/M_{ii}}
		=\tfrac{1 - \check M_{\imath j,-\imath}\check M_{j \imath,-i j}}{D_{\imath j}/M_{\imath \imath}}
		&= \tfrac{M_{jj} M_{ii} M_{\imath\imath} + M_{i\imath} M_{ij} M_{\imath j} - M_{ij}^2 M_{\imath\imath} - M_{\imath j}^2 M_{ii} }{D_{ij} D_{\imath j}}
		\end{align}
		is bounded in absolute value by $4D_{ij}\inverse D_{\imath j}\inverse$, and that a further special case of \eqref{eq:sumcheckM} yields the bound $\sum\nolimits_{k=1}^n \check M_{jk,-ij}^2 = \frac{M_{ii}}{D_{ij}} \le \frac{1}{D_{ij}} $. In turn these bounds lead to
		\begin{align}
		\tfrac{\sum\nolimits_{i\neq j \neq k}^n  a_{ijk}^2}{{\V}_{0}\!\left[{\cal F} - \hat{E}_{\cal F}\right]^2}
		&\le \max_{i \neq j} \tfrac{1}{D_{ij}} \times  \tfrac{\sum\nolimits_{i=1}^n \!\left(\sum\nolimits_{j\neq i} U_{ij}\right)^2 }{{\V}_{0}\!\left[{\cal F} - \hat{E}_{\cal F} \right]^2}
		= O\!\left( \tfrac{\max_i B_{ii}}{\epsilon_n r}  \right)\!, \\
		\tfrac{\sum\nolimits_{j \neq k}^n \!\left(\sum\nolimits_{i \neq j}^n a_{ijk} \E\!\left[y_i^2 \right]\! \right)^2}{{\V}_{0}\!\left[{\cal F} - \hat{E}_{\cal F} \right]^2}
		&\le \max_{i \neq j} \tfrac{\E\!\left[y_i^2 \right]^2}{D_{ij}^2} \times  \tfrac{4 \sum\nolimits_{i=1}^n \!\left(\sum\nolimits_{j\neq i} U_{ij}\right)^2 }{{\V}_{0}\!\left[{\cal F} - \hat{E}_{\cal F} \right]^2}
		= O\!\left( \tfrac{\max_i B_{ii}}{\epsilon_n r}  \right)\!
		\end{align}
		where the order statement of the first line stems from Assumption~\ref{ass:asympLeave3Out} from which it follows that $\max_{i \neq j} D_{ij}\inverse = O_p(1)$, while the second order statement additionally utilizes { Assumptions~\ref{ass:iid} and \ref{ass:reg}$(i)$} from which we obtain $\max_{i \neq j} {\E\!\left[y_i^2 \right]^2}/{D_{ij}^2} = O_p(\epsilon_n\inverse)$.
		
		Turning to a variance calculation for the second part of \eqref{eq:comp2} we have
		\begin{align}
		\V\!\left[\sum\nolimits_{i\neq j \neq k}^n a_{ijk}  {\boldsymbol x}_j'{\boldsymbol\beta} y_i^2 \varepsilon_k  \right]\! &=
		\sum\nolimits_{i \neq k}^n   a_{i\cdot k}^2 \E\!\left[ y_i^4 \varepsilon_k^2 \right]\! +  a_{i\cdot k} a_{k\cdot i} \E\!\left[ y_i^2 \varepsilon_{i} y_k^2 \varepsilon_k \right]\! \\
		&+ \sum\nolimits_{i \neq \ell \neq k}^n   a_{i\cdot k}  a_{\ell \cdot k} \E\!\left[ y_i^2 y_\ell^2 \varepsilon_k^2 \right]\! \\
		& \le 2 \max_{i,j} \E[y_i^4] \sigma_j^2 \sum\nolimits_{i , k}^n \! a_{i\cdot k}^2 + \max_i \sigma_i^2 \sum\nolimits_{k=1}^n \!\left( \sum\nolimits_{i=1}^n \!a_{i\cdot k} \E\!\left[ y_i^2 \right]\!\right)^2\!, \quad
		\label{eq:comp2part2}
		\end{align}
		where $a_{i \cdot k}= \sum\nolimits_{j \neq i,k} a_{ijk} {\boldsymbol x}_j'{\boldsymbol\beta}$.  From \eqref{eq:sumcheckM} we obtain the special case
		\begin{align}
		\sum\nolimits_{k=1}^n \check M_{jk,-ij} \check M_{\ell k,-i\ell}
		= \tfrac{\check M_{\ell j,-i\ell}}{D_{ij}/M_{ii}}
		= \tfrac{\check M_{j\ell,-ij}}{D_{i\ell}/M_{ii}}
		= \tfrac{M_{ii}\!\left( M_{ii}M_{j\ell} - M_{ij}M_{i\ell}\right)\!}{D_{ij}D_{i\ell}}
		\end{align}
		where a coarse bound on the absolute value of this expression is $2D_{ij}\inverse D_{i\ell}\inverse$. Utilizing this coarse bound, we immediately see that the first part of the variance in \eqref{eq:comp2part2} satisfies
		\begin{align}
		\tfrac{\sum\nolimits_{i , k}^n   a_{i\cdot k}^2}{{\V}_{0}\!\left[{\cal F} - \hat{E}_{\cal F} \right]^2} &\le \tfrac{2 \sum\nolimits_{i=1}^n \!\left(\sum\nolimits_{j\neq i} U_{ij} \abs{{\boldsymbol x}_j'{\boldsymbol\beta}} D_{ij}\inverse\right)^2}{{\V}_{0}\!\left[{\cal F} - \hat{E}_{\cal F} \right]^2}
		\le \max_{i\neq j}\tfrac{({\boldsymbol x}_i'{\boldsymbol\beta})^2}{D_{ij}^2}   \tfrac{2 \sum\nolimits_{i=1}^n \!\left(\sum\nolimits_{j\neq i} U_{ij} \right)^2}{{\V}_{0}\!\left[{\cal F} - \hat{E}_{\cal F} \right]^2}
		= O\!\left( \tfrac{\max_i B_{ii}}{\epsilon_n r}  \right)\!.
		\end{align}
		For the second part of \eqref{eq:comp2part2} we instead rely on the full generality of \eqref{eq:sumcheckM}
		\begin{align}
		\sum\nolimits_{k=1}^n \check M_{jk,-ij} \check M_{\ell k,-\imath\ell}
		&= \tfrac{M_{j\ell} M_{ii} M_{\imath\imath} + M_{i\imath} M_{ij} M_{\imath\ell} -  M_{\imath\imath} M_{ij} M_{i\ell} -  M_{ii} M_{\imath j} M_{\imath\ell}}{D_{ij} D_{\imath \ell}}.
		\end{align}
		When coupled with the observation that the eigenvalues of $M$ belong to $\{0,1\}$, this leads to
		\begin{align}
		\tfrac{\sum\nolimits_{k=1}^n \!\left( \sum\nolimits_{i=1}^n a_{i\cdot k}  \E\![ y_i^2 ]\!\right)^2}{{\V}_{0}\!\left[{\cal F} - \hat{E}_{\cal F} \right]^2}
		\le \max_{i \neq j} \tfrac{\E\![ y_i^2 ]^2 ({\boldsymbol x}_j'{\boldsymbol\beta})^2}{D_{ij}^{2}}   \tfrac{ 4\sum\nolimits_{i=1}^n  \!\left( \sum\nolimits_{j\neq i} U_{ij} \right)^2}{{\V}_{0}\!\left[{\cal F} - \hat{E}_{\cal F} \right]^2}  = O\!\left( \tfrac{\max_i B_{ii}}{\epsilon_n r}  \right)\!,
		\end{align}
		where the last relation follows from $\max_{i \neq j} D_{ij}^{-2}{\E\!\left[ y_i^2 \right]^2 ({\boldsymbol x}_j'{\boldsymbol\beta})^2} = O_p(\epsilon_n\inverse )$ which holds by { Assumptions~\ref{ass:iid} and \ref{ass:reg}$(i)$} and Assumption~\ref{ass:asympLeave3Out}.

		\noindent \textbf{Third component of $\hat{V}_{\cal F}$.}
		For the third component of \eqref{eq:Varsum} we similarly employ a decomposition
		\begin{align}\label{eq:comp3}
		\!\sum\nolimits_{i\neq j \neq k}^n\! b_{ijk}\!\left( {\boldsymbol x}_j'{\boldsymbol\beta} y_i \varepsilon_i y_k + {\boldsymbol x}_k'{\boldsymbol\beta} y_i \varepsilon_i \varepsilon_j  + y_i \varepsilon_i \varepsilon_j \varepsilon_k \right)\!
		\end{align}
		where the variance of the first part satisfies
		\begin{align}
		\V\!\left[\sum\nolimits_{i\neq j \neq k}^n b_{ijk} {\boldsymbol x}_j'{\boldsymbol\beta} y_i \varepsilon_i y_k\right]\!
		&= \sum\nolimits_{i \neq k}^n b_{i\cdot k}^2\V\!\left[ y_i \varepsilon_i y_k \right]\! + b_{i\cdot k}b_{k\cdot i}\mathbb{C}\!\left[y_i \varepsilon_i y_k,y_k \varepsilon_k y_i\right]\! \\
		&+ \sum\nolimits_{i\neq k \neq \ell}^n b_{i\cdot k}b_{\ell\cdot k}\mathbb{C}\!\left[y_i \varepsilon_i y_k,y_\ell \varepsilon_\ell y_k\right]\!+ b_{i\cdot k}b_{i\cdot \ell}\mathbb{C}\!\left[y_i \varepsilon_i y_k,y_i \varepsilon_i y_\ell\right]\! \\
		&+ 2\sum\nolimits_{i\neq k \neq \ell}^n b_{i\cdot k}b_{k\cdot \ell}\mathbb{C}\!\left[y_i \varepsilon_i y_k,y_k \varepsilon_k y_\ell\right]\!  \\
		&\le \max_{i,j} \E\!\left[y_i^2\right]\!\E\!\left[y_j^2 \varepsilon_j^2\right]\! 6 \sum\nolimits_{i , k}^n b_{i\cdot k}^2 \\
		&+ \max_{i} \!\left( \E\!\left[y_i^2\right]\! + \abs*{\mathbb{C}\!\left[y_i,y_i \varepsilon_i\right]\! }\right)\! \sum\nolimits_{k=1}^n\!\left( \sum\nolimits_{i =1}^n b_{i\cdot k} \sigma_i^2 \right)^2 \\
		&+ \max_{i} \!\left( \E\!\left[y_i^2 \varepsilon_i^2\right]\! + \abs*{\mathbb{C}\!\left[y_i,y_i \varepsilon_i\right]\! } \right)\! \sum\nolimits_{i=1}^n  b_{i\cdot \cdot}^2
		\end{align}
		for $b_{i \cdot k} = \sum_{j\neq i,k} b_{ijk} {\boldsymbol x}_j'{\boldsymbol\beta}$, $b_{i \cdot \cdot} = \sum_{k=1}^n b_{i\cdot k} {\boldsymbol x}_k'{\boldsymbol\beta}$ and we have used that $b_{iji}=0$. From the representation
		\begin{align}
			b_{i \cdot k} = -\check {\boldsymbol x}_i'{\boldsymbol\beta} V_{ik}  + \sum\nolimits_{j \neq i,k}  U_{ij} \check M_{jk,-ij} {\boldsymbol x}_j'{\boldsymbol\beta} - \sum\nolimits_{j \neq i} V_{ij}^2  \check M_{jk,-ij} {\boldsymbol x}_j'{\boldsymbol\beta}
		\end{align}
		and the previously derived bound $\abs*{\sum_{k=1}^n \check M_{jk,-ij} \check M_{\ell k,-i\ell}} = 2D_{ij}\inverse D_{i\ell}\inverse$, we immediately obtain
		\begin{align}
			\tfrac{\sum\nolimits_{i , k}^n b_{i\cdot k}^2}{{\V}_{0}\!\left[{\cal F} - \hat{E}_{\cal F} \right]^2}
			&\le \max_{i \neq j} \tfrac{({\boldsymbol x}_i'{\boldsymbol\beta})^2}{D_{ij}^2}   \tfrac{3 \sum\nolimits_{i=1}^n\!\left( \sum\nolimits_{j \neq i} U_{ij}  \right)^2
			+ 3\sum\nolimits_{i=1}^n\!\left( \sum\nolimits_{j \neq i} V_{ij}^2  \right)^2 }{{\V}_{0}\!\left[{\cal F} - \hat{E}_{\cal F} \right]^2}
			+\tfrac{3\sum_{i=1}^n (\check {\boldsymbol x}_i'{\boldsymbol\beta})^2 \sum_{j=1}^n V_{ij}^2}{{\V}_{0}\!\left[{\cal F} - \hat{E}_{\cal F} \right]^2}.
		\end{align}
		Since $\boldsymbol M^2 =\boldsymbol M$ and the largest eigenvalue of $\boldsymbol M \odot\boldsymbol M$ is bounded by one  (a consequence of the Gershgorin circle theorem), it follows that
		\begin{align}
			\tfrac{\sum\nolimits_{i=1}^n\!\left( \sum\nolimits_{j \neq i} V_{ij}^2  \right)^2 }{{\V}_{0}\!\left[{\cal F} - \hat{E}_{\cal F} \right]^2}
			\le \tfrac{16 \sum\nolimits_{i=1}^n B_{ii}^4/M_{ii}^4}{{\V}_{0}\!\left[{\cal F} - \hat{E}_{\cal F} \right]^2}
			= O\!\left( \tfrac{\max_i B_{ii}}{r}  \right)\!,
		\end{align}
		and since ${\V}_{0}\!\big[{\cal F} - \hat{E}_{\cal F} \big] \ge \min_i \sigma_i^2 \sum_{i=1}^n (\check {\boldsymbol x}_i'{\boldsymbol\beta})^2$ we similarly have that
		\begin{align}
			 \tfrac{\sum_{i=1}^n (\check {\boldsymbol x}_i'{\boldsymbol\beta})^2 \sum_{j=1}^n V_{ij}^2}{{\V}_{0}\!\left[{\cal F} - \hat{E}_{\cal F} \right]^2} \le \max_i M_{ii}^{-2}   \tfrac{2\sum_{i=1}^n (\check {\boldsymbol x}_i'{\boldsymbol\beta})^2}{{\V}_{0}\!\left[{\cal F} - \hat{E}_{\cal F} \right]^2}
			 =  o_p\!\left( 1 \right)\!.
		\end{align}
		Turning to the second part of this variance we reuse the expression in \eqref{eq:sumcheckM} to derive the bound
		\begin{align}			
			\tfrac{\sum\nolimits_{k=1}^n\!\left( \sum\nolimits_{i =1}^n b_{i\cdot k} \sigma_i^2 \right)^2}{{\V}_{0}\!\left[{\cal F} - \hat{E}_{\cal F} \right]^2}
			& \le \max_{i \neq j} \tfrac{({\boldsymbol x}_i'{\boldsymbol\beta})^2 \sigma_j^4}{D_{ij}^2}   \tfrac{12 \sum\nolimits_{i=1}^n\!\left( \sum\nolimits_{j \neq i} U_{ij}  \right)^2
			+ 12\sum\nolimits_{i=1}^n\!\left( \sum\nolimits_{j \neq i} V_{ij}^2  \right)^2 }{{\V}_{0}\!\left[{\cal F} - \hat{E}_{\cal F} \right]^2}
			+\max_{i,j} \tfrac{\sigma_i^4}{M_{jj}^2}   \tfrac{12\sum_{i=1}^n (\check {\boldsymbol x}_i'{\boldsymbol\beta})^2}{{\V}_{0}\!\left[{\cal F} - \hat{E}_{\cal F} \right]^2} \\
			&=  O_p\!\left( \tfrac{\max_i B_{ii}}{\epsilon_n r} \right)\! +  o_p\!\left( 1 \right)\!.
		\end{align}
		Finally, since $\sum_{k=1}^n \check M_{jk,-ij} {\boldsymbol x}_k = 0$, we have that $b_{i\cdot \cdot} = (\check {\boldsymbol x}_i'{\boldsymbol\beta})^2 - \sum_{j \neq i} U_{ij} ({\boldsymbol x}_j'{\boldsymbol\beta})^2$ so
		\begin{align}
			\tfrac{\sum_{i=1}^n b_{i\cdot \cdot}^2}{{\V}_{0}\!\left[{\cal F} - \hat{E}_{\cal F} \right]^2} \le
			\tfrac{2\sum_{i=1}^n (\check {\boldsymbol x}_i'{\boldsymbol\beta})^4}{{\V}_{0}\!\left[{\cal F} - \hat{E}_{\cal F} \right]^2} + \max_{i} ({\boldsymbol x}_i'{\boldsymbol\beta})^2   \tfrac{2\sum\nolimits_{i=1}^n\!\left( \sum\nolimits_{j \neq i} U_{ij}  \right)^2}{{\V}_{0}\!\left[{\cal F} - \hat{E}_{\cal F} \right]^2}
			=  o_p\!\left( 1 \right)\!
		\end{align}
		where the order statement regarding the first part follows from { Assumption~\ref{ass:reg}$(ii)$} and the derivation in the asymptotic normality part of this proof.

		For the second part of \eqref{eq:comp3} we have
		\begin{align}
		\V\!\left[\sum\nolimits_{i\neq j \neq k}^n b_{ijk} {\boldsymbol x}_k'{\boldsymbol\beta} y_i \varepsilon_i \varepsilon_j\right ] &= \sum\nolimits_{i\neq j}^n b_{ij\cdot}^2 \E\!\left[y_i^2 \varepsilon_i^2 \varepsilon_j^2\right]\! + b_{ij\cdot} b_{ji\cdot}\E\!\left[y_i \varepsilon_i^2 y_j \varepsilon_j^2\right]\! \\
		& + \sum\nolimits_{i\neq j \neq k}^n  b_{ij\cdot} b_{k j\cdot} \E\!\left[y_i \varepsilon_i y_k \varepsilon_k \varepsilon_j^2\right]\! \\
		&\le \max_{i,j} \E\!\left[y_i^2 \varepsilon_i^2\right]\!\! \sigma_j^2 2\sum\nolimits_{i\neq j}^n b_{ij\cdot}^2 +  \max_{i} \sigma_i^2 \sum\nolimits_{j=1}^n\!\left(\sum\nolimits_{i\neq j} b_{ij\cdot} \sigma_i^2\right)^2 \quad
		\label{eq:comp3part2}
		\end{align}
		for $b_{ij\cdot} = \sum_{k\neq j} b_{ijk} {\boldsymbol x}_k'{\boldsymbol\beta}$. We have $b_{ij\cdot} =  - U_{ij}  {\boldsymbol x}_j'{\boldsymbol\beta} + \check {\boldsymbol x}_i'{\boldsymbol\beta} V_{ij}$, which leads to
		\begin{align}
		\tfrac{\sum\nolimits_{i\neq j}^n b_{ij\cdot}^2}{{\V}_{0}\!\left[{\cal F} - \hat{E}_{\cal F} \right]^2}
		&\le \max_{i} ({\boldsymbol x}_i'{\boldsymbol\beta})^2 \tfrac{ 2 \sum\nolimits_{i\neq j}^n  U_{ij}^2 }{{\V}_{0}\!\left[{\cal F} - \hat{E}_{\cal F} \right]^2} +\max_{i} \tfrac{1}{M_{ii}^2} \tfrac{4\sum_{i=1}^n (\check {\boldsymbol x}_i'{\boldsymbol\beta})^2}{{\V}_{0}\!\left[{\cal F} - \hat{E}_{\cal F} \right]^2}
		=  O_p\!\left( \tfrac{\max_i B_{ii}}{\epsilon_n r} \right)\! +  o_p\!\left( 1  \right)\!,\\
		\tfrac{\sum\nolimits_{j=1}^n\!\left( \sum\nolimits_{i \neq j} b_{ij\cdot} \sigma_i^2 \right)^2}{{\V}_{0}\!\left[{\cal F} - \hat{E}_{\cal F} \right]^2}
		&\le \max_{i,j} ({\boldsymbol x}_i'{\boldsymbol\beta})^2 \sigma_j^4  \tfrac{ 8 \sum\nolimits_{i=1}^n\!\left( \sum\nolimits_{j \neq i} U_{ij}  \right)^2}{{\V}_{0}\!\left[{\cal F} - \hat{E}_{\cal F} \right]^2} +\max_{i,j} \tfrac{\sigma_i^4}{M_{jj}^2} \tfrac{8\sum_{i=1}^n (\check {\boldsymbol x}_i'{\boldsymbol\beta})^2}{{\V}_{0}\!\left[{\cal F} - \hat{E}_{\cal F} \right]^2}
		=  o_p\!\left( 1  \right)\!.
		\end{align}

		Turning to the third and final part of \eqref{eq:comp3} we have
		\begin{align}
		\V\!\left[\sum\nolimits_{i\neq j \neq k}^n b_{ijk} y_i \varepsilon_i \varepsilon_j \varepsilon_k\right]\! &= \sum\nolimits_{i\neq j \neq k}^n b_{ijk}\!\left( b_{ijk} + b_{ikj}\right)\! \E\!\left[y_i^2 \varepsilon_i^2 \varepsilon_j^2 \varepsilon_k^2 \right]\! \\
		&+ \sum\nolimits_{i\neq j \neq k}^n\!\left( b_{ijk}\!\left(b_{jik} + b_{jki}\right)\! + b_{ikj}\!\left(b_{jik} + b_{jki}\right)\!\right)\! \E\!\left[y_i \varepsilon_i^2 y_j \varepsilon_j^2 \varepsilon_k^2 \right]\! \\
		&+ \sum\nolimits_{i\neq j \neq k \neq \ell}^n b_{ijk}\!\left( b_{\ell jk} + b_{\ell kj}\right)\! \E\!\left[y_i \varepsilon_i y_\ell \varepsilon_\ell \varepsilon_j^2 \varepsilon_k^2 \right]\! \\
		& \le \max_{i,j} \E\!\left[y_i^2 \varepsilon_i^2 \right]\! \sigma_j^4  8 \sum\nolimits_{i\neq j \neq k}^n  b_{ijk}^2
		+ \max_{i} \sigma_i^2 2\sum\nolimits_{j\neq k}^n \!\left( \sum\nolimits_{i \neq j} b_{ijk} \sigma_i^2 \right)^2.
		\end{align}
		By reusing the inequalities $\abs{\sum\nolimits_{k=1}^n \check M_{jk,-ij} \check M_{jk,-\imath j}} \le  \tfrac{4}{D_{ij} D_{\imath j}}$ and $\abs{\sum\nolimits_{k=1}^n \check M_{jk,-ij}^2 } \le  \tfrac{1}{D_{ij}}$, one can show that
		\begin{align}
			\tfrac{ \sum\nolimits_{i\neq j \neq k}^n  b_{ijk}^2 }{{\V}_{0}\!\left[{\cal F} - \hat{E}_{\cal F} \right]^2}
			&\le \max_{i \neq j} \tfrac{ 1}{D_{ij}} \tfrac{12 \sum\nolimits_{i=1}^n\!\left( \sum\nolimits_{j \neq i} U_{ij}  \right)^2
			+ 12\sum\nolimits_{i=1}^n\!\left( \sum\nolimits_{j \neq i} V_{ij}^2  \right)^2 }{{\V}_{0}\!\left[{\cal F} - \hat{E}_{\cal F} \right]^2}
			+\max_{i} \tfrac{1}{M_{ii}^2} \tfrac{48\sum_{i=1}^n B_{ii}^2/M_{ii}^2}{{\V}_{0}\!\left[{\cal F} - \hat{E}_{\cal F} \right]^2} \\
			&=  O_p\!\left( \tfrac{\max_i B_{ii}}{ r} \right)\!,
			\label{eq:comp3part3} \\
			\tfrac{\sum\nolimits_{j\neq k}^n \!\left( \sum\nolimits_{i \neq j} b_{ijk} \sigma_i^2 \right)^2}{{\V}_{0}\!\left[{\cal F} - \hat{E}_{\cal F} \right]^2}
			&\le \max_{i \neq j} \tfrac{ \sigma_i^4}{D_{ij}^2} \tfrac{12 \sum\nolimits_{i=1}^n\!\left( \sum\nolimits_{j \neq i} U_{ij}  \right)^2
			+ 12\sum\nolimits_{i=1}^n\!\left( \sum\nolimits_{j \neq i} V_{ij}^2  \right)^2 }{{\V}_{0}\!\left[{\cal F} - \hat{E}_{\cal F} \right]^2}
			+\max_{i,j} \tfrac{\sigma_i^4}{M_{jj}^2} \tfrac{48\sum_{i=1}^n B_{ii}^2/M_{ii}^2}{{\V}_{0}\!\left[{\cal F} - \hat{E}_{\cal F} \right]^2} \\
			&=  O_p\!\left( \tfrac{\max_i B_{ii}}{ r} \right)\!.
		\end{align}

		\noindent \textbf{Fourth component of $\hat{V}_{\cal F}$.}
		The fourth and final component of \eqref{eq:Varsum} we can rewrite as
		\begin{align}\label{eq:comp4}
		\!\sum\nolimits_{i\neq j \neq k \neq \ell}^n\! b_{ijk\ell}\!\left( {\boldsymbol x}_k'{\boldsymbol\beta} {\boldsymbol x}_j'{\boldsymbol\beta} y_i \varepsilon_\ell + {\boldsymbol x}_k'{\boldsymbol\beta} y_i \varepsilon_j \varepsilon_\ell + {\boldsymbol x}_j'{\boldsymbol\beta} y_i \varepsilon_k \varepsilon_\ell + {\boldsymbol x}_i'{\boldsymbol\beta} \varepsilon_j \varepsilon_k \varepsilon_\ell + \varepsilon_i \varepsilon_j \varepsilon_k \varepsilon_\ell \right)\!
		\end{align}
		where the variance of the first term satisfies
		\begin{align}
		\V\!\left[\!\sum\nolimits_{i\neq j \neq k \neq \ell}^n\! b_{ijk\ell}{\boldsymbol x}_k'{\boldsymbol\beta} {\boldsymbol x}_j'{\boldsymbol\beta} y_i \varepsilon_\ell\right]\! &= \sum\nolimits_{i \neq \ell}^n  b_{i\cdot \cdot \ell}^2 \E\!\left[ y_i^2 \varepsilon_\ell^2 \right]\! +  b_{i\cdot \cdot \ell} b_{\ell\cdot \cdot i} \E\!\left[ y_i \varepsilon_{i} y_\ell \varepsilon_\ell \right]\! \\
		&+ \sum\nolimits_{i \neq \ell \neq k}^n  b_{i\cdot \cdot \ell} b_{k\cdot \cdot \ell} \E\!\left[ y_i y_k \varepsilon_\ell^2 \right]\! \\
		&\le \max_{i,j} \E\!\left[y_i^2\right]\! \sigma_j^2 2 \sum\nolimits_{i \neq \ell}^n  b_{i\cdot \cdot \ell}^2
		+ \max_{i} \sigma_i^2 \sum\nolimits_{\ell=1}^n \!\left(\sum\nolimits_{i \neq \ell}^n  b_{i\cdot \cdot \ell} {\boldsymbol x}_i'{\boldsymbol\beta}\right)^2
		\end{align}
		for $b_{i\cdot \cdot \ell} = \sum_{j\neq i} b_{ij \cdot \ell} {\boldsymbol x}_j'{\boldsymbol\beta}$ and $b_{ij \cdot \ell} = \sum_{k \neq i,j} b_{ijk\ell}{\boldsymbol x}_k'{\boldsymbol\beta}$ and we have used that $b_{ijkk}=b_{ijkj}=0$. To further upper bound the first part of this final expression, we rely on \eqref{eq:sumcheckM3} which yields
		\begin{align}
			\sum\nolimits_{\ell = 1}^n \check M_{i\ell,-ijk} \check M_{i\ell,-i \jmath \kappa}
			&= \tfrac{1 - \check M_{i \jmath,-i j k} \check M_{\jmath i,-\jmath \kappa} - \check M_{i \kappa,-i j k} \check M_{\kappa i,-\jmath \kappa}  }{D_{i j \kappa}/D_{\jmath \kappa}}.
		\end{align}
		
		Using this identity in conjunction with Cauchy-Schwarz and $M$ having all its eigenvalues in $\{0,1\}$ we obtain that
		\begin{align}
			\tfrac{\sum\nolimits_{i \neq \ell}^n  b_{i\cdot \cdot \ell}^2 }{{\V}_{0}\!\left[{\cal F} - \hat{E}_{\cal F} \right]^2}
			&\le \tfrac{\sum\nolimits_{i=1}^n  b_{i\cdot \cdot}^2 }{{\V}_{0}\!\left[{\cal F} - \hat{E}_{\cal F} \right]^2} + \max_{i \neq j \neq k \neq i} \tfrac{({\boldsymbol x}_i'{\boldsymbol\beta})^2 ({\boldsymbol x}_j'{\boldsymbol\beta})^2}{D_{ijk}^2} \tfrac{12\sum\nolimits_{i\neq j \neq k}^n \!  b_{i jk}^2 }{{\V}_{0}\!\left[{\cal F} - \hat{E}_{\cal F} \right]^2} \\
			&+ \tfrac{2\sum\nolimits_{i=1}^n  \left( \sum\nolimits_{j \neq i} \sum\nolimits_{k \neq i,j} b_{ijk} {\boldsymbol x}_j'{\boldsymbol\beta} {\boldsymbol x}_k'{\boldsymbol\beta} D_{jk}/D_{ijk}\right)^2 }{{\V}_{0}\!\left[{\cal F} - \hat{E}_{\cal F} \right]^2} \\
			&+\max_{i} ({\boldsymbol x}_i'{\boldsymbol\beta})^2 \tfrac{6 \sum\nolimits_{i\neq j}^n  \!\left( \sum\nolimits_{k \neq i,j} b_{i jk} {\boldsymbol x}_k'{\boldsymbol\beta} M_{kk}/D_{ijk} \right)^2}{{\V}_{0}\!\left[{\cal F} - \hat{E}_{\cal F} \right]^2}  \\
			&+ \max_{i} ({\boldsymbol x}_i'{\boldsymbol\beta})^2 \tfrac{6\sum\nolimits_{i\neq k}^n  \!\left( \sum\nolimits_{j \neq i,k} b_{i jk} {\boldsymbol x}_j'{\boldsymbol\beta} M_{jj}/D_{ijk} \right)^2}{{\V}_{0}\!\left[{\cal F} - \hat{E}_{\cal F} \right]^2}.
		\end{align}
		The first two terms of this bound were shown to be $o_p(1)$ in the treatment of the third component of \eqref{eq:Varsum}, so here we focus on the latter three. The identities
		\begin{align}
			\tfrac{M_{jj}}{D_{ijk}} &= \tfrac{1}{D_{ik}} + \tfrac{M_{ji} \check M_{ij,-ik}+ M_{jk} \check M_{kj,-ik}}{D_{ijk}}, \\
			\tfrac{M_{kk}}{D_{ijk}} &= \tfrac{1}{D_{ij}} + \tfrac{M_{ki} \check M_{ik,-ij}+ M_{kj} \check M_{jk,-ij}}{D_{ijk}}, \\
			\tfrac{D_{jk}}{D_{ijk}} &= \tfrac{1}{M_{ii}} + \tfrac{M_{ij}^2 M_{kk} + M_{ik}^2 M_{jj} - 2 M_{ij} M_{ik} M_{jk}}{M_{ii}D_{ijk}},
		\end{align}
		immediately leads to
		\begin{align}
			\tfrac{\sum\nolimits_{i\neq j}^n  \!\left( \sum\nolimits_{k \neq i,j} b_{i jk} {\boldsymbol x}_k'{\boldsymbol\beta} M_{kk}/D_{ijk} \right)^2}{{\V}_{0}\!\left[{\cal F} - \hat{E}_{\cal F} \right]^2}
			& \le \max_{i \neq j} \tfrac{1}{D_{ij}^2} \tfrac{2\sum\nolimits_{i\neq j}^n  b_{ij\cdot}^2}{{\V}_{0}\!\left[{\cal F} - \hat{E}_{\cal F} \right]^2}
			+ \max_{i \neq j \neq k \neq i} \tfrac{({\boldsymbol x}_i'{\boldsymbol\beta})^2}{D_{ij} D_{ijk}^2} \tfrac{4\sum\nolimits_{i\neq j \neq k}^n  b_{i jk}^2}{{\V}_{0}\!\left[{\cal F} - \hat{E}_{\cal F} \right]^2} =  o_p\!\left( 1 \right)\! \\
			\tfrac{\sum\nolimits_{i\neq k}^n  \!\left( \sum\nolimits_{j \neq i,k} b_{i jk} {\boldsymbol x}_k'{\boldsymbol\beta} M_{kk}/D_{ijk} \right)^2}{{\V}_{0}\!\left[{\cal F} - \hat{E}_{\cal F} \right]^2}
			& \le \max_{i \neq j} \tfrac{1}{D_{ij}^2} \tfrac{2\sum\nolimits_{i\neq k}^n  b_{i \cdot k}^2}{{\V}_{0}\!\left[{\cal F} - \hat{E}_{\cal F} \right]^2}
			+ \max_{i \neq j \neq k \neq i} \tfrac{({\boldsymbol x}_i'{\boldsymbol\beta})^2}{D_{ij} D_{ijk}^2} \tfrac{4\sum\nolimits_{i\neq j \neq k}^n  b_{i jk}^2}{{\V}_{0}\!\left[{\cal F} - \hat{E}_{\cal F} \right]^2} =  o_p\!\left( 1 \right)\!
		\end{align}
		and
		\begin{align}
			\tfrac{\sum\nolimits_{i=1}^n  \left( \sum\nolimits_{j \neq i} \sum\nolimits_{k \neq i,j} b_{ijk} {\boldsymbol x}_j'{\boldsymbol\beta} {\boldsymbol x}_k'{\boldsymbol\beta} D_{jk}/D_{ijk}\right)^2}{{\V}_{0}\!\left[{\cal F} - \hat{E}_{\cal F} \right]^2}
			& \le \max_{i} \tfrac{1}{M_{ii}^2} \tfrac{16\sum\nolimits_{i \neq j \neq k}^n  b_{i j k}^2}{{\V}_{0}\!\left[{\cal F} - \hat{E}_{\cal F} \right]^2}
			+ \max_{i,j} {({\boldsymbol x}_i'{\boldsymbol\beta})^2({\boldsymbol x}_j'{\boldsymbol\beta})^2} \tfrac{4\sum\nolimits_{i=1}^n  b_{i \cdot \cdot}^2}{{\V}_{0}\!\left[{\cal F} - \hat{E}_{\cal F} \right]^2} \\
			&+ \max_{i,j} \tfrac{({\boldsymbol x}_i'{\boldsymbol\beta})^2}{M_{jj}} \tfrac{4\sum\nolimits_{i\neq j}^n  \!\left( \sum\nolimits_{k \neq i,j} b_{i jk} {\boldsymbol x}_k'{\boldsymbol\beta} M_{kk}/D_{ijk} \right)^2}{{\V}_{0}\!\left[{\cal F} - \hat{E}_{\cal F} \right]^2} \\
			&+ \max_{i,j} \tfrac{({\boldsymbol x}_i'{\boldsymbol\beta})^2}{M_{jj}} \tfrac{4\sum\nolimits_{i\neq k}^n  \!\left( \sum\nolimits_{j \neq i,k} b_{i jk} {\boldsymbol x}_j'{\boldsymbol\beta} M_{jj}/D_{ijk} \right)^2}{{\V}_{0}\!\left[{\cal F} - \hat{E}_{\cal F} \right]^2} =  o_p\!\left( 1  \right),\!
		\end{align}
		where the order statements were established in the treatment of the third component of \eqref{eq:Varsum}.
		To deal with $\sum\nolimits_{\ell=1}^n \!\left(\sum\nolimits_{i \neq \ell}^n  b_{i\cdot \cdot \ell} {\boldsymbol x}_i'{\boldsymbol\beta}\right)^2$, we use the full generality of \eqref{eq:sumcheckM3}
		\begin{align}
			 \sum\nolimits_{\ell=1}^n \check M_{i\ell,-ijk} \check M_{\imath \ell,-\imath \jmath \kappa}
			 = \tfrac{\check M_{i \imath,-ijk} - \check M_{i \jmath,-i j k} \check M_{\jmath \imath,-\jmath \kappa}  - \check M_{i \kappa,-ijk} \check M_{\kappa \imath,-\jmath \kappa} }{D_{\imath \jmath \kappa}/D_{\jmath \kappa}}
		\end{align}
		from which we derive that
		\begin{align}
			\!\left(\sum\nolimits_{i \neq \ell}^n  b_{i\cdot \cdot \ell} {\boldsymbol x}_i'{\boldsymbol\beta}\right)^2
			&\le \max_{i,j,k} \tfrac{({\boldsymbol x}_i'{\boldsymbol\beta})^2 ({\boldsymbol x}_j'{\boldsymbol\beta})^2 ({\boldsymbol x}_k'{\boldsymbol\beta})^2}{D_{ijk}^2} \tfrac{13 \sum\nolimits_{i \neq j \neq k}^n  b_{i j k}^2}{{\V}_{0}\!\left[{\cal F} - \hat{E}_{\cal F} \right]^2} \\
			&+\max_{i} ({\boldsymbol x}_i'{\boldsymbol\beta})^2 \tfrac{3 \sum\nolimits_{i=1}^n  \left( \sum\nolimits_{j \neq i} \sum\nolimits_{k \neq i,j} b_{ijk} {\boldsymbol x}_j'{\boldsymbol\beta} {\boldsymbol x}_k'{\boldsymbol\beta} D_{jk}/D_{ijk}\right)^2}{{\V}_{0}\!\left[{\cal F} - \hat{E}_{\cal F} \right]^2} \\
			&+\max_{i,j} ({\boldsymbol x}_i'{\boldsymbol\beta})^2 ({\boldsymbol x}_j'{\boldsymbol\beta})^2 \tfrac{7 \sum\nolimits_{i\neq j}^n  \!\left( \sum\nolimits_{k \neq i,j} b_{i jk} {\boldsymbol x}_k'{\boldsymbol\beta} M_{kk}/D_{ijk} \right)^2}{{\V}_{0}\!\left[{\cal F} - \hat{E}_{\cal F} \right]^2}  \\
			&+ \max_{i,j} ({\boldsymbol x}_i'{\boldsymbol\beta})^2 ({\boldsymbol x}_j'{\boldsymbol\beta})^2 \tfrac{7\sum\nolimits_{i\neq k}^n  \!\left( \sum\nolimits_{j \neq i,k} b_{i jk} {\boldsymbol x}_j'{\boldsymbol\beta} M_{jj}/D_{ijk} \right)^2}{{\V}_{0}\!\left[{\cal F} - \hat{E}_{\cal F} \right]^2}
			=  o_p\!\left( 1  \right)\!.
		\end{align}

		For the second part of \eqref{eq:comp4}, we have the following variance expression and bound
		\begin{align}
			\V\!\!\left[\!\sum\nolimits_{i\neq j \neq k \neq \ell}^n\! b_{ijk\ell}{\boldsymbol x}_k'{\boldsymbol\beta} y_i \varepsilon_j \varepsilon_\ell \right]\!
			&=  \sum\nolimits_{i\neq j \neq \ell}^n  b_{ij \cdot \ell}\!\left(  b_{ij \cdot \ell} +  b_{i\ell \cdot j}\right)\! \E\!\left[y_i^2 \varepsilon_j^2 \varepsilon_\ell^2 \right]\! \\
			&+ \sum\nolimits_{i\neq j \neq \ell}^n\!\left(  b_{ij\cdot\ell}\!\left( b_{ji\cdot\ell} +  b_{j\ell \cdot i}\right)\! +  b_{i\ell\cdot j}\!\left( b_{ji\cdot\ell} +  b_{j\ell \cdot i}\right)\!\right)\! \E\!\left[\varepsilon_i^2 \varepsilon_j^2 \varepsilon_\ell^2 \right]\! \\
			&+ \sum\nolimits_{i\neq j \neq k \neq \ell}^n  b_{ij\cdot \ell}\!\left(   b_{k j \cdot \ell} +  b_{k\ell \cdot j}\right)\! \E\!\left[y_i y_k \varepsilon_j^2 \varepsilon_\ell^2 \right]\! \\
			& \le  \max_{i,j} \E\!\left[y_i^2 \right]\! \sigma_j^4 8\!\sum\nolimits_{i\neq j \neq \ell}^n \! b_{ij \cdot \ell}^2
			+ \max_{j} \sigma_j^4 2  \sum\nolimits_{j \neq  \ell}^n \!\left( \sum\nolimits_{i \neq j} \! b_{ij \cdot \ell} {\boldsymbol x}_i'{\boldsymbol\beta} \right)^2
		\end{align}
		and we have used that $b_{i j k i} = b_{i j k}$. For this variance bound we utilize \eqref{eq:sumcheckM3} to obtain
		\begin{align}
		\sum\nolimits_{\ell = 1}^n \check M_{i\ell,-ijk} \check M_{\imath \ell,-\imath j \kappa}
		&= \tfrac{\check M_{i \imath,-ijk} - \check M_{i \kappa,-i j k} \check M_{\kappa \imath,-j \kappa}  }{D_{\imath j \kappa}/D_{j \kappa}}, \\
		\sum\nolimits_{\ell = 1}^n \check M_{i\ell,-ijk} \check M_{i\ell,-i j \kappa}
		&= \tfrac{1 - \check M_{i \kappa,-i j k} \check M_{\kappa i,-j \kappa}  }{D_{i j \kappa}/D_{\jmath \kappa}}.
		\end{align}
		When combined with Cauchy-Schwarz and $\boldsymbol M$ having its eigenvalues in $\{0,1\}$, we therefore obtain the further bounds
		\begin{align}
		\tfrac{\sum\nolimits_{i\neq j \neq \ell}^n \!  b_{i j\cdot \ell}^2}{{\V}_{0}\!\left[{\cal F} - \hat{E}_{\cal F} \right]^2}
		&\le \tfrac{\sum\nolimits_{i\neq j}^n \!  b_{i j\cdot}^2}{{\V}_{0}\!\left[{\cal F} - \hat{E}_{\cal F} \right]^2}
		+ \max_{i \neq j \neq  k \neq i}  \tfrac{({\boldsymbol x}_i'{\boldsymbol\beta})^2}{D_{ijk}^2} \tfrac{5 \sum\nolimits_{i\neq j \neq k}^n \!  b_{i jk}^2}{{\V}_{0}\!\left[{\cal F} - \hat{E}_{\cal F} \right]^2}
		+ \tfrac{3\sum\nolimits_{i\neq j}^n \!\left( \sum\nolimits_{k \neq i,j} \! b_{i jk} {\boldsymbol x}_k'{\boldsymbol\beta} M_{kk}/D_{ijk} \right)^2}{{\V}_{0}\!\left[{\cal F} - \hat{E}_{\cal F} \right]^2}
		=  o_p\!\left( 1 \right)\!
		\end{align}
		and
		\begin{align}
		\tfrac{\sum\nolimits_{j \neq  \ell}^n \!\left( \sum\nolimits_{i \neq j} b_{ij \cdot \ell} {\boldsymbol x}_i'{\boldsymbol\beta} \right)^2}{{\V}_{0}\!\left[{\cal F} - \hat{E}_{\cal F} \right]^2}
		&\le \max_{i \neq j \neq  k \neq i} \tfrac{({\boldsymbol x}_i'{\boldsymbol\beta})^2 ({\boldsymbol x}_j'{\boldsymbol\beta})^2}{D_{ijk}^2} \tfrac{11 \sum\nolimits_{i\neq j \neq k}^n  b_{i jk}^2}{{\V}_{0}\!\left[{\cal F} - \hat{E}_{\cal F} \right]^2}   \\
		&+ \max_{i} ({\boldsymbol x}_i'{\boldsymbol\beta})^2 \tfrac{7\sum\nolimits_{i\neq j}^n  \!\left( \sum\nolimits_{k \neq i,j} b_{i jk} {\boldsymbol x}_k'{\boldsymbol\beta} M_{kk}/D_{ijk} \right)^2}{{\V}_{0}\!\left[{\cal F} - \hat{E}_{\cal F} \right]^2}
		=  o_p\!\left( 1 \right)\!
		\end{align}
		where the order statement follows from arguments given for the third component of \eqref{eq:Varsum} and the first part of \eqref{eq:comp4}.

		The variance of the third part of \eqref{eq:comp4} satisfies
		\begin{align}
			\V\!\!\left[\!\sum\nolimits_{i\neq j \neq k \neq \ell}^n\! b_{ijk\ell}{\boldsymbol x}_j'{\boldsymbol\beta} y_i \varepsilon_k \varepsilon_\ell \right]\! &=  \sum\nolimits_{i\neq k \neq \ell}^n b_{i \cdot k \ell}\!\left(  b_{i \cdot k \ell} +  b_{i\cdot \ell  k}\right)\! \E\!\left[y_i^2 \varepsilon_k^2 \varepsilon_\ell^2 \right]\! \\
			&+ \sum\nolimits_{i\neq k \neq \ell}^n\!\left( b_{i\cdot k\ell}\!\left( b_{k\cdot i\ell} +  b_{k\cdot\ell  i}\right)\! +  b_{i\cdot\ell k}\!\left( b_{k\cdot i\ell} +  b_{k\cdot\ell  i}\right)\!\right)\! \E\!\left[\varepsilon_i^2 \varepsilon_k^2 \varepsilon_\ell^2 \right]\! \\
			&+ \sum\nolimits_{i\neq j \neq k \neq \ell}^n  b_{i\cdot k \ell}\!\left(   b_{j \cdot k  \ell} +  b_{j\cdot\ell  k}\right)\! \E\!\left[y_i y_j \varepsilon_k^2 \varepsilon_\ell^2 \right]\! \\
			& \le  \max_{i,j} \E\!\left[y_i^2 \right]\! \sigma_j^4 8\!\sum\nolimits_{i\neq k \neq \ell}^n \! b_{i \cdot k \ell}^2
			+ \max_{j} \sigma_j^4 2 \! \sum\nolimits_{k \neq \ell}^n \!\left( \sum\nolimits_{i \neq k} \! b_{i \cdot k \ell} {\boldsymbol x}_i'{\boldsymbol\beta} \right)^2,
		\end{align}
		where $b_{i \cdot k \ell} = \sum_{j \neq i,k} b_{ijk\ell}{\boldsymbol x}_j'{\boldsymbol\beta}$. In complete analogy with the preceding argument, we use \eqref{eq:sumcheckM3} to obtain
		\begin{align}
			\sum\nolimits_{\ell = 1}^n \check M_{i\ell,-ijk} \check M_{\imath \ell,-\imath \jmath k}
			&= \tfrac{\check M_{i \imath,-ijk} - \check M_{i \jmath,-i j k} \check M_{\jmath \imath,-\jmath k }}{D_{\imath \jmath k}/D_{\jmath k}}, \\
			\sum\nolimits_{\ell = 1}^n \check M_{i\ell,-ijk} \check M_{i\ell,-i \jmath k}
			&= \tfrac{1 - \check M_{i \jmath,-i j k} \check M_{\jmath i,-\jmath k}  }{D_{i \jmath k}/D_{\jmath k}},
		\end{align}
		which leads to
		\begin{align}
			\tfrac{\sum\nolimits_{i\neq k \neq \ell}^n  b_{i \cdot k \ell}^2}{{\V}_{0}\!\left[{\cal F} - \hat{E}_{\cal F} \right]^2}
			&\le \tfrac{\sum\nolimits_{i\neq k}^n  b_{i\cdot k}^2}{{\V}_{0}\!\left[{\cal F} - \hat{E}_{\cal F} \right]^2} + \max_{i \neq j \neq k \neq i} \tfrac{({\boldsymbol x}_i'{\boldsymbol\beta})^2}{D_{ijk}^2} \tfrac{5\sum\nolimits_{i\neq k \neq j}^n  b_{ijk}^2}{{\V}_{0}\!\left[{\cal F} - \hat{E}_{\cal F} \right]^2}
			+ \tfrac{3\sum\nolimits_{i\neq k}^n \!\left( \sum\nolimits_{j \neq i,k} \! b_{i jk} {\boldsymbol x}_j'{\boldsymbol\beta} M_{jj}/D_{ijk} \right)^2}{{\V}_{0}\!\left[{\cal F} - \hat{E}_{\cal F} \right]^2}
			=  o_p\!\left( 1 \right)\!
		\end{align}
		and
		\begin{align}
		\tfrac{\sum\nolimits_{k \neq \ell}^n \!\left( \sum\nolimits_{i \neq k} b_{i \cdot k \ell} {\boldsymbol x}_i'{\boldsymbol\beta} \right)^2 }{{\V}_{0}\!\left[{\cal F} - \hat{E}_{\cal F} \right]^2}
		&\le \max_{i \neq j \neq  k \neq i} \tfrac{({\boldsymbol x}_i'{\boldsymbol\beta})^2 ({\boldsymbol x}_j'{\boldsymbol\beta})^2}{D_{ijk}^2} \tfrac{11 \sum\nolimits_{i\neq j \neq k}^n  b_{i jk}^2}{{\V}_{0}\!\left[{\cal F} - \hat{E}_{\cal F} \right]^2}   \\
		&+ \max_{i} ({\boldsymbol x}_i'{\boldsymbol\beta})^2 \tfrac{7\sum\nolimits_{i\neq k}^n  \!\left( \sum\nolimits_{j \neq i,k} b_{i jk} {\boldsymbol x}_j'{\boldsymbol\beta} M_{jj}/D_{ijk} \right)^2}{{\V}_{0}\!\left[{\cal F} - \hat{E}_{\cal F} \right]^2}
		=  o_p\!\left( 1 \right)\!
		\end{align}
		where the order statement follows from arguments given for the third component of \eqref{eq:Varsum} and the first part of \eqref{eq:comp4}.

		Now the fourth term of \eqref{eq:comp4} satisfies that
		\begin{align}
			\V\!\left[\!\sum\nolimits_{i\neq j \neq k \neq \ell}^n\! b_{ijk\ell}{\boldsymbol x}_i'{\boldsymbol\beta} \varepsilon_j \varepsilon_k \varepsilon_\ell \right]\! &=  \sum\nolimits_{j\neq k \neq \ell}^n b_{\cdot j k \ell}\!\left(  b_{\cdot j k \ell} +  b_{\cdot j \ell  k}\right)\! \E\!\left[\varepsilon_j^2 \varepsilon_k^2 \varepsilon_\ell^2 \right]\! \\
			&+ \sum\nolimits_{i\neq k \neq \ell}^n\!\left( b_{\cdot j k\ell}\!\left( b_{\cdot k j\ell} +  b_{\cdot k \ell  j}\right)\! +  b_{\cdot j \ell k}\!\left( b_{\cdot k j\ell} +  b_{\cdot k \ell  j}\right)\!\right)\! \E\!\left[\varepsilon_j^2 \varepsilon_k^2 \varepsilon_\ell^2 \right]\! \\
			&\le \max_{i} \sigma_i^6 6 \! \left(\sum\nolimits_{j\neq k \neq i}^n b_{i j k}^2+ \sum\nolimits_{j\neq k \neq \ell}^n b_{\cdot j k \ell}^2\right)\!
		\end{align}
		for $b_{\cdot j k \ell} = \sum_{i \neq j,k} b_{ijk\ell}{\boldsymbol x}_i'{\boldsymbol\beta}$ and we have used that $b_{ijk}=b_{ijki}$. The first term was dealt with in \eqref{eq:comp3part3}. For the second term we use a special case of \eqref{eq:sumcheckM3}
		\begin{align}
		 \sum\nolimits_{\ell = 1}^n \check M_{i\ell,-ijk}\check M_{\imath\ell,-\imath jk} 		
		 &= \tfrac{\check M_{\imath i,-\imath jk}}{D_{ijk}/D_{jk}}  =\tfrac{\check M_{ i\imath,-i jk}}{D_{\imath jk}/D_{jk}}  \\
		 &= \tfrac{D_{jk}\left(M_{i\imath}D_{jk} - \left(M_{jj} M_{ik}M_{\imath k} + M_{kk} M_{ij}M_{\imath j} - M_{jk} (M_{ij}M_{\imath k}+M_{ik}M_{\imath j})\right) \right)}{D_{ijk} D_{\imath jk}}
		\end{align}
		so that it follows from the largest eigenvalue of $M$ being one that
		\begin{align}
			\tfrac{\sum\nolimits_{j\neq k \neq \ell}^n b_{\cdot j k \ell}^2}{{\V}_{0}\!\left[{\cal F} - \hat{E}_{\cal F} \right]^2}
			&\le \max_{i\neq j \neq k \neq i} \tfrac{({\boldsymbol x}_i'{\boldsymbol\beta})^2}{D_{ijk}^2}\tfrac{\sum\nolimits_{j \neq k \neq i}^n b_{i j k}^2}{{\V}_{0}\!\left[{\cal F} - \hat{E}_{\cal F} \right]^2}
			=  o_p\!\left( 1 \right)\!.
		\end{align}
		
		Finally, the variance of the fifth term of \eqref{eq:comp4} satisfies the bound
		\begin{align}		
		\V\!\left[\!\sum\nolimits_{i\neq j \neq k \neq \ell}^n\! b_{ijk\ell} \varepsilon_i \varepsilon_j \varepsilon_k \varepsilon_\ell \right]\! &=  \sum\nolimits_{i\neq j \neq k \neq \ell}^n b_{ijk\ell} \E\!\left[ \varepsilon_i^2 \varepsilon_j^2 \varepsilon_k^2 \varepsilon_\ell^2 \right]\!  \\
		&\times \big( b_{ijk\ell} + b_{ij\ell k} + b_{ikj\ell} + b_{ik\ell j} + b_{i\ell jk} + b_{i\ell kj}  \\
		& \phantom{\times \big(\ } +  b_{jik\ell} + b_{ji\ell k} + b_{jki\ell} + b_{jk\ell i} + b_{j\ell ik} + b_{j\ell ki} \\
		& \phantom{\times \big(\ } +  b_{kj\ell i} + b_{kji\ell} + b_{kij\ell} + b_{ki\ell j} + b_{k\ell ij} + b_{k\ell ji} \\
		& \phantom{\times \big(\ }+  b_{\ell j ik} + b_{\ell j ki} + b_{\ell k ij} + b_{\ell k ji} + b_{\ell i jk} + b_{\ell i kj}
		\big) \\
		& \le \max_{i} \sigma_i^8 24 \sum\nolimits_{i\neq j \neq k \neq \ell}^n b_{ijk\ell}^2.
		\end{align}
		Since a special case of \eqref{eq:sumcheckM3} is $\sum\nolimits_{\ell=1}^n \check M_{i\ell,-ijk}^2 = \frac{1}{D_{ijk}/D_{jk}} \le \frac{1}{D_{ijk}}$, we have that
		\begin{align}
			\tfrac{\sum\nolimits_{i\neq j \neq k \neq \ell}^n b_{ijk\ell}^2}{{\V}_{0}\!\left[{\cal F} - \hat{E}_{\cal F} \right]^2}
			&\le \max_{i \neq j \neq k \neq i} D_{ijk}^{-2} \tfrac{\sum\nolimits_{i\neq j \neq k}^n b_{ijk}^2}{{\V}_{0}\!\left[{\cal F} - \hat{E}_{\cal F} \right]^2}
			=  o_p\!\left( 1 \right)\!
		\end{align}
		which completes our proof that the variance estimator $\hat{\V}_0\!\big[{\cal F} - \hat{E}_{\cal F} \big]\!$ is consistent.
	\end{proof}

	\subsection{Asymptotic power}\label{app:power}

	\begin{proof}[Proof of Theorem \ref{thm:power}]
		First, we define the null vector ${\boldsymbol\beta}_0$ corresponding to the parameter vector ${\boldsymbol\beta}$ under $H_\delta$. That is, we let
		\begin{align}\label{eq:Driftbeta}
			{\boldsymbol\beta}_0 = {\boldsymbol\beta} - \boldsymbol S_{xx}\inverse \boldsymbol R' \!\left( \boldsymbol R\boldsymbol S_{xx}^{-1}\boldsymbol R'\right) ^{-1/2} \cdot {\boldsymbol\delta},
		\end{align}
		where we see that ${\boldsymbol\beta}_0$ satisfies the null, $\boldsymbol R{\boldsymbol\beta}_0 = \boldsymbol q$, since ${\boldsymbol\beta}$ is generated by a local alternative with $\boldsymbol R {\boldsymbol\beta} = \boldsymbol q + \!\left( \boldsymbol R\boldsymbol S_{xx}^{-1}\boldsymbol R'\right) ^{1/2} \boldsymbol\delta$. Furthermore, we denote by $\big[ { \cal F}-\hat{E}_{\cal F} \big]_\delta$ the value of ${ \cal F}-\hat{E}_{\cal F}$ under $H_{\delta},$ and by $\big[ { \cal F}-\hat{E}_{\cal F} \big]_0$ its value when the parameter vector ${\boldsymbol\beta}$ is equal to ${\boldsymbol\beta}_0$.
		
		The value of  ${ \cal F}-\hat{E}_{\cal F}$ under $H_{\delta}$ can be represented in accordance with \eqref{eq:Driftbeta} as
		\begin{align}
		\!\left[ { \cal F}-\hat{E}_{\cal F}\right]_\delta
		&=\!\left( \boldsymbol R\boldsymbol S_{xx}^{-1}\sum_{i=1}^n {\boldsymbol x}_i \varepsilon_i +\boldsymbol R{\boldsymbol\beta}-\boldsymbol q\right)'\!\left( \boldsymbol R\boldsymbol S_{xx}^{-1}\boldsymbol R'\right) ^{-1}
		\!\left(\boldsymbol R\boldsymbol S_{xx}^{-1}\sum_{i=1}^n {\boldsymbol x}_i \varepsilon_i +\boldsymbol R{\boldsymbol\beta}-\boldsymbol q\right)\! \\
		&\phantom{=} -\sum\limits_{i=1}^{n}\tfrac{B_{ii}}{M_{ii}} \!\left({\boldsymbol x}_i'{\boldsymbol\beta}_0 + {\boldsymbol x}_{i}'({\boldsymbol\beta}-{\boldsymbol\beta}_0) + \varepsilon_i\right)\! \sum_{j=1}^n M_{ij} \varepsilon_j \\
		&=\!\left[ \mathcal{F}-\hat{E}_{\cal F} \right] _{0}+ {\boldsymbol\delta}'{\boldsymbol\delta}
		+ 2({\boldsymbol\beta}-{\boldsymbol\beta}_0)' \sum_{i=1}^n {\boldsymbol x}_i \varepsilon_i
		- \sum_{j=1}^n \check {\boldsymbol x}_j'({\boldsymbol\beta}-{\boldsymbol\beta}_0) \varepsilon _{j}.
		\end{align}
		
		Note that the third and fourth terms in this representation of $\big[ \mathcal{F}-\hat{E}_{\cal F}\big] _{\delta }$ are both of smaller order than the sum of the first two when $r \rightarrow \infty$. Indeed, both have conditional mean zero, while the third term has variance
		{\small
		\begin{align}
		\V\!\left[ ({\boldsymbol\beta}-{\boldsymbol\beta}_0)'\sum\nolimits_{i=1}^n {\boldsymbol x}_i \varepsilon_i \right]\!
		 &= ({\boldsymbol\beta}-{\boldsymbol\beta}_0)' \sum\nolimits_{i=1}^{n}{\boldsymbol x}_{i}{\boldsymbol x}_{i}'\sigma_{i}^{2} ({\boldsymbol\beta}-{\boldsymbol\beta}_0)
		\leq \max_{i}\sigma _{i}^{2} {\boldsymbol\delta}'{\boldsymbol\delta},
		\end{align}
		}%
		which implies that the third term is $o_p\!\left(\norm{\boldsymbol\delta}\right)\!$ and therefore either of a smaller magnitude than the first term in $\big[ \mathcal{F}-\hat{E}_{\cal F}\big] _{\delta }$ if ${\boldsymbol\delta}'{\boldsymbol\delta}$ is bounded, or of a smaller magnitude than the second if ${\boldsymbol\delta}'{\boldsymbol\delta} \rightarrow \infty$. For the fourth term we similarly have that
		\begin{align}
		\V\!\left[ \sum\nolimits_{j=1}^n \check {\boldsymbol x}_j'({\boldsymbol\beta}-{\boldsymbol\beta}_0) \varepsilon _{j}\right]
		&=\sum\nolimits_{j=1}^{n}\left(\check {\boldsymbol x}_j'({\boldsymbol\beta}-{\boldsymbol\beta}_0) \right) ^{2}\sigma _{j}^{2}
		\le \max_i \sigma_i^2 \sum\nolimits_{j=1}^{n}\left(\check {\boldsymbol x}_j'({\boldsymbol\beta}-{\boldsymbol\beta}_0) \right) ^{2} \\
		&= \max_i \sigma_i^2 \sum\nolimits_{i=1}^n \sum\nolimits_{j=1}^n \tfrac{B_{ii}}{M_{ii}}{\boldsymbol x}_i'({\boldsymbol\beta}-{\boldsymbol\beta}_0) \tfrac{B_{jj}}{M_{jj}}{\boldsymbol x}_j'({\boldsymbol\beta}-{\boldsymbol\beta}_0) M_{ij}  \\
		&\le \max_i \sigma_i^2 {\sum\nolimits_{i=1}^n \tfrac{B_{ii}^2}{M_{ii}^2} ({\boldsymbol x}_i'({\boldsymbol\beta}-{\boldsymbol\beta}_0))^2 }
		\le  \max_{i,j} \tfrac{ \sigma_i^2}{M_{jj}^2} {\boldsymbol\delta}'{\boldsymbol\delta},
		\end{align}
		and since { Assumptions~\ref{ass:iid} and \ref{ass:asympLeave3Out}} imply that $\max_{i,j} \tfrac{ \sigma_i^2}{M_{jj}^2} = O_p(1)$ the argument applied to the third term applies here as well.

		Hence, we have that
		\begin{align}
		\tfrac{\left[ \mathcal{F}-\hat{E}_{\cal F} \right]_{\delta }}{\V_{0}\left[ \mathcal{F} - \hat{E}_{\cal F} \right]^{1/2}} - \tfrac{{\boldsymbol\delta}'{\boldsymbol\delta}}{\sqrt{r}} \!\left(\tfrac{\V_{0}\left[ \mathcal{F}-\hat{E}_{\cal F}\right]}{r} \right)^{-1/2}
		=\tfrac{\left[ \mathcal{F}-\hat{E}_{\cal F} \right] _{0}}{\V_{0}\left[ \mathcal{F}-\hat{E}_{\cal F}\right]^{1/2}}  +o_{p}\left( 1\right) \overset{d}{%
		\rightarrow } {\cal N}\left( 0 ,1\right),
		\end{align}
		and since $\frac{\V_{0}[ \mathcal{F}-\hat{E}_{\cal F}]}{r}$ is bounded and bounded away from zero, it also follows that
		\begin{align}
			\left(\tfrac{{\boldsymbol\delta}'{\boldsymbol\delta}}{\sqrt{r}} - \Delta_\delta^2\right)\!\!\left(\tfrac{\V_{0}\left[ \mathcal{F}-\hat{E}_{\cal F}\right]}{r} \right)^{-1/2} &= o_p(1),  &\quad \text{if } \Delta_\delta < \infty, \\
			\tfrac{\left[ \mathcal{F}-\hat{E}_{\cal F} \right]_{\delta }}{\V_{0}\left[ \mathcal{F} - \hat{E}_{\cal F} \right]^{1/2}} &\xrightarrow{p} \infty, &\quad \text{if } \Delta_\delta = \infty,
		\end{align}
		from which the statement of the theorem follows if $\hat{V}_{\cal F}/{\V}_{0}\!\big[ \mathcal{F}-\hat{E}_{\cal F}	 \big] \xrightarrow{p} 1$.
		
		Next we argue that the variance estimator remains consistent under the sequence of alternatives characterized by $H_{\delta}$. A recall of the argumentation in Appendix \ref{app:varest} reveals that $\hat{V}_{\cal F}$ is an unbiased estimator of ${\V}_{0}\!\big[ \mathcal{F}-\hat{E}_{\cal F} \big] $ for any value of ${\boldsymbol\beta}$. Similarly, an inspection of the proof of Theorem~\ref{thm:size} reveals that the variance bounds derived for components of $\hat{V}_{\cal F}$ do not depend on the particular value for ${\boldsymbol\beta}$. Thus, it suffices that the sequence of local alternatives satisfy {Assumption~\ref{ass:reg}($i$)} as assumed.
	\end{proof}

	Finally, we substantiate the comparison between the power of our proposed LO test and the exact F test provided in Remark~\ref{rem:power}. Specifically, we show that the ratio
	\begin{align}
		\frac{\frac{1}{r}\V_{0}\big[ \mathcal{F}-\hat{E}_{\cal F}\big]}{2\sigma^4 + \frac{2r}{n-m}\sigma^4}
	\end{align}
	converges in probability to unity under either of the following two conditions;
	\begin{enumerate}
		\item $\frac{1}{r}\sum_{i=1}^n B_{ii}^2 = o_p(1)$ which effectively covers settings with $\frac{r}{n} \rightarrow 0$.
		
		\item $\frac{1}{r}\sum_{i=1}^n \!\left(\tfrac{B_{ii}}{M_{ii}} - \mu_n \right)^2 = o_p(1)$ for $\mu_n=\frac{1}{n}\sum_{i=1}^n \tfrac{B_{ii}}{M_{ii}}$ which corresponds to settings with approximately balanced $\frac{B_{ii}}{M_{ii}}$ across observations.
	\end{enumerate}

	To see why that convergence holds, note first that
	\begin{align}
		\frac{\frac{1}{r}\V_{0}\big[ \mathcal{F}-\hat{E}_{\cal F}\big]}{2\sigma^4 + \frac{2r}{n-m}\sigma^4} = \frac{1}{r\!\left(1 + \frac{r}{n-m}\right)\!}\!\left( \text{trace}(\boldsymbol C^2) + \frac{1}{2\sigma^2} \sum\nolimits_{i=1}^n \left( \sum\nolimits_{j\neq i} V_{ij} {\boldsymbol x}_j'{\boldsymbol\beta}\right)^2\right).
	\end{align}
	The second component is negligible under either 1. or 2., as
	\begin{align}
		\tfrac{1}{r}\sum\nolimits_{i=1}^n \!\left(\sum\nolimits_{j \neq i} V_{ij} {\boldsymbol x}_j'{\boldsymbol\beta}\right)^2 &\le \min \!\left\{ \max_{i} \tfrac{({\boldsymbol x}_i'{\boldsymbol\beta})^2 }{M_{ii}^2} \tfrac{1}{r}\sum\nolimits_{i=1}^n B_{ii}^2, \ \max_{i} ({\boldsymbol x}_i'{\boldsymbol\beta})^2 \tfrac{1}{r}\sum\nolimits_{i=1}^n \!\left(\tfrac{B_{ii}}{M_{ii}} - \bar{\tfrac{\boldsymbol B}{\boldsymbol M}} \right)^2\right\}\!.
	\end{align}
	For the first component we use $\boldsymbol C =\boldsymbol B - \frac{1}{2}\!\left(\boldsymbol D_{B \oslash M}\boldsymbol M +\boldsymbol M\boldsymbol D_{B \oslash M} \right)\!$, which leads to
	\begin{align}
		\text{trace}(\boldsymbol C^2) &= r + \tfrac{1}{2}\sum\nolimits_{i=1}^n \tfrac{B_{ii}^2}{M_{ii}} + \tfrac{1}{2}\sum\nolimits_{i=1}^n \sum\nolimits_{j=1}^n \tfrac{B_{ii}}{M_{ii}} \tfrac{B_{jj}}{M_{jj}}  M_{ij}^2 \\
		&= r(1 + \mu_n) + \tfrac{1}{2}\sum\nolimits_{i=1}^n B_{ii}\!\left(\tfrac{B_{ii}}{M_{ii}} - \mu_n\right)\! + \tfrac{1}{2}\sum\nolimits_{i=1}^n \sum\nolimits_{j=1}^n \!\left(\tfrac{B_{ii}}{M_{ii}} - \mu_n\right)\! \tfrac{B_{jj}}{M_{jj}}  M_{ij}^2.
	\end{align}
	To see that the last two terms in this expression for $\text{trace}(\boldsymbol C^2)$ are $o_p(r)$ under either 1. or 2., we note that
	\begin{align}
		\abs*{\tfrac{1}{r}\sum\nolimits_{i=1}^n B_{ii}\!\left(\tfrac{B_{ii}}{M_{ii}} - \mu_n\right)\!} &\le \!\left( \tfrac{1}{r}\sum\nolimits_{i=1}^n B_{ii}^2 \cdot \tfrac{1}{r}\sum\nolimits_{i=1}^n \!\left(\tfrac{B_{ii}}{M_{ii}} - \mu_n\right)^2  \right)^{1/2} = o_p(1)\\
		\abs*{\tfrac{1}{r}\sum\nolimits_{i=1}^n \sum\nolimits_{j=1}^n \!\left(\tfrac{B_{ii}}{M_{ii}} - \mu_n\right)\! \tfrac{B_{jj}}{M_{jj}}  M_{ij}^2} &\le \max_i \tfrac{1}{M_{ii}} \!\left( \tfrac{1}{r}\sum\nolimits_{i=1}^n B_{ii}^2 \cdot \tfrac{1}{r}\sum\nolimits_{i=1}^n \!\left(\tfrac{B_{ii}}{M_{ii}} - \mu_n\right)^2  \right)^{1/2} \\&= o_p(1).
	\end{align}
	Thus the claim of Remark~\ref{rem:power} follows if $\frac{\frac{r}{n-m} - \mu_n}{1+\frac{r}{n-m}}= o_p(1)$ which follows from
	\begin{align}
		\tfrac{\!\left(\tfrac{r}{n-m} - \mu_n\right)^2}{\left(1+\frac{r}{n-m}\right)^2} &= \tfrac{1}{\left(1+\frac{r}{n-m}\right)^2}\!\left( \sum\nolimits_{i=1}^n \tfrac{B_{ii}}{\sum\nolimits_{i=1}^n M_{ii}} - \mu_n\right)^2 = \tfrac{1}{\left(1+\frac{r}{n-m}\right)^2}\!\left( \sum\nolimits_{i=1}^n \tfrac{M_{ii}}{\sum\nolimits_{i=1}^n M_{ii}} \!\left( \tfrac{B_{ii}}{M_{ii}} - \mu_n \right)\! \right)^2 \\
		&\le  \tfrac{\tfrac{r}{n-m}}{\left(1+\frac{r}{n-m}\right)^2}  \tfrac{1}{r}\sum\nolimits_{i=1}^n \!\left( \tfrac{B_{ii}}{M_{ii}} - \mu_n \right)^2.
	\end{align}
	The last expression is $o_p(1)$ under either 1. or 2., as Assumption~\ref{ass:asympLeave3Out} implies that $\frac{n}{n-m} = O(1)$ and 1. implies that $\frac{r}{n} = o(1)$ so that $\tfrac{r}{n-m} = o(1)$ under 1.
	
	\section{If leave-three-out fails}\label{app:cons}

	\subsection{Variance estimator}\label{app:adjusted}
	
	Here we show that if $i$ does not cause $D_{ijk}=0$, i.e., if $D_{ij} D_{ik} >0$ and $D_{jk}=0$, then $\hat \sigma_{i,-j}^2 = \bar \sigma_{i,-k}^2$ and $\hat \sigma_{i,-j}^2$ is independent of both $y_j$ and $y_k$.
	
	When $D_{jk}=0$, the leverage of observation $k$ after leaving $j$ out is one. This, in particular, implies that the weight given to observation $k$ in $-{\boldsymbol x}_i'\hat {\boldsymbol\beta}_{-j}$ is zero, i.e.,
	\begin{align}\label{eq:Djk0}
	0 = - {\boldsymbol x}_i'\left(\sum\nolimits_{\ell \neq j} {\boldsymbol x}_\ell {\boldsymbol x}_\ell'\right)\inverse {\boldsymbol x}_k = M_{ik} - \frac{M_{ij} M_{jk}}{M_{jj}} = \frac{D_{ij}}{M_{jj}} \check M_{ik,-ij}
	\end{align}
	where the second equality follows from \eqref{eq:SMW}. Thus $\check M_{ik,-ij}$, the weight given to $y_k$ in $\hat \sigma_{i,-j}^2$, is zero since $D_{ij} > 0$. Hence $\hat \sigma_{i,-j}^2$ is independent of both $y_j$ and $y_k$.
	
	We now have that $\hat \sigma_{i,-j}^2 = \bar \sigma_{i,-k}^2$ if $ \check M_{i\ell,-ij} = \check M_{i\ell,-ik}$ for all $\ell$ different from $j$ and $k$. Note that under $D_{jk}=0$, equation \eqref{eq:Djk0} shows that $M_{jj} M_{ik} - M_{ij} M_{jk}=0$, and reversing the roles of $j$ and $k$ also leads to $M_{kk} M_{ij} - M_{ik} M_{kj}=0$. By rearranging terms we then obtain
	\begin{align}
	M_{jj} D_{ik} \! =\! M_{jj} D_{ik} \! -\! M_{ij}(M_{kk} M_{ij} \! -\! M_{ik} M_{jk}) \! =\! M_{kk} D_{ij} \! -\! M_{ik}(M_{jj} M_{ik} \! -\! M_{ij} M_{jk}) \! =\! M_{kk} D_{ij}.
	\end{align}
	which implies that $D_{ij} >0$ if and only if $D_{ik} >0$. Since \eqref{eq:Djk0} also applies when $i$ is replaced by $\ell$, we have that $M_{jj} M_{\ell k} - M_{\ell j} M_{jk}=0$ which in turn implies that
	\begin{align}
	M_{kk} M_{ij} M_{j\ell} &= M_{ik} M_{kj} M_{j\ell} = M_{ik} M_{jj} M_{k\ell}.
	\end{align}
	From the two previous highlighted equations it follows that $ \check M_{i\ell,-ij} = \check M_{i\ell,-ik}$ since
	\begin{align}
	M_{kk} D_{ij} \check M_{i\ell,-ij} \!=\! M_{kk}M_{jj} M_{i\ell} - M_{kk}M_{ij} M_{j\ell} \!=\! M_{jj}M_{kk} M_{i\ell} - M_{jj}M_{ik} M_{k\ell} \!=\! M_{jj}D_{ik} \check M_{i\ell,-ik}.
	\end{align}
	
	Finally, we clarify that the first line in the definition of $\overline{\sigma_i^2 \sigma_j^2}$ correspond to the case where none of the leave-three-out failures are caused by both $i$ and $j$. This statement is
	\begin{align}
		D_{ijk} >0 \text{ or } (D_{ij} D_{ik} > 0 \text{ and } D_{jk}= 0) \text{ or } (D_{ij} D_{jk} > 0 \text{ and } D_{ik}= 0) \text{ for all } k
	\end{align}
	and since $D_{ij} >0$ if and only if $D_{ik} >0$ when $D_{jk} =0$ this statement is equivalent to
	\begin{align}
		D_{ijk} >0 \text{ or } (D_{ij} > 0 \text{ and } D_{jk}= 0) \text{ or } (D_{ij}  > 0 \text{ and } D_{ik}= 0) \text{ for all } k
	\end{align}
	which is easily seen to be equivalent to
	\begin{align}
		D_{ij} >0 \text{ and } (D_{ijk} > 0 \text{ or } D_{ik} D_{jk}= 0 \text{ for all } k).
	\end{align}

	\subsection{Asymptotic size}\label{app:adjsize}

	\begin{proof}[Proof of Theorem~\ref{thm:sizeLeave1Out}]

		It suffices to show that $\hat{V}_{\cal F}$ is a non-negatively biased estimator of the relevant target ${\V}_0\big[{\cal F} - \hat{E}_{\cal F} \big]$, that this adjusted variance estimator concentrates around its expectation, and that $\lim \inf_{n \rightarrow \infty} q_{1-\alpha}(\bar F_{\hat{\boldsymbol w},n-m}) \ge 1$ in probability when $r$ is fixed. To establish the required properties regarding $\hat{V}_{\cal F}$, it is useful to let $H_{i,jk}$ and $H_{ij}$ be indicators for presence of bias in the error variance estimators, i.e., let
		\begin{align}
		H_{i,jk} &=  \mathbf{1}\left\{ D_{ij}D_{ik}=0 \text{ or } (D_{ijk} =0, \ D_{jk} > 0) \right\}, \\
		H_{ij} &= \mathbf{1}\left\{ D_{ij} = 0 \text{ or } \exists k : D_{ijk} =0, \ D_{ik} D_{jk} > 0  \right\}
		\end{align}
		We can then write $\hat{V}_{\mathcal{F}}= A_1 + A_2  + B_1 + B_2$ for
		\begin{align}
		A_1 &= \sum_{i=1}^{n}\sum_{j\neq i} (1-H_{ij}) \left( U_{ij}-V_{ij}^{2}\right)  \overline{\sigma_i^2 \sigma_j^2}, &
		A_2 &=  \sum_{i=1}^{n}\sum_{j\neq i} H_{ij}\left( U_{ij}-V_{ij}^{2}\right)_+ y_i^2 \bar \sigma_{j,-i}^2, \\
		B_1 &= \sum_{i=1}^{n}\sum_{j\neq i}\sum_{k\neq i} (1-H_{i,jk})V_{ij}V_{ik} y_{j}y_{k}\bar {\sigma}_{i,-jk}^{2}, &
		B_2 &=\sum_{i=1}^{n} \left(\sum_{j\neq i}\sum_{k\neq i}H_{i,jk}V_{ij}V_{ik}y_{j}y_{k} \right)_+ y_{i}^{2}
		\end{align}
		where $(\cdot)_+$ stands for taking a positive part.
		
		First we consider the expectation of $\hat{V}_{\cal F}$. The bias in $\hat{V}_{\mathcal{F}}$ stems from $A_2$ and $B_2$ so we have that
		\begin{align}
			\E\!\left[\hat{V}_{\cal F}\right] - {\V}_0\!\left[{\cal F} - \hat{E}_{\cal F} \right]
			&= \sum_{i=1}^n \sum_{j \neq i} H_{ij}  \!\left( \big( U_{ij} - V_{ij}^2\big)_+ \E[ y_i^2 ]\E[ \hat \sigma_{j,-i}^2 ]- \big( U_{ij} - V_{ij}^2\big) \sigma_i^2 \sigma_j^2 \right)\! \\
			&\quad +  \E[B_2]
			 - \sum_{i=1}^n\sum_{j\neq i} \sum_{k\neq i}  H_{i,jk} V_{ij}   V_{ik} \E\!\left[y_j y_k\right]\! \sigma_i^2,
		\end{align}
		The first component of this bias is non-negative since $\big( U_{ij} - V_{ij}^2\big)_+ \E[ y_i^2 ]\E[ \hat \sigma_{j,-i}^2 ]$ is never smaller than $\big( U_{ij} - V_{ij}^2\big) \sigma_i^2 \sigma_j^2$.  For the second line we use that the mapping $(\cdot)_+$ is convex and larger than its argument. These two properties yield
		\begin{align}
			\E[B_2] & =\sum_{i=1}^{n}\E\!\left[ \!\left(\sum\nolimits_{j\neq i} \sum\nolimits_{k\neq i}  H_{i,jk} V_{ij}  y_j  V_{ik}  y_k  \right)_+ \right]\! \E[ y_i^2 ] \\
			&\ge \sum_{i=1}^{n}\sum_{j\neq i} \sum_{k\neq i} H_{i,jk} V_{ij} V_{ik} \E\!\left[y_j y_k\right]\! \sigma_i^2
			+\sum_{i=1}^{n}\E\!\left[ \!\left(\sum\nolimits_{j\neq i} \sum\nolimits_{k\neq i}  H_{i,jk} V_{ij}  y_j  V_{ik}  y_k  \right)_+ \right]\! ({\boldsymbol x}_i'{\boldsymbol\beta})^2,
		\end{align}
		and since the second part of this lower bound is non-negative, we conclude that the second component of the bias in $\hat{V}_{\cal F}$ is also greater than equal to zero.

		We now show that $\hat{V}_{\cal F}$ concentrates around its expectation, i.e.,
$\big( \hat{V}_{\cal F} - \E\big[\hat{V}_{\cal F} \big]  \big)/\E\big[\hat{V}_{\cal F} \big] \xrightarrow{p} 0$. Since $\E[\hat V_{\cal F}]\inverse \le {\V}_0\big[{\cal F} - \hat{E}_{\cal F} \big]\inverse = O_p(\frac{1}{r})$, it suffices for this conclusion to show that $\hat{V}_{\mathcal{F}}-\E\big[\hat{V}_{\mathcal{F}}\big] = o_p(r)$. Since $A_1$, $A_2$, and $B_2$ are quartic functions of the outcome variables it can be shown that $A_1 - \E[A_1] = o_p(r)$, $A_2 - \E[A_2] = o_p(r)$, and $B_1 - \E[B_1] = o_p(r)$ by the same argumentation as in the proof of Theorem~\ref{thm:size}.

$B_2$ involves additional non-linearities due to the presence of outcome variables inside the positive part function. For this reason we handle this term using \cite{soelvsten2017robust}, Lemma A2.2,  which is a version of the Efron-Stein inequality. Letting $\Delta_\ell B_2 = B_2 - B_{2,-\ell}$ where
\begin{align}
B_{2,-\ell} &= \sum_{i=1}^{n} \left(\sum_{j\neq i}\sum_{k\neq i}H_{i,jk}V_{ij}V_{ik}y_{j,-\ell}y_{k,-\ell} \right)_+ y_{i,-\ell}^{2} & y_{i,-\ell}& = \begin{cases}
y_i, & \text{if } i \neq \ell, \\ {\boldsymbol x}_i'{\boldsymbol\beta}, & \text{if } i = \ell
\end{cases}
\end{align}
it follows from Lemma A2.2 of \cite{soelvsten2017robust} that $B_1 - \E[B_1] = o_p(r)$ provided that $\sum_{\ell = 1}^n \E[(\Delta_\ell B_2)^2] = o_p(r^2)$. That $\sum_{\ell = 1}^n \E[(\Delta_\ell B_2)^2] = o_p(r^2)$ holds can be established following the argumentation in the proof of Theorem~\ref{thm:size} and we therefore omit the details.

To show that $\lim \inf_{n \rightarrow \infty} q_{1-\alpha}(\bar F_{\hat{\boldsymbol w},n-m}) \ge 1$ in probability when $r$ is fixed, we can argue along subsequences where $\hat{\boldsymbol w} \xrightarrow{p} \overrightarrow{\boldsymbol{w}}_{\cal F}$ and simply treat the limit problem of showing that $q_{1-\alpha}(\bar \chi^2_{\overrightarrow{\boldsymbol{w}}_{\cal F}}) \ge 1$ whenever $\alpha \le 0.31$. Now use $G$ to denote the distribution function of a $\chi^2_1$ random variable and note that $G$ is a concave function. Therefore, we have that (without loss of generality let $w_1 >0$)
\begin{align}
\Pr \left( \sum_{\ell = 1}^r w_\ell Z_\ell \le 1 \right) = \E\left[ G\left(\frac{1 - \sum_{\ell = 2}^r w_\ell Z_\ell}{w_1}\right) \right] \le G\left(\E\left[ \frac{1 - \sum_{\ell = 2}^r w_\ell Z_\ell}{w_1} \right] \right) = G(1) < 0.69.
\end{align}
Thus it follows that $q_{1-\alpha}(\bar \chi^2_{\overrightarrow{\boldsymbol{w}}_{\cal F}}) > 1$ for any value of $\overrightarrow{\boldsymbol{w}}_{\cal F}$.
\end{proof}

\subsection{Location invariance}\label{app:small}

	The test statistic considered in the simulation study of Section \ref{sec:sim} relies on demeaned outcomes $\dot y_i = y_i - \frac{1}{n} \sum_{i=1}^n y_i$ when centering and studentizing $\cal F$. This is done to ensure that the critical value is invariant to shifts in location of the outcomes. Specifically, we estimate $\mathbb{E}_0[{\cal F}]$ using $\tilde{E}_{\cal F} =  \sum\nolimits_{i=1}^n B_{ii} \tilde \sigma_i^2$,  where $\tilde \sigma_i^2 = \dot y_i(y_i - {\boldsymbol x}_i'\hat {\boldsymbol\beta}_{-i})$. Furthermore, we let $\tilde {\boldsymbol w} = (\tilde w_1,\dots,\tilde w_n)'$, where $\tilde w_\ell = \frac{\dot w_\ell \vee 0}{\sum_{\ell=1}^r  (\dot w_\ell \vee 0) }$ and $\dot w_\ell$ is the $\ell$-th eigenvalue of $\varOmega \!\left(\tilde \sigma_1^2,\dots,\tilde \sigma_n^2 \right)$.

	The variance estimator similarly relies on demeaned outcomes in its construction. In analogy with the above definition, we let
	\begin{align}
	\tilde \sigma_{i,-jk}^2 =
	\begin{cases}
	\dot y_i(y_i - {\boldsymbol x}_i'\hat {\boldsymbol\beta}_{-ijk}), & \text{if }  D_{ijk}>0, \\
	\dot y_i(y_i - {\boldsymbol x}_i'\hat {\boldsymbol\beta}_{-ij}), & \text{if }  D_{ij} D_{ik}>0 \text{ and } D_{jk}=0, \\
	\dot y_i^2, & \text{otherwise},
	\end{cases}
	\end{align}
	where we also write $\tilde \sigma_{i,-j}^2$ when $j$ is equal to $k$. We use $\tilde \sigma_{i,-jk}^2$ in the construction of the variance product estimator
	\begin{align}
	\widetilde{\sigma_i^2 \sigma_j^2} =
	\begin{cases}
	\dot y_i \sum_{k \neq j} \check M_{ik,-ij} \dot y_k \cdot \tilde \sigma_{j,-ik}^2, & \text{if }  D_{ij}>0 \text{ and } (D_{ijk}>0 \text{ or } D_{ik} D_{jk}=0 \text{ for all } k), \\
	\dot y_i^2 \tilde \sigma_{j,-i}^2, & \text{otherwise}.
	\end{cases}
	\end{align}
	This leads to the variance estimator
	\begin{align}
		\tilde{V}_{\cal F}
		&= \sum_{i=1}^n \sum_{j \neq i} \big( U_{ij} - V_{ij}^2\big) \cdot G_{ij} \cdot \widetilde{\sigma_i^2 \sigma_j^2}
		+ \sum_{i=1}^n \sum_{j\neq i} \sum_{k\neq i}  V_{ij}  \dot y_j \cdot V_{ik}  \dot y_k \cdot G_{i,-jk} \cdot \tilde \sigma_{i,-jk}^2,
	\end{align}
	where the indicators $G_{ij}$ and $\tilde G_{i,-jk}$ remove biased estimators with negative weights:
	\begin{align}
	G_{ij} &= H_{ij} \mathbf{1}\!\left\{ U_{ij} - V_{ij}^2 <0 \right\}\!,
%
	& \tilde G_{i,-jk} &= H_{i,jk} \mathbf{1}\!\left\{ \sum_{j \neq i} \sum_{k \neq i} V_{ij}  \dot y_j \cdot V_{ik}  \dot y_k \cdot H_{i,jk} < 0 \right\}
	\end{align}
	for $H_{ij}$ and $H_{i,jk}$ as introduced in the proof of Theorem~\ref{thm:sizeLeave1Out}.
	
	While $\tilde{V}_{\cal F}$ is positive with probability approaching one in large samples, it may be negative in small samples. When this occurs we instead rely on a variance estimator that estimates all error variances unconditionally. This guarantees positivity of the variance estimator:
	\begin{align}
		\tilde {V}^+_{\cal F}
		&= \!\begin{cases}
		\tilde{V}_{\cal F},& \text{if } \tilde{V}_{\cal F}>0, \\
		\sum\limits_{i=1}^n \sum\limits_{j \neq i} \big( U_{ij} - V_{ij}^2\big)_+ \dot y_i^2 \dot y_j^2
			+ \sum\limits_{i=1}^n \left( \sum_{j\neq i} V_{ij}  \dot y_j \right)^2 \dot y_i^2, & \text{otherwise}.
		\end{cases}
	\end{align}
	The test considered in the simulations rejects when
	\begin{align}
		F >  \frac{1}{r\hat \sigma^2_{\varepsilon}} \left(\tilde{E}_{\cal F}+({\tilde{V}_{\cal F}^+)^{1/2}} \frac{\hat q_{1-\alpha}(\bar F_{\tilde {\boldsymbol w},n-m})-1}{\sqrt{2\sum_{\ell=1}^r \tilde w_\ell^2+2/(n-m)}} \right)\!.
	\end{align}
	and $\hat q_{1-\alpha}(\bar F_{\tilde {\boldsymbol w},n-m})$ is the $(1-\alpha)$-th quantile among $49,999$ independent draws of $\frac{\sum_{\ell=1}^r \tilde w_\ell Z_\ell}{Z_0/(n-m)}.$

\section{Simulation evidence}\label{app:sim}

\subsection{Simulation design}

For $k \in\{2,\dots,m-1\}$, let the $k$-th continuous regressor for observation $i$ have the representation $x_{ik}=\left( \frac{1}{2}+u_{i}\right)\! x_{ik}^{0},$ where $u_{i}\sim \mathrm{IID}\,U\!\left[ 0,1\right]\! $ with $\E\!\left[ \frac{1}{2}+u_{i}\right]\! =1$, $\E[\left( \frac{1}{2}+u_{i}\right) ^{2}]=\frac{13}{12},$ and $x_{ik}^{0}\sim \mathrm{IID}\,LN$ with  $\E\!\left[ x_{ik}^{0}\right]\!=e^{1/2}$ and $\V\!\left[ x_{ik}^{0}\right]\! =e\left( e-1\right) .$ For the mixed design, let additionally $d_{i\ell }$ be the $\ell$-th discrete regressor for observation $i$, where $\ell =1,\dots,r$.

To report $\E\!\left[ \boldsymbol S_{xx}\right] $, we note that for the continuous regressors $\E\!\left[ x_{ik}^{2}\right] =\frac{13}{12}e^{2}$,  $\E\!\left[x_{ik}\right]\! =e^{1/2}$, and $\E\!\left[ x_{ik}x_{ik'}\right]\! =\frac{13}{12}e$ for $1\neq k\neq k' \neq 1$.  For discrete regressors, we have
\begin{align}
s_{dd,\ell }
&\equiv \E\!\left[ d_{i\ell }^{2}\right]\!
=\Pr \left\{ \frac{\ell-1}{r+1}\leq \frac{u_{i}+u_{i}^{2}}{2}<\frac{\ell }{r+1}\right\}
=\sqrt{\frac{1}{4}+2\frac{\ell }{r+1}}-\sqrt{\frac{1}{4}+2\frac{\ell -1}{r+1}},
\end{align}
$\E\!\left[d_{i\ell }\right] =s_{dd,\ell },$ and $\E\!\left[ d_{i\ell}d_{i\ell' }\right]\! =0$ for $1\leq \ell\neq \ell' \leq r$. For the cross-moments between continuous and discrete regressors, we have
\begin{align}
s_{xd,\ell }
&\equiv \E\!\left[ x_{ik}d_{i\ell }\right]
=e^{1/2}\E\!\left[ \left(\frac{1}{2}+u_{i}\right) \mathbf{1}\!\left\{ \frac{\ell -1}{r+1}\leq \frac{u_{i}+u_{i}^{2}}{2}<\frac{\ell }{r+1}\right\}\! \right]
=\frac{e^{1/2}}{r+1}\equiv s_{xd},
\end{align}
which does not depend on $\ell$.

The matrix $\E\!\left[ \boldsymbol S_{xx}\right] $ is therefore structured as follows. In the continuous design,
\begin{align}
\E\!\left[ \boldsymbol S_{xx}\right]\! =n
\begin{bmatrix}
1 & e^{1/2}{\boldsymbol\iota}_{m-1}^{\prime } \\
e^{1/2}\iota _{m-1} & \frac{13}{12}e\cdot {\boldsymbol\iota}_{m-1}{\boldsymbol\iota}_{m-1}^{\prime }+%
\frac{13}{12}e\left( e-1\right)\boldsymbol I_{m-1}%
\end{bmatrix},
\end{align}
while in the mixed design,%
\begin{align}
\E\!\left[ \boldsymbol S_{xx}\right]\! =n
\begin{bmatrix}
1 & e^{1/2}{\boldsymbol\iota}_{m-r-1}^{\prime } & (s_{dd,1},\dots,s_{dd,r}) \\
e^{1/2}{\boldsymbol\iota}_{m-r-1} & \frac{13}{12}e\cdot {\boldsymbol\iota}_{m-r-1}{\boldsymbol\iota}
_{m-r-1}^{\prime }+\frac{13}{12}e\left( e-1\right)\boldsymbol I_{m-r-1} & s_{xd}\cdot
{\boldsymbol\iota}_{m-r-1}{\boldsymbol\iota}_{r}^{\prime } \\
(s_{dd,1},\dots,s_{dd,r})' & s_{xd}\cdot {\boldsymbol\iota}
_{r}{\boldsymbol\iota}_{m-r-1}^{\prime } & \mathrm{diag}\{s_{dd,\ell }\}_{\ell =1}^{r}%
\end{bmatrix}.
\end{align}

The link between the value of the regression coefficients under the null, $\varrho$, and the $\mathrm{R}^{2}$ is
\begin{align}
\varrho =\frac{1}{\sqrt{m-1}}\sqrt{\frac{\mathrm{R}^{2}}{1-\mathrm{R}^{2}}\frac{12}{13e^{2}+\left( m-14\right)\! e}}
\end{align}
for the continuous design. In the mixed design, $m$ is replaced by $m-r$ as all coefficients for the discrete regressors are equal to zero under the null.

\subsection{Simulation results}

\begin{table}[b!]
	\centering
	{\footnotesize
		\caption{Empirical size (in percent)}
		\label{tab:ASize}
		\begin{threeparttable}
			\begin{widetable}{.95\columnwidth}{llcrrrrrcrcrrrrrcrcrrrrrcr}
				\toprule
				\multicolumn{2}{l}{Test} &
				&LO&EF&$\text{W}_\text{1}$&$\text{W}_\text{K}$&$\text{W}_\text{L}$&& \textsc{neg}
				&&LO&EF&$\text{W}_\text{1}$&$\text{W}_\text{K}$&$\text{W}_\text{L}$&& \textsc{neg}
				&&LO&EF&$\text{W}_\text{1}$&$\text{W}_\text{K}$&$\text{W}_\text{L}$&& \textsc{neg} \\
				\midrule
				\multicolumn{2}{l}{Baseline}
				&& \multicolumn{7}{c}{}	
				&& \multicolumn{7}{c}{$\mathrm{R}^2=\frac{1}{6}$, $\frac{r}{m}=\frac{9}{12}$, $\frac{m}{n}=\frac{4}{5}$}
				&& \multicolumn{7}{c}{}  \\
				\cmidrule(r){1-2} \cmidrule(lr){12-18}
				$n=80$     &  && &&&& && &&    $5$ & $47$ &$100$ & $52$ & $15$ && $12.8$   &&  &&&& && \\
				$n=160$   &  &&  &&&& && &&    $5$ & $61$ &$100$ & $50$ & $16$ && $3.7$  &&  &&&& && \\
				$n=320$   & &&  &&&& && &&    $6$ & $81$ &$100$ & $49$ & $18$  && $0.9$  &&  &&&& && \\
				$n=640$   & &&  &&&& && &&    $5$ & $96$ &$100$ & $49$ & $19$  && $0.1$  &&  &&&& && \\
				$n=1280$ & &&  &&&& && &&   $5$ & $100$ & $100$ & $48$ & $19$   && $0.0$     &&  &&&& && \\			 
				\midrule
				&
				&& \multicolumn{7}{c}{$\mathrm{R}^2=\frac{0}{6}$}	
				&& \multicolumn{7}{c}{$\mathrm{R}^2=\frac{2}{6}$}
				&& \multicolumn{7}{c}{$\mathrm{R}^2=\frac{3}{6}$}  \\
				\cmidrule(lr){4-10} \cmidrule(lr){12-18} \cmidrule(lr){20-26}
				$n=80$     &  && 5&47&100&51&16 &&11.5 &&    4&47&100&52&16 &&15.3 &&  3&47&100&51&15 && 18 \\
				$n=160$   &  &&  6&61&100&50&18 &&3.1 &&    5&61&100&50&17 &&4.4  &&  6&61&100&50&17 && 5.3 \\
				$n=320$   & &&  6&81&100&49&19 &&0.7 &&    5&81&100&50&17 &&1.0  &&  6&81&100&50&18 &&1.2 \\
				$n=640$   & &&  5&96&100&48&19 && 0.2 &&    5&96&100&49&18 &&0.2  &&  5&96&100&49&18 &&0.2 \\
				$n=1280$ & &&  5&100&100&49&20 && 0.0 &&   5&100&100&48&19 &&0.0     &&  5&100&100&49&18 &&0.0 \\			
				\midrule
				&
				&& \multicolumn{7}{c}{$\frac{r}{m}=\frac{3}{12}$}	
				&& \multicolumn{7}{c}{$\frac{r}{m}=\frac{5}{12}$}
				&& \multicolumn{7}{c}{$\frac{r}{m}=\frac{7}{12}$}  \\
				\cmidrule(lr){4-10} \cmidrule(lr){12-18} \cmidrule(lr){20-26}
				$n=80$     &  && 5&18&66&34&18 &&9.8 &&    5&26&91&40&16 &&10.6   &&  5&36&99&46&16 &&12.2 \\
				$n=160$   &  &&  6&24&78&34&21 &&2.5 &&    5&34&98&41&18 &&2.7  &&  6&49&100&46&17 &&3.4 \\
				$n=320$   & &&  6&31&90&33&23 &&0.5 &&    5&49&100&42&20 &&0.6  && 7&68&100&46&18 &&0.7  \\
				$n=640$   & &&  5&46&98&35&27 &&0.0 &&    5&72&100&41&22 &&0.2  &&  5&88&100&45&20 &&0.1 \\
				$n=1280$ & &&  5&67&100&35&27 &&0.0  &&   5&92&100&41&23 &&0.0    &&  5&99&100&45&21 && 0.0\\			
				\midrule
				&
				&& \multicolumn{7}{c}{$\frac{m}{n}=\frac{1}{5}$}	
				&& \multicolumn{7}{c}{$\frac{m}{n}=\frac{2}{5}$}
				&& \multicolumn{7}{c}{$\frac{m}{n}=\frac{3}{5}$}  \\
				\cmidrule(lr){4-10} \cmidrule(lr){12-18} \cmidrule(lr){20-26}
				$n=80$     &  && 7&78&81&86&35 &&2.5 &&    6&72&96&76&32 && 3.6   &&  6&64&100&61&25 &&5.3 \\
				$n=160$   &  &&  6&85&87&92&45 &&1.0 &&    6&82&99&76&39 &&1.5  &&  5&77&100&60&26 &&1.8 \\
				$n=320$   & &&  5&91&94&97&57 &&0.5 &&    5&92&100&76&43 &&0.4  &&  5&92&100&60&29 && 0.5 \\
				$n=640$   & &&  5&97&99&99&68 &&0.1 &&    5&99&100&78&46 &&0.1  &&  5&99&100&58&30 && 0.1 \\
				$n=1280$ & &&  5&100&100&100&78 &&0.0 &&   5&100&100&79&48 &&0.0   &&  5&100&100&59&32 && 0.0 \\			
				\bottomrule
			\end{widetable}
			\begin{tablenotes}
				\scriptsize 
				\item \textsc{NOTE}: All size results are for $5\%$ nominal size under the continuous design with heteroskedasticity as described in Section \ref{sec:sim} of the paper. The first panel repeats baseline results from Table \ref{tab:Size} with a coefficient of determination $\mathrm{R}^2$ at $1/6$, a fraction of tested restrictions relative to number of regressors $r/m$ at $9/12$, and a fraction of regressors relative to sample size $m/n$ at $4/5$. The remaining three panels make ceteris paribus deviations by varying either $\mathrm{R}^2$, $r/m$, or $m/n$. LO: leave-out test, EF: exact F test, $\text{W}_\text{1}$: heteroskedastic Wald test with degrees-of-freedom correction, $\text{W}_\text{K}$: heteroskedastic Wald test with \cite{cattaneo2017inference} correction, $\text{W}_\text{L}$: heteroskedastic Wald test with \cite{kline2018leave} correction; \textsc{neg}: fraction of negative variance estimates for LO (in percent). Results from 10000 Monte-Carlo replications.
			\end{tablenotes}
		\end{threeparttable}
	}
\end{table}

Tables A1-A2 contain results from additional simulations where we, starting from the baseline continuous regressor design, for which the results are reported in the main text, vary parameters related to the strength of the signal relative to the noise: the coefficient of determination $\mathrm{R}^{2}$, and the relative numerosities of regressors and restrictions $r/m$ and $m/n$. Table A1 contains figures on empirical size and positivity failure rate of $\hat{V}_{\cal F}$ for the heteroskedastic setup, and Table A2 contains figures on empirical power for the homoskedastic setup with dense deviations from the null.

\begin{table}[bt!]
	\centering
	{\footnotesize
		\caption{Empirical power (in percent) corresponding to 5\% size}
		\label{tab:APower}
		\begin{threeparttable}
			\begin{widetable}{.9\columnwidth}{llcrrcrrcrr}
				\toprule
				Test &  &&  LO & EF  && LO & EF && LO & EF    \\
				\midrule
				\multicolumn{2}{l}{Baseline} \\
				\cmidrule(lr){1-2}
				{$n=80$}  &&  && & & $5$ & $15$ & &&  \\
				{$n=160$}&&  && & &$12$ & $21$ & &&  \\
				{$n=320$}&&  && & &$26$ & $36$ & &&  \\
				{$n=640$}&&  && & &$44$ & $57$ & &&  \\
				{$n=1280$}&&  && & &$68$ & $84$ & &&  \\
				\midrule
				&
				&& \multicolumn{2}{c}{$\mathrm{R}^2=\frac{0}{6}$}	
				&& \multicolumn{2}{c}{$\mathrm{R}^2=\frac{2}{6}$}
				&& \multicolumn{2}{c}{$\mathrm{R}^2=\frac{3}{6}$}  \\
				\cmidrule(lr){4-5} \cmidrule(lr){7-8} \cmidrule(lr){10-11}
				{$n=80$}  &&  &7&15 & &4&15 & &3&15  \\
				{$n=160$}&&  &16&21 & &11&21 & &10&21  \\
				{$n=320$}&&  &29&36 & &23&36 & &21&36  \\
				{$n=640$}&&  &48&57 & &41&57 & &39&57  \\
				{$n=1280$}&&  &73&84 & &65&84 & &63&84  \\
				\midrule
				&
				&& \multicolumn{2}{c}{$\frac{r}{m}=\frac{3}{12}$}	
				&& \multicolumn{2}{c}{$\frac{r}{m}=\frac{5}{12}$}
				&& \multicolumn{2}{c}{$\frac{r}{m}=\frac{7}{12}$}  \\
				\cmidrule(lr){4-5} \cmidrule(lr){7-8} \cmidrule(lr){10-11}
				{$n=80$}  &&  &9&17 & &7&15 & &6&15  \\
				{$n=160$}&&  &19&25 & &16&22 & &15&22  \\
				{$n=320$}&&  &33&39 & &28&36 & &27&35  \\
				{$n=640$}&&  &55&61 & &48&58 & &46&58  \\
				{$n=1280$}&&  &81&87 & &74&84 & &70&83  \\
				\midrule
				&
				&& \multicolumn{2}{c}{$\frac{m}{n}=\frac{1}{5}$}	
				&& \multicolumn{2}{c}{$\frac{m}{n}=\frac{2}{5}$}
				&& \multicolumn{2}{c}{$\frac{m}{n}=\frac{3}{5}$}  \\
				\cmidrule(lr){4-5} \cmidrule(lr){7-8} \cmidrule(lr){10-11}
				{$n=80$}  &&  &37&72 & &23&48 & &14&28  \\
				{$n=160$}&&  &71&93 & &47&71 & &28&45  \\
				{$n=320$}&&  &96&100 & &81&93 & &52&70  \\
				{$n=640$}&&  &100&100 & &98&100 & &81&91  \\
				{$n=1280$}&&  &100&100 & &100&100 & &98&100  \\
				\bottomrule
			\end{widetable}
			\begin{tablenotes}
				\scriptsize
				\item \textsc{NOTE}:  All power results are for $5\%$ nominal size under the continuous design with homoskedasticity and dense deviations as described in Section \ref{sec:sim} of the paper. The first panel repeats baseline results from Table \ref{tab:Power} with a coefficient of determination $\mathrm{R}^2$ at $1/6$ under the null, a fraction of tested restrictions relative to number of regressors $r/m$ at $9/12$, and a fraction of regressors relative to sample size $m/n$ at $4/5$. The remaining three panels make ceteris paribus deviations by varying either $\mathrm{R}^2$, $r/m$, or $m/n$. LO: leave-out test, EF: exact F test. Results from 10000 Monte-Carlo replications.
			\end{tablenotes}
		\end{threeparttable}
	}
\end{table}

\newpage

\end{appendices}

\end{document}